%% file: ApproxCovariantQEC.tex
\documentclass[%
  aps,%
  pra,%
  superscriptaddress,%
  reprint,%
  nofootinbib,%
  longbibliography]{revtex4-1}

\input{preamblesetup.tex}

\newcommand{\CC}{\mathbb C}
\newcommand{\RR}{\mathbb R}
\newcommand{\ZZ}{\mathbb Z}

\newcommand{\EE}{\operatorname{\mathbb E}}

\newcommand{\Hil}{\Hs}%

\newcommand{\ot}{\otimes}

\global\long\def\xr{\overleftarrow{X}}
\global\long\def\xl{\overrightarrow{X}}

\csname input\endcsname{bib_duplicates_aliases.tex}

\begin{document}

\title{Continuous symmetries and approximate quantum error correction}

\date{February 19, 2019}

\author{Philippe Faist}
\thanks{These authors contributed equally to this work.}
\affiliation{Institute for Quantum Information and Matter, Caltech, Pasadena, CA, USA}

\author{Sepehr Nezami}
\thanks{These authors contributed equally to this work.}
\affiliation{Stanford Institute for Theoretical Physics, Stanford University, Stanford, CA, USA}

\author{Victor V. Albert}
\affiliation{Institute for Quantum Information and Matter, Caltech, Pasadena, CA, USA}
\affiliation{Walter Burke Institute for Theoretical Physics, Caltech, Pasadena, CA, USA}

\author{Grant~Salton}
\affiliation{Institute for Quantum Information and Matter, Caltech, Pasadena, CA, USA}
\affiliation{Stanford Institute for Theoretical Physics, Stanford University, Stanford, CA, USA}

\author{Fernando Pastawski}
\affiliation{Dahlem Center for Complex Quantum Systems, Freie Universit{\"a}t Berlin, Berlin, Germany}

\author{Patrick Hayden}
\affiliation{Stanford Institute for Theoretical Physics, Stanford University, Stanford, CA, USA}

\author{John Preskill}
\affiliation{Institute for Quantum Information and Matter, Caltech, Pasadena, CA, USA}
\affiliation{Walter Burke Institute for Theoretical Physics, Caltech, Pasadena, CA, USA}

\begin{abstract}
Quantum error correction and symmetry arise in many areas of physics, including many-body systems, metrology in the presence of noise, fault-tolerant computation, and holographic quantum gravity. Here we study the compatibility of these two important principles. If a logical quantum system is encoded into $n$ physical subsystems, we say that the code is covariant with respect to a symmetry group $G$ if a $G$ transformation on the logical system can be realized by performing transformations on the individual subsystems. For a $G$-covariant code with $G$ a continuous group, we derive a lower bound on the error correction infidelity following erasure of a subsystem. This bound approaches zero when the number of subsystems $n$ or the dimension $d$ of each subsystem is large. We exhibit codes achieving approximately the same scaling of infidelity with $n$ or $d$ as the lower bound. Leveraging tools from representation theory, we prove an approximate version of the Eastin-Knill theorem: If a code admits a universal set of transversal gates and corrects erasure with fixed accuracy, then, for each logical qubit, we need a number of physical qubits per subsystem that is inversely proportional to the error parameter. We construct codes covariant with respect to the full logical unitary group, achieving good accuracy for large $d$ (using random codes) or $n$ (using codes based on $W$-states). We systematically construct codes covariant with respect to general groups, obtaining natural generalizations of qubit codes to, for instance, oscillators and rotors. In the context of the AdS/CFT correspondence, our approach provides insight into how time evolution in the bulk corresponds to time evolution on the boundary without violating the Eastin-Knill theorem, and our five-rotor code can be stacked to form a covariant holographic code.
\end{abstract}

\maketitle

\section{Introduction}
\label{sec:introduction}
\input{Introduction}

\section{Summary of main results}
\label{sec:summary}
\input{ResultsSummary}

\section{Setup \& notation}
\label{sec:mainsetup}
\input{MainSetup}

\section{Inaccuracy of covariant codes for a continuous symmetry}
\label{sec:resultsbounds}
\input{ResultsBounds}

\section{Criterion for certifying code performance}
\label{sec:criterioncode}
\input{CriterionCode}

\section{Examples of covariant codes}
\label{sec:Uonecodes}
\input{ResultsExamplesCovariantCodes}

\section{Approximate Eastin-Knill theorem}
\label{sec:repbound}
\input{Repbound}
\subsection{Random Constructions}
\label{sec:random-codes}
\input{Random_Code_Main_Text}
\subsection{Generalized $W$-state encoding}
\label{sec:w-codes}
\input{WStateExample}

\section{Error-correcting codes for general groups}
\label{subsec:Gcodes}
\input{Gcodes}

\section{Symmetries and error correction in quantum gravity}
\label{sec:holography}
\input{Holography}

\section{Discussion}
\label{sec:discussion}
\input{Discussion}

\clearpage
\newgeometry{hmargin=1.5in,vmargin=1.25in}
\onecolumngrid
\setstretch{1.1}
\appendix

{\par{\centering\bfseries \uppercase{Supplemental Material}\par}

\appendixtableofcontents

\section{Proof of our bounds for a covariant code}
\label{appx:ProofBounds}
\input{AppendixProofGeneralBound}

\section{Correlation functions and bounds}
\label{appx:corrbound}
\input{CorrelationBound}

\section{Criterion for approximate codes}
\label{appx:criterion-certify-code}
\input{AppendixCriteriaCodes}

\section{Calculations for covariant code examples}
\label{appx:AppendixCalcCodes}
\subsection{Three-rotor secret-sharing code}
\paragraph{Sharp cutoff.}
\input{AppendixCalcTruncHaydenCode}
\paragraph{Smooth cutoff.}
\input{AppendixCalcGaussianHaydenCode}
\subsection{Five-rotor perfect code}
\label{appx:AppendixCalPerfectCode}
\input{AppendixCalPerfectCode}

\subsection{Thermodynamic codes}
\label{appx:AppendixCalcThermoCode}
\input{AppendixCalcThermoCode}

\section{Proof of the approximate Eastin-Knill theorem}
\label{appx:ProofUdBounds}
\input{udproof}

\section{Circumventing the Eastin-Knill theorem by randomized constructions}   
\label{appx:Random_Consts}
\label{appx:Random_Theorems}

\subsection{Randomized constructions: Overview}
\label{subsec:Random_Consts-overview}
\input{Random_Consts}

\subsection{Randomized constructions: Detailed proofs}   
\input{Random_Theorems}

\section{Some general lemmas}
\input{AppendixGeneralLemmas}

%
\catcode`\&=12\relax %
\def\doibase#110.{https://doi.org/10.}%
\def\ {\unskip\space}%
\bibsep=2pt\relax

\input{ApproxCovariantQEC.bbl}
\end{document}

%% file: preamblesetup.tex
\usepackage[reset]{geometry}
\makeatletter\Gm@restore@org\makeatother

\usepackage{setspace}
\usepackage[preset=reset,pkgset=extended,hyperrefdefs={noemail}]{phfnote}
\usepackage[thmset=rich,smallproofs=false]{phfthm}

\usepackage[sfdefault]{sourcesanspro}

\usepackage[capitalize,nameinlink,noabbrev]{cleveref}

\colorlet{sdpdualvar}{gray}

\usepackage{ragged2e}
\DeclareCaptionJustification{myjust}{\justify}
\makeatletter
\def\big{\bBigg@{1.2}}
\def\Big{\bBigg@{1.5}}
\def\bigg{\bBigg@{2}}
\def\Bigg{\bBigg@{2.5}}
\def\biggg{\bBigg@{3}}
\def\Biggg{\bBigg@{3.5}}
\usepackage{breqn}
\preto\maketitle{%
  \begingroup\lccode`~=`,
  \lowercase{\endgroup
  \let\saved@breqn@active@comma~%
  \let~}\active@comma %
}
\appto\maketitle{%
  \begingroup\lccode`~=`,
  \lowercase{\endgroup
  \let~}\saved@breqn@active@comma %
}
\AtBeginDocument{%
    \catcode`_=12
    \begingroup\lccode`~=`_
    \lowercase{\endgroup\let~}\sb
    \mathcode`_="8000
    \immediate\write\@auxout{\catcode`_=12 }%
    \immediate\write\@auxout{\catcode`^=12 }%
}
\makeatother

\makeatletter
\appto\phfthm@hookproof@proof@startafterdisplay{%
  \parskip=0.6em plus 0.1em minus 0.1em%
  \parindent=0pt%
  \allowdisplaybreaks%
}
\makeatother

\makeatletter
\let\phf@old@leq\leq
\let\phf@old@geq\geq
\let\leq\leqslant
\let\geq\geqslant
\expandafter\apptocmd\csname dmath*\endcsname{
  \let\leq\phf@old@leq
  \let\geq\phf@old@geq
}{}{PATCH FAILURE}
\makeatother

\usepackage{phfqit}
\usepackage{phfparen}

\csdef{phfqit@fmtGroup}#1{\mathit{#1}}
\csdef{phfqit@fmtLieAlgebra}#1{\mathfrak{#1}}

\makeatletter
\def\ee#1#2{%
  \edef\@tmpa{\detokenize{#1}}%
  \edef\@tmpb{\detokenize{^}}%
  \ifx\@tmpa\@tmpb%
    e^{#2}%
  \else%
    \PackageError{no-package}{SYNTAX ERROR:
      \string\ee expects \string\ee\string^{...}}{help here...}%
  \fi%
}%
\makeatother

\DeclareMathOperator\erfc{erfc}

\crefname{theorem}{Theorem}{Theorems}
\crefname{proposition}{Proposition}{Propositions}
\crefname{lemma}{Lemma}{Lemmas}
\crefname{corollary}{Corollary}{Corollaries}
\crefname{equation}{Eq.}{Eqs.}
\Crefname{equation}{Equation}{Equations}
\crefname{figure}{Fig.}{Figs.}

\crefname{section}{}{}
\Crefname{section}{Section}{Sections}
\creflabelformat{section}{#2\S#1#3}
\crefrangelabelformat{section}{#3\S#1#4--#5#2#6}
\crefname{subsection}{}{}
\Crefname{subsection}{Section}{Sections}
\creflabelformat{subsection}{#2\S#1#3}
\crefrangelabelformat{subsection}{#3\S#1#4--#5#2#6}
\crefname{subsubsection}{}{}
\Crefname{subsubsection}{Section}{Sections}
\creflabelformat{subsubsection}{#2\S#1#3}
\crefrangelabelformat{subsubsection}{#3\S#1#4--#5#2#6}

\makeatletter
\renewcommand*{\eqref}[1]{%
  \hyperref[{#1}]{\textup{\tagform@{\ref*{#1}}}}%
}
\makeatother

\makeatletter
\def\phfcrefnamelbl{\@ifstar{\@phfcrefnamelbl{cref}{@plural}}{\@phfcrefnamelbl{cref}{}}}
\def\phfCrefnamelbl{\@ifstar{\@phfcrefnamelbl{Cref}{@plural}}{\@phfcrefnamelbl{Cref}{}}}
\def\@phfcrefnamelbl#1#2#3{%
  \ifcsname #1@#3@name#2\endcsname
    \csname #1@#3@name#2\endcsname
  \else
    \PackageError{preamblesetup.tex}{Error, invalid \string\cref-label type: #3}{}%
  \fi
}
\makeatother

\makeatletter
\def\appendixtableofcontents{%
  \vspace{2em}%
  \par\noindent{\itshape List of Appendices\par}%
  \vspace*{0ex}%
  \@starttoc{appendixtoc}%
}
\def\l@f@section{%
 \addpenalty{\@secpenalty}%
 \addvspace{0.4em plus\p@}%
}
\makeatother

\makeatletter
\newcounter{appendix}
\def\theappendix{\Alph{appendix}}

\newcommand\appendixsection[2][]{%
  \@startsection
    {appendix}%
    {1}%
    {\z@}%
    {1.2cm \@plus1ex \@minus .2ex}%
    {.6cm}%
    {%
      \normalfont%
      \bfseries
      \centering
    }[#1]{#2}%
    \addcontentsline{appendixtoc}{section}{Appendix~\theappendix: #2}%
  }
\def\appx@phantomsection{%
  \Hy@MakeCurrentHrefAuto{appendix*}%
  \Hy@raisedlink{\hyper@anchorstart{\@currentHref}\hyper@anchorend}%
}
\apptocmd\appendix{%
  \let\phantomsection\appx@phantomsection%
  \let\section\appendixsection%
  \numberwithin{subsection}{appendix}
}{}{PATCH FAILURE}
\def\appendixparagraph{\@ifstar\appendixparagraph@\appendixparagraph@}
\def\appendixparagraph@{%
  \@startsection{paragraph}{4}{\z@}%
  {3.25ex \@plus1ex \@minus.2ex}%
  {-1em}%
  {\normalfont\normalsize\bfseries}*}
\apptocmd\appendix{%
  \let\paragraph\appendixparagraph}{}{PATCH FAILURE}
\makeatother
\apptocmd\appendix{%
  \numberwithin{equation}{appendix}%
  \def\theequation{\theappendix.\arabic{equation}}%
}{}{PATCH FAILURE}


%% file: Introduction.tex
Quantum error-correcting codes protect fragile quantum states against
noise~\cite{BookNielsenChuang2000}.  
If quantum information is cleverly encoded in a highly entangled state of many physical subsystems, then damage inflicted by local interactions with the environment can be reversed by a suitable recovery operation.
Aside from their applications to resilient quantum computing, quantum
error-correcting codes appear in a wide variety of physical settings where quantum states are delocalized over many subsystems, such as
topological phases of
matter~\cite{Kitaev2003AoP_anyons,Dennis2002JMP_topological,%
  Nayak2008RMP_nonabelian,BookZeng2015arXiv_matter} and the AdS/CFT
correspondence in holographic quantum
gravity~\cite{Almheiri2015JHEP_bulk,Pastawski2015JHEP_holographic}.

On the other hand, naturally occurring physical systems often respect symmetries, and phases of matter can be classified according to how these symmetries are realized in equilibrium states. 
Likewise, quantum error-correcting codes often have approximate or exact symmetries with important implications. In the case of a time-translation-invariant many-body system, for example,
certain energy subspaces are known to form
approximate quantum error-correcting codes~\cite{Brandao2017arXiv_chainAQECC,Gschwendtner2019arXiv_lowenergies}, which are preserved under time evolution. Limits to sensitivity in quantum metrology are related to the degree of asymmetry of probe states, a notion formalized in the resource theory of asymmetry and
reference frames~\cite{Bartlett2007_refframes,Marvian2014NC_extending}.  Thus, reference frame information can be protected against noise using quantum codes with
suitable symmetry properties~\cite{Hayden2017arXiv_frame}.
Furthermore, recent developments in quantum gravity have shown that the
AdS/CFT correspondence can be viewed as a quantum error-correcting code which is
expected to be compatible with the natural physical symmetries of the system,
such as time-translation invariance~\cite{Harlow2018TASI_emergence,%
  Susskind2018PiTP_complexity1,Pastawski2015JHEP_holographic,%
  Harlow2018arXiv_constraints,Harlow2018arXiv_symmetries}.
Finally, the Eastin-Knill theorem~\cite{Eastin2009PRL_restrictions,%
  Zeng2011IEEETIT_transversality,Chen2008PRA_subsystem}, which complicates the
construction of fault-tolerant schemes for quantum computation by forbidding
quantum error-correcting codes from admitting a universal set of transversal
gates, can be viewed as the statement that finite-dimensional quantum codes
which correct erasure have no continuous
symmetries~\cite{Hayden2017arXiv_frame}. Thus, there are loopholes to the
Eastin-Knill theorem that are naturally exploited by holographic theories of
quantum gravity. This article provides a detailed quantitative investigation of
those loopholes, critically evaluating their potential for application to
quantum fault-tolerance.

A continuous symmetry, as opposed to a discrete symmetry, allows for
infinitesimally small transformations that are arbitrarily close to the identity
operation.  Such symmetry transformations are generated by conserved operators called \emph{charges}.
For instance, consider a particle in three-dimensional space that we rotate
about the $Z$-axis by an angle $\theta$.  Acting on the Hilbert space, this
symmetry transformation is represented by a unitary $U_\theta$ that is generated
by the $Z$-component of the Hermitian angular momentum operator $J_z$, i.e.,
$U_\theta = \ee^{-iJ_z\theta}$.  Crucially, a unitary operation $U$ that is
covariant with respect to rotations about the $Z$ axis must conserve the
physical quantity $J_z$.  In particular, if the initial state $\ket\psi$ is an
eigenstate of $J_z$ with eigenvalue $m$, then the transformed state $U\ket\psi$
must also be an eigenstate of $J_z$ with the same eigenvalue (up to a constant
shift in all the eigenvalues).

\begin{figure}
  \centering
  \includegraphics{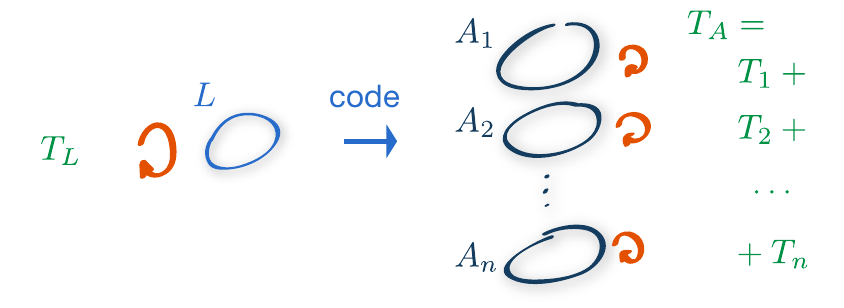}
  \caption{Quantum information represented on an abstract logical system $L$ is
    encoded on several physical subsystems $A_1\ldots A_n$ using a code.
    Suppose that the code is compatible with a continuous transversal symmetry,
    for instance, rotations in 3D space.  This means that by rotating all
    individual physical subsystems one induces the same transformation as if we
    had simply rotated the initial logical system $L$.  We show that such codes
    necessarily perform poorly as approximate error-correcting codes against
    erasures.  The reason is that the code must encode an eigenstate of the
    logical charge $T_L$ that generates the symmetry as a codeword that is a
    global eigenstate of the corresponding physical charge $T_A$.  Since the
    latter is a sum of local terms $T_A = \sum T_i$, and since the environment
    is handed the local reduced states of the codeword, the environment can
    deduce on average the total charge of the codeword.  Because logical
    information leaks to the environment, the code cannot be a good
    error-correcting code.}
  \label{fig:covariant-code}
\end{figure}

Here, we study the accuracy of quantum error-correcting codes that are covariant with
respect to continuous symmetries (\cref{fig:covariant-code}). Our results
build on earlier work showing that infinite-dimensional covariant quantum codes exist, while finite-dimensional covariant codes cannot correct erasure errors perfectly~\cite{Preskill2000arXiv_synchronization,%
Hayden2017arXiv_frame}.

A finite-dimensional error-correcting code that is covariant with respect to a continuous
symmetry cannot correct erasure of a subsystem exactly, because an adversary who steals the erased subsystem could acquire some information about the encoded state, hence driving irreversible decoherence of the logical quantum information~\cite{Preskill2000arXiv_synchronization,%
  Hayden2017arXiv_frame}.
More concretely, if $\Pi$ is the projector onto the code space, then the
error-correction
conditions~\cite{Knill1997PRA_correction,Bennett1996PRA_MSEntglQECorr} state
that any operator $O$ supported on the erased subsystem must act trivially within the codespace, i.e.,
$\Pi O\Pi\propto \Pi$.  If the symmetry acts transversally, the corresponding
generator $T_A$ is a sum of strictly local terms, $T_A = \sum T_i$, where each $T_i$ is supported on a single subsystem.  However, this
implies that $\Pi T_A \Pi = \sum \Pi T_i \Pi \propto \Pi$, and hence it follows from the error-correction condition that any such
$T_A$ must act trivially on the 
codewords.

Crucially for the considerations in this paper, the above argument makes two
implicit assumptions: that the sum over $i$ is finite (bounded number of
subsystems), and that the codewords are normalizable (finite-dimensional
subsystems). If both assumptions are relaxed, then quantum codes covariant with
respect to a continuous symmetry are possible, as shown
in~\cite{Hayden2017arXiv_frame}.
Our main task in this paper is to explore quantitatively the case where the number of subsystems and the dimension of each subsystem are finite, using the tools of approximate quantum
error correction~\cite{Leung1997PRA_better,Crepeau2005ECRY_AQECC,Beny2010PRL_AQECC}. That is, we will quantify the deviation from perfect correctability in this case, for a code covariant with respect to a continuous symmetry. Assuming that the symmetry acts transversally and that
the noise acts by erasing one or more subsystems, we provide upper bounds on the accuracy
of the code, characterized using either the average entanglement fidelity or the
worst-case entanglement fidelity of the error-corrected state.  Our proof strategy is to show that in the presence of a continuous symmetry, the environment necessarily learns some
information about the logical charge, which implies that the code necessarily
performs imperfectly as an error-correcting code~\cite{Hayden2008OSID_decoupling,%
  Beny2010PRL_AQECC,Beny2018arXiv_constraints}.
In fact, some of these assumptions may be relaxed in our main technical theorem; for
instance, the generating charge may be a sum of $k$-local terms, instead of a sum of strictly local terms as for a transversal symmetry action, and the code only needs to be
approximately rather than exactly covariant.

Our lower bound on infidelity vanishes in two interesting regimes: as the dimension $d$ of the physical
subsystems gets large, or as the number $n$ of physical subsystems gets large.  In these limits we can find error-correcting codes whose infidelity approximately matches the
scaling of our bound with $d$ or $n$.  We construct explicit examples based on normalized
versions of the rotor code presented in Ref.~\cite{Hayden2017arXiv_frame}, and
note that codes considered in Ref.~\cite{Brandao2017arXiv_chainAQECC} provide
further examples.  We also discuss a 5-rotor code that can be stacked to construct a
covariant holographic code~\cite{Pastawski2015JHEP_holographic}.

Furthermore, our results provide an approximate version of the
Eastin-Knill theorem~\cite{Eastin2009PRL_restrictions,%
  Zeng2011IEEETIT_transversality,Chen2008,Hayden2017arXiv_frame}, which states
that a universal set of transversal logical gates cannot exist for a
finite-dimensional encoding that protects perfectly against erasure.
By applying our bounds and exploiting the nonabelian nature of the full unitary
group on the logical space, we derive a lower bound on infidelity which scales
as $1/\log d$, where $d$ is the subsystem dimension, for a code that admits
universal transversal logical gates.  We also find that if a code admits a
universal set of transversal logical gates, then there are strong lower bounds
on the subsystem dimension $d$ that depend on the code's infidelity, and which
in some regimes are even exponential in the logical system dimension $d_L$.
Using randomized code constructions, we prove the existence of codes which approximately achieve this relationship between $d$ and $d_L$. In addition, we exhibit codes with universal transversal logical gates which achieve arbitrarily small infidelity when the number $n$ of subsystems becomes large with the logical dimension $d_L$ fixed. 

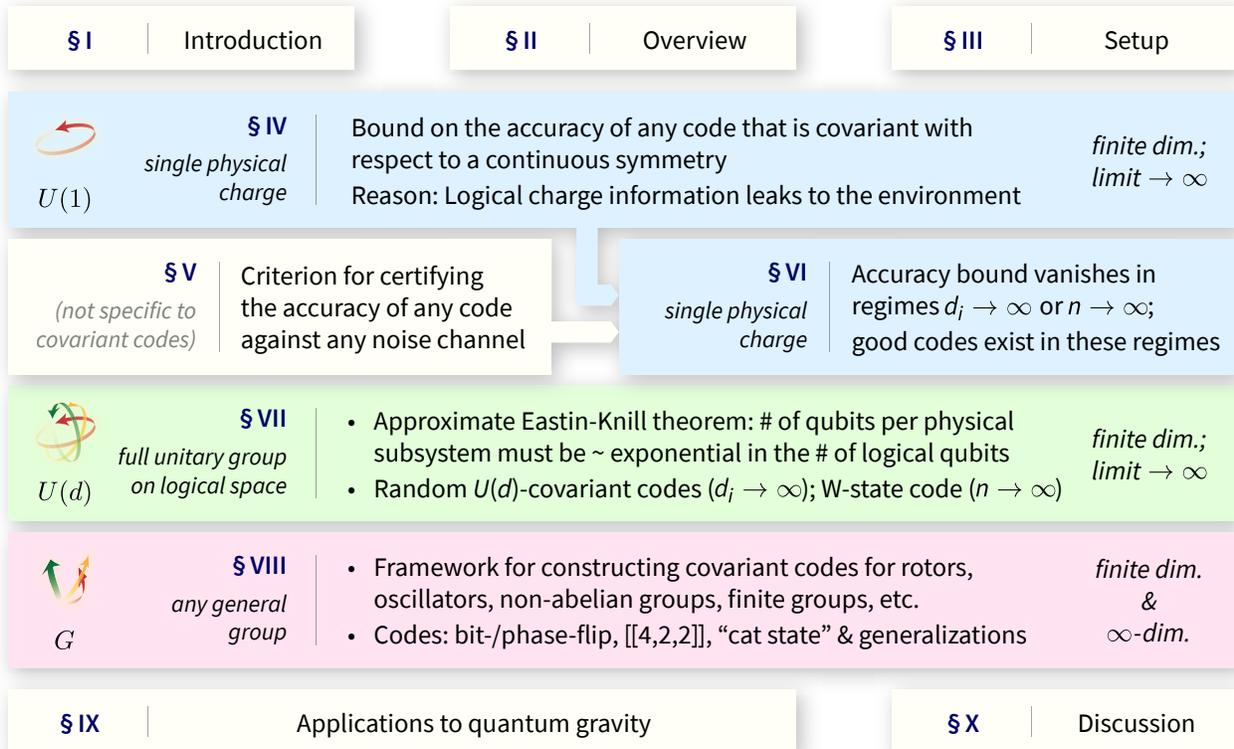
\begin{figure*}
  \centering
  \input{fig/OverviewResultsWithRef}
  \caption{Our paper is built around three main technical results.  In the
    presence of a $\UU(1)$ symmetry, which is implied by any continuous
    symmetry, we prove a general bound on how well any covariant code can
    correct against erasures.  Our second main contribution is an approximate
    version of the Eastin-Knill theorem: If the code admits a universal set of
    transversal logical gates, i.e., if it is covariant with respect to the full
    unitary group $\UU(d)$ on the logical space, then our bound can be expressed
    in terms of the physical subsystem dimensions $d_i$.  Our third main
    contribution is a general framework for constructing codes that are
    covariant with respect to any general group $G$.  Along the way, we develop
    a new criterion for certifying the accuracy of any code, which we use to analyze our
    examples.}
  \label{fig:OverviewResults}
\end{figure*}

We also provide a general framework for constructing codes that are covariant with
respect to general symmetry groups, by encoding logical information into the
so-called \emph{regular representation} of the groups.  Using this framework
we can generalize several widely-known codes (bit-flip, phase-flip, $[[4,2,2]]$
code, etc.) to infinite-dimensional covariant codes based on oscillators or rotors.

Finally, we discuss the interpretation of our results in the context of quantum gravity and, in particular, the AdS/CFT correspondence. Time evolution itself provides an example of a symmetry that must be reconciled with the error-correcting properties of the system.

The remainder of the manuscript is organized as follows
(\cref{fig:OverviewResults}). In \cref{sec:summary}, we summarize our main
results. We set up notation in \cref{sec:mainsetup} and prove a bound on the
performance of codes covariant with respect to a $\UU(1)$ symmetry in
\cref{sec:resultsbounds}. A criterion certifying code performance is derived in
\cref{sec:criterioncode}. In \cref{sec:Uonecodes}, we apply our bounds and
criterion to the following examples of $\UU(1)$-covariant encodings: an
infinite-dimensional rotor extension of the qutrit $[[3,1,2]]$ and qubit
$[[5,1,3]]$ codes as well as a many-body Dicke-state code. We apply our bound to
codes admitting universal transversal gates in \cref{sec:repbound}, discussing a
$\UU(d)$-invariant encoding based on $W$-states in
\cref{sec:w-codes}. Erasure-correcting codes whose transversal gates form a
general group $G$ are introduced in \cref{subsec:Gcodes}.  In
\cref{sec:holography} we study applications to quantum gravity.  We conclude
with a discussion in \cref{sec:discussion}.


%% file: fig/OverviewResultsWithRef.tex
\begingroup\makeatletter
\def\tocsectionreflink#1{\makebox[15mm][c]{\hyperref[#1]{\fontseries{sb}\selectfont\S\,\ref*{#1}}}}
\def\tocsectionreflinkbig#1{\makebox[15mm][r]{\hyperref[#1]{\fontseries{sb}\selectfont\S\,\ref*{#1}}}}
\def\mksectionref#1{%
  \raisebox{-2.62mm}{\makebox[0pt][r]{%
      \sffamily
      #1}}}
\def\putfromtop(#1,#2){%
  \edef\x{\noexpand\put(#1,\strip@pt\dimexpr \mypictheight pt - #2pt)}\x}
\setlength{\unitlength}{1mm}
\def\mypictwidth{170}%
\def\mypictheight{105}%
\begin{picture}(\mypictwidth,\mypictheight)
\put(0,0){\includegraphics{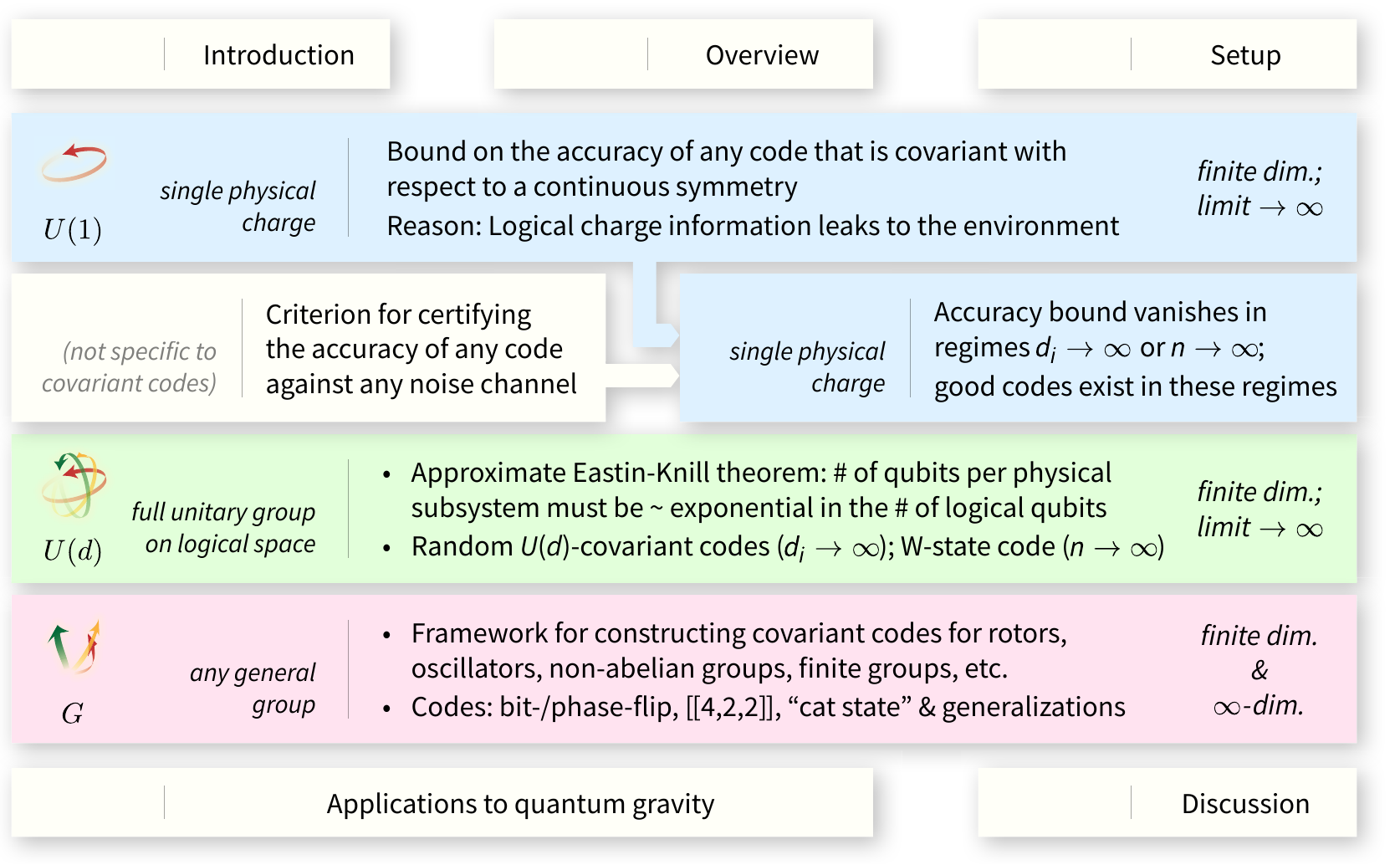}}
\putfromtop(18.376,5.3){\mksectionref{\tocsectionreflink{sec:introduction}}}
\putfromtop(77.087,5.3){\mksectionref{\tocsectionreflink{sec:summary}}}
\putfromtop(135.823,5.3){\mksectionref{\tocsectionreflink{sec:mainsetup}}}
\putfromtop(38.4,16.65){\mksectionref{\tocsectionreflinkbig{sec:resultsbounds}}}
\putfromtop(26.365,36.15){\mksectionref{\tocsectionreflinkbig{sec:criterioncode}}}
\putfromtop(107.589,36.15){\mksectionref{\tocsectionreflinkbig{sec:Uonecodes}}}
\putfromtop(38.4,55.65){\mksectionref{\tocsectionreflinkbig{sec:repbound}}}
\putfromtop(38.4,75.15){\mksectionref{\tocsectionreflinkbig{subsec:Gcodes}}}
\putfromtop(18.376,96.2){\mksectionref{\tocsectionreflink{sec:holography}}}
\putfromtop(135.823,96.2){\mksectionref{\tocsectionreflink{sec:discussion}}}
\end{picture}
\endgroup

%% file: ResultsSummary.tex
\subsection{Bound on the accuracy of codes covariant with respect to a
  continuous symmetry}
  
Our first main result is a bound on the accuracy of any approximate quantum
error-correcting code that is covariant with respect to a continuous symmetry. We consider an encoding map from a logical system $L$ to a physical system $A$ consisting of $n$ subsystems denoted $A_1, A_2, \dots A_n$. A one-parameter family of continuous unitary symmetries acting on $L$ is generated by the logical charge observable $T_L$, which corresponds to the physical charge observable $T_A$ acting on $A$. We assume that the symmetry acts transversally, so that $T_A = \sum_{i=1}^n T_i$, where $T_i$ acts on subsystem $T_i$.

How well does this code protect the logical system against erasure of one of the
subsystems? To quantify the code's performance we may use the \textit{worst-case
  entanglement fidelity}, where ``worst-case'' means the minimal fidelity for
any entangled state shared by the logical system and a reference system. (See
\cref{sec:mainsetup} for a precise definition.) Then we consider the value
$f_{\mathrm{worst}}$ of this worst-case entanglement fidelity which is achieved
by the best possible recovery map applied after an erasure error. A measure of
the residual error after recovery is
\begin{equation}
  \epsilon_{\mathrm{worst}} = \sqrt{1 - f_{\mathrm{worst}}^2}.
\end{equation}
Our result is a lower bound on $\epsilon_{\mathrm{worst}}$ which limits the
performance of any covariant quantum code:
\begin{align}
  \epsilon_{\mathrm{worst}}
  &\geqslant  \frac{\Delta T_L}{2n\max_i \Delta T_i}\ ,
  \label{eq:overview-main-eps-worst-bound}
\end{align}
where $\Delta T$ denotes the difference between the maximal and minimal
eigenvalue of $T$.  That is, the code's accuracy is constrained by the range of
charges one wishes to be able to encode, by the size of the charge fluctuations
within each subsystem, and by the number of physical subsystems.

We also find that~\eqref{eq:overview-main-eps-worst-bound} can be generalized in a number of ways. We can express the limit on code performance in terms of other measures besides worst-case entanglement fidelity, such as average entanglement fidelity, or the entanglement fidelity of a fixed input state. We can derive bounds that apply in the case where more than one subsystem is erased, or where the erasure occurs for an unknown subsystem rather than a known subsystem. We can consider cases where the charge distribution for a subsystem has infinite range, but with a normalizable tail. We can also treat the case where the covariance of the code is approximate, or where the physical charge operator is not strictly transversal.

\subsection{Regimes where our bound is circumvented and criterion for code performance}

The idea underlying~\eqref{eq:overview-main-eps-worst-bound} is that for erasure
correction to work well one should not be able to learn much about the global
value of the charge by performing a local measurement on a subsystem. Hence, to
be able to correct the errors to good accuracy, we need either large local
charge fluctuations ($\Delta T_i\to\infty$), or many subsystems ($n\to\infty$)
so that the global charge is a sum of many local contributions. In fact, codes
can be constructed in either limit for which $\epsilon_{\mathrm{worst}}$
approximately matches the scaling in $\Delta T_i$ and $n$ of the lower
bound~\eqref{eq:overview-main-eps-worst-bound}.

To study the case of large $\Delta T_i$, we consider a normalized variant of the
infinite-dimensional covariant code constructed
in~\cite{Hayden2017arXiv_frame}. The infinite-dimensional version encodes one
logical rotor (with unbounded $U(1)$ charge) in a code block of three rotors. In
the modified version of this code, we either truncate the charge of the logical
system to $\{ -h, -h+1, \ldots, +h\}$ or use a Gaussian envelope of width $w$ to
normalize the physical codewords. The value of $\epsilon_{\mathrm{worst}}$
achieved by this code, and our lower bound, both scale like $h/w$ up to a
logarithmic factor.

Regarding the limit of a large number of subsystems, we observe that a code
discussed in Ref.~\cite{Brandao2017arXiv_chainAQECC} matches the $1/n$ scaling
of our lower bound on $\epsilon_{\mathrm{worst}}$.  Here the subsystems are
qubits, regarded as spin-$1/2$ particles, and the code space is two-dimensional,
spanned by two Dicke states with different values of the total angular momentum
$J_z$ along the $z$-axis. (A Dicke state is a symmetrized superposition of all
basis states with a specified $J_z$). This code is covariant with respect to
$z$-axis rotations by construction, and can be shown to achieve
$\epsilon_{\mathrm{worst}}$ scaling like $1/n$, where $n$ is the number of
physical qubits.

A further result of independent interest is a general criterion used in our
analysis for certifying the performance of an error-correcting code against
arbitrary noise.  Stated informally, this criterion asserts that if the reduced
density operator on each subsystem is approximately the same for all codewords,
and if the environment does not get any information from the off-diagonal terms
in the logical density operator, then the code performs well.  While this
criterion is sufficient to certify the performance of an approximate
error-correcting code, it is not necessary---there may be codes achieving small
$\epsilon_{\mathrm{worst}}$ that do not satisfy it.

\subsection{Approximate Eastin-Knill theorem and random $\UU(d)$-covariant codes}
Quantum error-correcting codes are essential for realizing scalable quantum computing using realistic noisy physical gates. In a fault-tolerant quantum computation, logical quantum gates are applied to encoded quantum data, and error recovery is performed repeatedly to prevent errors due to faulty gates from accumulating and producing uncorrectable errors at the logical level. For this purpose, transversal logical gates are especially convenient. For example, if a logical gate on an $n$-qubit code block can be achieved by applying $n$ single-qubit gates in parallel, then each faulty physical gate produces only a single error in the code block. Nontransversal logical gates, on other hand, either require substantially more computational overhead, or propagate errors more egregiously, allowing a single faulty gate to produce multiple errors in a code block. 

A nontrivial transversal logical gate can be regarded as a covariant symmetry operation acting on the code. If all the logical gates in a complete universal gate set could be chosen to be transversal, then the Lie group of transversal logical gates would coincide with the group $\UU(d_L)$ of unitary gates acting on the $d_L$-dimensional logical system (up to an irrelevant overall phase). It then follows that \textit{any} generator $T_A$ of $\UU(d_L)$ acting on the physical system $A$ could be expressed as a sum of terms, where each term in the sum has support on a single subsystem. 
Unfortunately, the Eastin-Knill theorem rules out this appealing scenario, if erasure of each subsystem is correctable and the code is finite-dimensional. But now that we have seen that there are parameter regimes in which covariance \textit{can} be compatible with good performance of approximate quantum error-correcting codes, one wonders whether a universal transversal logical gate set is possible after all, at the cost of a small but nonzero value of $\epsilon_{\mathrm{worst}}$.

We have found, however, that a fully $\UU(d_L)$-covariant code requires a value of $\epsilon_{\mathrm{worst}}$ which scales quite unfavorably with the local subsystem dimension.
Leveraging tools from representation theory, we show that the lower bound on
$\epsilon_{\mathrm{worst}}$ becomes
\begin{align}
  \epsilon_{\mathrm{worst}}
  \geqslant \frac1{2n \max_i \ln d_i} + O\mathopen{}\left(\frac{1}{n d_L}\right)
  \ ,\label{eq:udboundsimple}
\end{align}
where $d_i$ is the dimension of the $i$th physical subsystem.
We also find lower bounds for the local subsystem dimension that depend on the
number of logical qubits and the code's infidelity.  This result also applies to
the case when each gate can be approximated with a discrete sequence of
transversal operations to arbitrary accuracy, as in the context of the
Solovay-Kitaev theorem.

Furthermore, using randomized constructions, we prove the existence of
$\UU(d_L)$-covariant code families which achieve arbitrarily small infidelity in
the limit of large subsystem dimension.  In addition, we exhibit a simple
$\UU(d_L)$-covariant code family, whose codewords are generalized $W$-states,
such that $\epsilon_{\mathrm{worst}}$ approaches zero as the number of
subsystems $n$ approaches infinity.

\subsection{Framework for constructing covariant codes}
We also develop a general framework for constructing codes that are covariant
with respect to any group $G$ admitting a Haar measure, where both the logical
system and the physical subsystems transform as the regular representation of
$G$. In this construction, the dimension of each subsystem is the order $\abs{G}$ of
the group when $G$ is finite, and infinite when $G$ is a Lie group.

Using this formalism we construct natural generalizations of well-known families
of qubit codes, such as the bit-flip and phase-flip codes, with the qubits
replaced by $\abs{G}$-dimensional systems. These codes admit transversal logical
gates representing each element of $G$.

When $G$ is a Lie group, the qubits are replaced by infinite-dimensional systems
such as rotors or oscillators. These infinite-dimensional codes circumvent the
Eastin-Knill theorem---they are covariant with respect to a continuous symmetry
group, yet erasure of a subsystem is perfectly correctable.


%% file: MainSetup.tex
\subsection{Approximate error correction}

Consider a code, which to each logical state $\ket{x}_L$ on some abstract
logical system $L$ associates a state $\ket{\psi_x}_{A_1A_2\ldots A_n}$ on a
physical system $A$ consisting of $n$ subsystems
$A = A_1\otimes A_2\otimes \cdots \otimes A_n$ (\cref{fig:code}).  The span
of all codewords $\{ \ket{\psi_x}_{A} \}$ forms the \emph{code subspace}.  More
generally, we denote by $\mathcal{E}_{L\to A}(\cdot)$ the encoding channel which
associates to any logical state the corresponding encoded physical state.  In
this work, the encoding is usually an isometry, \textit{i.e.}, the encoding itself does
not introduce noise into the system.
\begin{figure}
  \centering
  \includegraphics{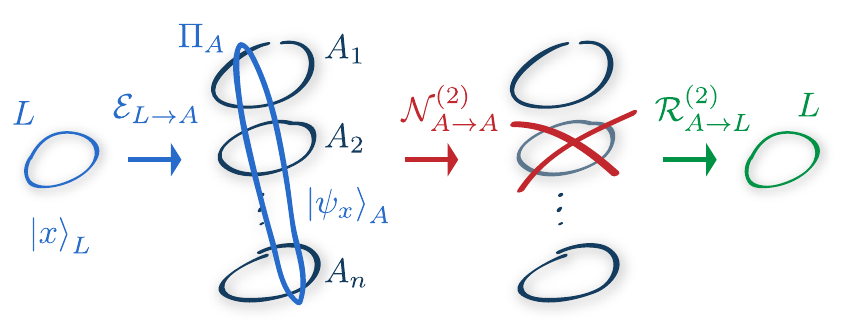}
  \caption{A code $\mathcal{E}_{L\to A}$ maps a logical state $\ket{x}$ on an
    abstract logical space $L$ to a state $\ket{\psi_x}_A$ on a physical system
    $A$.  Here, we consider a physical system composed of several subsystems
    $A=A_1\otimes A_2\otimes\cdots A_n$.  The code space, with associated
    projector $\Pi_A$, is the range of the encoding map.  The environment acts
    by erasing a subsystem, represented as a noise channel
    $\mathcal{N}_{A\to A}^{i}$.  A good error-correcting code is capable of
    recovering the original logical state $\ket{x}$ from the remaining
    subsystems, by applying a recovery map $\mathcal{R}^{i}_{A\to L}$.  In our
    analysis, we assume that the environment chooses randomly which subsystem is
    erased.  A record of which subsystem was chosen is provided, allowing to
    apply a different recovery map for each erasure situation.  The quality of
    the code is characterized by how close the overall process is to the
    identity process on the logical system, as measured by either the average or
    the worst-case entanglement fidelity.}
  \label{fig:code}
\end{figure}

The noise channel is the process to which the physical system is exposed, which
might cause the information encoded in it to get degraded.  It is a quantum
channel $\mathcal{N}_{A\to B}$ mapping the physical system to physical system
$B$.  (The system $B$ might be the same as $A$, but it might be different; for
instance, $B$ might include a register which remembers which type of error
occurred or which subsystem was lost.)

To study the approximate error correction properties of a code, we need to
quantify the approximation quality using distance measures between states and
channels.
Proximity between quantum states can be quantified using the trace distance
$\delta(\rho,\sigma) = \norm{\rho-\sigma}_1/2$, or using the fidelity\footnote{%
  Throughout this paper, we stick to the convention that the fidelity and its
  derived quantities refer to an amplitude rather than a probability, i.e., we
  use the convention of ref.~\cite{BookNielsenChuang2000}.  In the literature,
  the quantity that we denote by $F^2$ is also referred to as ``fidelity,''
  while the quantity we represent by $F$ is sometimes called ``root fidelity.''}
$F(\rho,\sigma) =
\norm{\sqrt{\rho}\sqrt{\sigma}}_1$~\cite{BookNielsenChuang2000}.  We need to
quantify how close a quantum channel $\mathcal{K}_{L\to L}$ is to the identity
channel.  Two standard measures to achieve this are the average entanglement
fidelity $F_{\mathrm{e}}$ and the worst-case entanglement fidelity
$F_{\mathrm{worst}}$~\cite{Schumacher1996PRA_sending,Gilchrist2005PRA_processes},
defined as
\begin{align}
  F_{\mathrm{e}}^2(\mathcal{K})
  &= \bra{\hat\phi} (\mathcal{K}\otimes\IdentProc[]{})(\proj{\hat\phi}) \ket{\hat\phi}\ ;
  \\
  F_{\mathrm{worst}}^2(\mathcal{K})
  &= \min_{\ket\phi}\;
    \bra{\phi} (\mathcal{K}\otimes\IdentProc[]{})(\proj{\phi}) \ket{\phi}\ .
\end{align}
Here the input state appearing in the definition of $F_{\mathrm{e}}$ is
$\ket{\hat\phi}_{LR} = \sum_{k=0}^{d_L -1} \ket{k}\otimes\ket{k} / \sqrt{d_L}$,
the maximally entangled state of $L$ and a reference system $R$; the system $R$
has the same dimension as $L$, which we denote by $d_L$. The optimization in the
definition of $F_{\mathrm{worst}}$ ranges over all bipartite states of $L$ and
$R$.  We may also use the state fidelity $F(\rho,\sigma)$ to compare two
channels $\mathcal{K}$ and $\mathcal{K'}$; the entanglement fidelity between
$\mathcal{K}$ and $\mathcal{K'}$ for a fixed bipartite input state
$\ket\phi_{LR}$ is defined as
\begin{align}
  F_{\ket\phi}^2(\mathcal{K}, \mathcal{K}')
  &= F^2 `\Big((\mathcal{K}\otimes\IdentProc[]{})(\proj\phi),
    (\mathcal{K}'\otimes\IdentProc[]{})(\proj\phi))\ ;
  \label{eq:fixed-input-entgl-fidelity-channels}
\end{align}
thus 
\begin{align}
F_{\mathrm{e}}(\mathcal{K}) = F_{\ket{\hat\phi}}(\mathcal{K}, \IdentProc[]{}),\quad
F_{\mathrm{worst}}(\mathcal{K}) = \min_{\ket\phi} F_{\ket\phi}(\mathcal{K},
\IdentProc[]{}).
\end{align}
By optimizing over the input state, we may define
$F_{\mathrm{worst}}(\mathcal{K},\mathcal{K'})$, which is closely related to the
diamond distance between the
channels~\cite{Schumacher1996PRA_sending,Gilchrist2005PRA_processes}.

We now ask how well one can recover the logical state after the
encoding and the application of the noise channel.  That is, we seek a
completely positive map $\mathcal{R}_{B\to L}$ (the \emph{recovery map}), such
that $\mathcal{R}_{B\to L}\circ\mathcal{N}_{A\to B}\circ\mathcal{E}_{L\to A}$ is
as close as possible to the identity channel $\IdentProc[L][L]{}$.  The
resilience of a code $\mathcal{E}_{L\to A}$ to errors caused by a noise map
$\mathcal{N}_{A\to B}$ is thus quantified by the proximity to the identity
channel of the combined process
$\mathcal{R}_{B\to L}\circ\mathcal{N}_{A\to B}\circ\mathcal{E}_{L\to A}$ for the
best possible recovery map $\mathcal{R}_{B\to L}$.
Using either the entanglement fidelity with fixed input $\ket\phi_{LR}$ or the
worst-case entanglement fidelity measures, the quality of the code
$\mathcal{E}_{L\to A}$ under the noise $\mathcal{N}_{A\to B}$ is quantified as
\begin{subequations}
  \label{eq:def-aqecc-f-maxR}
  \begin{align}
    f_{\mathrm{e}}(\mathcal{N}\circ\mathcal{E}) &= \max_{\mathcal{R}_{B\to L}}
    F_{\mathrm{e}}(\mathcal{R}\circ\mathcal{N}\circ\mathcal{E})\ ;
    \label{eq:def-aqecc-f-fixedinput-maxR}
    \\
    f_{\mathrm{worst}}(\mathcal{N}\circ\mathcal{E}) &=
    \max_{\mathcal{R}_{B\to L}}
    F_{\mathrm{worst}}(\mathcal{R}\circ\mathcal{N}\circ\mathcal{E})\ .
    \label{eq:def-aqecc-f-worst-maxR}
  \end{align}
\end{subequations}
We will also find it convenient to work with the alternative quantities
\begin{subequations}
  \begin{align}
    \epsilon_{\mathrm{e}}(\mathcal{N}\circ\mathcal{E})
    &= \sqrt{1 - f_{\mathrm{e}}^2(\mathcal{N}\circ\mathcal{E})}\ ;
      \label{eq:def-aqecc-epsilon-fixedinput}
    \\
    \epsilon_{\mathrm{worst}}(\mathcal{N}\circ\mathcal{E})
    &= \sqrt{1 - f_{\mathrm{worst}}^2(\mathcal{N}\circ\mathcal{E})}\ ,
      \label{eq:def-aqecc-epsilon-worst}
  \end{align}
\end{subequations}
which are closely related to the infidelity and Bures distance measures.  A code
which performs well has $f\approx 1$ and correspondingly $\epsilon\approx 0$.

\subsection{Erasures at known locations}
In this work, we consider the noise model consisting of erasures which occur at
known locations.  (Our bound then naturally applies also to erasures at unknown
locations, since the latter are necessarily harder to correct against.)  For
instance, if the $i$th physical subsystem is lost to the environment with
probability $q_i$, then the corresponding noise map is
\begin{align}
    \mathcal{N}^{(1)}_{A\to AC}(\cdot)
  = \sum q_i\, \proj{i}_C\otimes \proj{\phi_i}_{A_i}\otimes \tr_{A_i}(\cdot)\ ,
  \label{eq:noise-map-one-erasure}
\end{align}
where we have introduced a classical register $C$ which records which one of the
$n$ systems was lost, and where ${\ket\phi_i}$ are some fixed states.

One can also consider more general erasure scenarios, where any given
combination of subsystems can be lost with a given probability.  For instance,
one might assume that systems $A_1$ and $A_2$ are simultaneously lost with
probability $q_{\{1,2\}}$, systems $A_2$ and $A_3$ are simultaneously lost with
probability $q_{\{2,3\}}$, and systems $A_1$ and $A_3$ are lost with probability
$q_{\{1,3\}}$.  More generally, a combination of subsystems, which we label
generically by $\alpha$, can be lost with probability $q_\alpha$; we assume we
know exactly which systems were lost.  The corresponding general noise map is
then
\begin{subequations}
  \label{eq:noise-map-general-alpha}
  \begin{align}
  \mathcal{N}_{A\to AC}(\cdot)
  &= \sum_{\alpha\in K} q_\alpha
    \proj{\alpha}_C \otimes \mathcal{N}_{A\to A}^{\alpha}(\cdot)\ ;
    \\
  \mathcal{N}_{A\to A}^{\alpha}(\cdot)
  &= \proj{\phi_\alpha}_{A_\alpha} \!
    \otimes \tr_{A_\alpha}(\cdot)\ ,
    \label{eq:noise-map-general-alpha--Nalpha}
  \end{align}
\end{subequations}
where the register $C$ encodes the exact locations at which simultaneous
erasures have occurred, where $A_\alpha$ denotes the physical systems labeled by
$\alpha$ (for instance, if $\alpha=\{2,3\}$ then $A_\alpha = A_2\otimes A_3$),
and where $\{ \ket\phi_\alpha \}$ are fixed states. The sum ranges over
a set $K$ of possible $\alpha$'s corresponding to erasures which may occur.
Technically, $K$ is any set of subsets of $\{1, 2, \ldots, n\}$.  Situations
which can be described using this setting include for instance any $k$
consecutive erasures, or the erasure of any $k$ subsystems.

\subsection{Characterization via the environment}
A very useful characterization of the quantities~\eqref{eq:def-aqecc-f-maxR} is
provided by B\'eny and Oreshkov~\cite{Beny2010PRL_AQECC}, building upon the
decoupling approach to error correction~\cite{Hayden2008OSID_decoupling}.  The
recoverability of the logical information can be characterized by studying how
much information is leaked to the environment, as represented by a complementary
channel $\widehat{\mathcal{N}\circ\mathcal{E}}$ of
$\mathcal{N}\circ\mathcal{E}$.  Recall that a \emph{complementary channel}
$\hat{\mathcal{F}}_{A\to C}$ of a quantum channel $\mathcal{F}_{A\to B}$ is a
channel of the form
$\hat{\mathcal{F}}_{A\to C}(\cdot) = \tr_B(W_{A\to BC}(\cdot) W^\dagger)$, where
$W_{A\to BC}$ is a Stinespring dilation isometry for the map $\mathcal{F}$,
i.e., $\mathcal{F}_{A\to B}(\cdot) = \tr_C(W_{A\to BC}(\cdot) W^\dagger)$.
B\'eny and Oreshkov show that the fidelity with which one can reverse the
action of the encoding and the noise is exactly the fidelity of the total
complementary channel to a constant channel:
\begin{subequations}
  \label{eq:Beny-Oreshkov-f}
  \begin{align}
    f_{\mathrm{e}}(\mathcal{N}\circ\mathcal{E})
    &= \max_{ \zeta } F_{\ket{\hat\phi}}(
      \widehat{\mathcal{N}\circ\mathcal{E}}, \mathcal{T}_\zeta )\ ;
      \label{eq:Beny-Oreshkov-f-fixedinput}
    \\
    f_{\mathrm{worst}}(\mathcal{N}\circ\mathcal{E})
    &= \max_{ \zeta } \min_{\ket\phi} F_{\ket\phi}(
      \widehat{\mathcal{N}\circ\mathcal{E}},
      \mathcal{T}_\zeta )\ ,
      \label{eq:Beny-Oreshkov-f-worst}
  \end{align}
\end{subequations}
where $\mathcal{T}_\zeta(\cdot) = \tr(\cdot)\,\zeta$ is the constant channel
outputting the state $\zeta$ and where the maximizations range over all quantum
states $\zeta$ on the output system of $\widehat{\mathcal{N}\circ\mathcal{E}}$.

Now we determine a complementary channel $\widehat{\mathcal{N}\circ\mathcal{E}}$
to the encoding and noise channels.  Consider first the single-erasure noise
channel~\eqref{eq:noise-map-one-erasure}.  A Stinespring dilation of
$\mathcal{N}^{(1)}_{A\to AC}$ on two additional systems $C'\otimes E$ is given
as $\mathcal{N}^{(1)}_{A\to AC} = \tr_{C'E}(W (\cdot) W^\dagger)$, with
\begin{multline}
  W_{A\to ACC'E} =
  \\
  \sum \sqrt{q_i}\,\ket{i}_C \otimes \ket{i}_{C'} \otimes \ket{\phi_i}_{A_i}
  \otimes \Ident_{A_i\to E} \otimes \Ident_{A\setminus A_i}\ ,
\end{multline}
where $\Ident_{A_i\to E}$ is an isometric embedding of $A_i$ into $E$ and
$\Ident_{A\setminus A_i}$ is the identity operator on all systems $A$ except
$A_i$.  Now consider a Stinespring dilation of $\mathcal{E}_{L\to A}$ as
$\mathcal{E}_{L\to A} = \tr_{F}(V_{L\to AF}\,(\cdot)\,V^\dagger)$.  Then, we
may take
\begin{align}
  \widehat{\mathcal{N}\circ\mathcal{E}}_{L\to C'EF}(\cdot)
  &= \tr_{AC}(W\, V\, (\cdot)\, V^\dagger\, W^\dagger)
    \nonumber\\
  &= \sum q_i \, \proj{i}_{C'}\otimes 
    \tr_{A\setminus A_i}(V(\cdot)V^\dagger) \ ,
    \label{eq:compl-channel-N-E-one-erasure}
\end{align}
where $\tr_{A\setminus A_i}$ denotes the partial trace over all systems except
$A_i$ (the latter is then embedded in the $E$ system).
Hence, the complementary channel to the single erasure channel simply gives the
erased information to the environment with the corresponding erasure
probability.  It is straightforward to see that for the more general noise
channel~\eqref{eq:noise-map-general-alpha} a complementary channel is given by
\begin{align}
  \widehat{\mathcal{N}\circ\mathcal{E}}_{L\to C'EF}(\cdot)
  &= \sum q_\alpha \, \proj{\alpha}_{C'}\otimes 
    \tr_{A\setminus A_\alpha}(V\,(\cdot)\,V^\dagger) \ ,
    \label{eq:compl-channel-N-E-general-alpha}
\end{align}
where the register $C'$ now remembers which combination of systems were lost.
This channel provides the environment with the systems that were erased, where
each erasure combination $\alpha$ appears with probability $q_\alpha$.

\subsection{Covariant codes}

The final ingredient we introduce is covariance with respect to a symmetry group
(\cref{fig:covariant-code}).  Let $G$ be any Lie group acting unitarily on
the logical and physical systems, with representing unitaries $U_L(g)$ and
$U_A(g)$, respectively, for any $g\in G$.  A code $\mathcal{E}_{L\to A}$ is
\emph{covariant} if it commutes with the group action:
\begin{align}
  \mathcal{E}_{L\to A}`\big( U_L(g)\,(\cdot)\,U_L^\dagger(g) )
  = 
  U_A(g) \, \mathcal{E}_{L\to A}(\cdot)\,U_A^\dagger(g) \ .
  \label{eq:condition-covariant-code-cpm}
\end{align}

On either logical and physical systems, we can expand the unitary action of $G$
in terms of generators of the corresponding Lie algebra, i.e., for a given $g$
there is a generator $T_L$ on $L$ and a generator $T_A$ on $A$ such that
\begin{align}
  U_L(g) &= \ee^{-i\theta T_L}\ ;
  &
    U_A(g) &= \ee^{-i\theta T_A}\ ,
\end{align}
for some $\theta\in\mathbb{R}$ that we can choose to normalize our generators.
The generators are Hermitian matrices, and they can be interpreted as physical
observables.  (For instance, the generators of the rotations in 3-D space are
the angular momenta.)

If the encoding map is isometric,
$\mathcal{E}_{L\to A}(\cdot) = V_{L\to A}\,(\cdot)\, V^\dagger$, then any
eigenstate $\ket{t}$ of $T_L$ with eigenvalue $t$ must necessarily be encoded
into an eigenstate of $T_A$ with the same eigenvalue $t$ (up to a constant
offset).  This can be seen as follows.  Expanding the
condition~\eqref{eq:condition-covariant-code-cpm} for small $\theta$ yields
\begin{align}\label{eq:expanded_covariance}
  V\,[T_L, (\cdot)]\, V^\dagger = [T_A, V `*(\cdot) V^\dagger]\ .
\end{align}
Let $\{ \ket{t,j}_L \}$ be a basis of eigenstates of $T_L$ where $t$ is the eigenvalue and
where $j$ is a degeneracy index.  Inserting in the place of $(\cdot)$ the operator
$\ketbra{t,j}{t',j'}$, we obtain
\begin{align}
  (t-t')\ketbra{\psi_{t,j}}{\psi_{t',j'}} = [T_A,\ketbra{\psi_{t,j}}{\psi_{t',j'}}]\ ,
  \label{eq:condition-covariant-code-derivation-commutator-L-A}
\end{align}
where $\ket{\psi_{t,j}} = V\ket{t,j}$.  Setting $t=t', j=j'$, we see that
$\ket{\psi_{t,j}}$ is necessarily an eigenstate of $T_A$; let $u_{t,j}$ be its
corresponding eigenvalue.  Setting $t=t',j\neq j'$ in
\eqref{eq:condition-covariant-code-derivation-commutator-L-A} implies
$0 = (u_{t,j}-u_{t,j'})\ketbra{\psi_{t,j}}{\psi_{t,j'}}$ and hence
$u_{t,j} = u_{t,j'} =: u_t$.  Now
\eqref{eq:condition-covariant-code-derivation-commutator-L-A} tells us for any $t,t',j,j'$
that $t - t' = u_t - u_{t'}$.  It follows that $u_t = t - \nu$ for all $t$, for some
constant offset $\nu$; in other words, the codewords must have the same charge as the
logical state, except for a possible constant offset $\nu$.  We may condense this
condition into the constraint $[T_A, VV^\dagger]=0$ along with the identity
\begin{align}
  V^\dagger T_A V = T_L - \nu\Ident_L\ .
  \label{eq:condition-covariant-code-isometry-charge}
\end{align}
Equivalently, acting with $V$ on
\eqref{eq:condition-covariant-code-isometry-charge} we have
\begin{align}
  T_A V =V (T_L - \nu\Ident).
  \label{eq:covarianve_and_generators}
\end{align}
This is a crucial property of covariant codes, and is a central ingredient of
the proof of our main result.

Our main result, in its simplified form, further assumes that the action of the
group is transversal on the physical systems, meaning that
$U_A(g) = U_1(g)\otimes U_2(g)\otimes \cdots U_n(g)$.  In this case, the
corresponding generator is strictly local, $T_A = T_1 + T_2 + \cdots + T_n$,
where each of the $T_i$'s act only on
$A_i$.

As opposed to covariant isometries, covariant channels in general do not
conserve charge since they may exchange charge with the environment.  For
instance, the fully depolarizing channel is covariant with respect to any
symmetry but it changes the charge of its input.  Our main result in its fully
general form is formulated for approximately charge-conserving channel
encodings, which is a superset of covariant isometries.


%% file: ResultsBounds.tex
Our first main result is a general characterization of how poorly a code
necessarily performs against erasures at known locations, given that the code
must be covariant with respect to a continuous symmetry.
For the sake of clarity, we first present a simplified version of our general
bound.  Consider an encoding map $\ket{x}_L\to\ket{\psi_x}_{A}$ with respect to
some basis $\{ \ket{x}_L \}$, which we may represent by an isometry
$V_{L\to A} = \sum_x \ket{\psi_x}_A\bra{x}_L$.  Denote the corresponding
encoding channel by $\mathcal{E}_{L\to A}(\cdot) = V(\cdot) V^\dagger$.

Pick any generator $T_L$ from the Lie algebra of the symmetry acting on $L$.
Let $T_{i}$ be the corresponding generator acting on the $i$th physical
subsystem $A_i$, with the total generator on $A$ being $T_A = \sum_i T_{i}$.  As
Hermitian matrices, these are quantum mechanical observables whose eigenvalues
we may think of as abstract ``charges.'' (These charges might
correspond to the component of angular momentum in a given direction, the number of particles, or
some other physical quantity.)  Crucially, since the code $\mathcal{E}$ is
covariant, a logical charge eigenstate $\ket{t}_L$ must be encoded into a
codeword $\ket{\psi_t}_A$ which is an eigenvector of $T_A$ with the same
eigenvalue $t$, up to a constant offset $\nu$.  Let us assume for simplicity
that $\nu=0$.

Assume the environment erases a subsystem $i$ chosen at random with probability
$q_i = 1/n$.  Then the environment gets the information represented by the
complementary channel~\eqref{eq:compl-channel-N-E-one-erasure}. That is, if the
original state was $\ket{x}_L$, then the environment gets the state
$\rho^x_i = \tr_{A\setminus A_i}(\proj{\psi_x}_L)$ on subsystem $i$ with
probability $1/n$.
Yet, because the charge observable is local, the environment can learn the
expectation value of the charge.  Indeed, for any $\ket{x}_L$,
\begin{align}
  \tr(T_L \proj{x}_L)
  = \tr(T_A \proj{\psi_x}_A)
  = \sum_i \tr(T_i \rho_i^x)\ ,
\end{align}
where the first equality holds because the code is covariant, and the second
because the charge is local.  Hence, if we define the observable
$Z_{C'E} = n \sum_i \proj{i}_{C'}\otimes T_i$ on the environment systems, we
have
\begin{align}
  \tr(T_L \proj{x}_L) =
  \tr(Z_{C'E}\,\widehat{\mathcal{N}\circ\mathcal{E}}(\proj{x}_L))\ ,
\end{align}
making it clear that the environment can measure the average charge using the
information it has available.

Surely, if the charge expectation value leaks to the environment, then the code
must be bad.  However, the accuracy of the code is measured in terms of an
entanglement fidelity (worst-case or fixed input) to the identity channel.
Hence, it still remains to relate the accuracy of the code to the environment's
ability to access the codeword's total charge.
On one hand, we observe that the difference in expectation value of $Z_{C'E}$
on the environment can be translated into a distinguishability of codewords in
terms of the trace distance.  More precisely and in general, for any two states
$\rho,\sigma$, if there is an observable $Q$ for which $\rho,\sigma$ have
different expectation values, then
$\delta(\rho,\sigma) \geqslant \abs{\tr(Q\rho) -
  \tr(Q\sigma)}/(2\norm{Q}_\infty)$.  In our case, consider two logical charge
eigenstates $\ket{\phi_\pm}_L$ corresponding to the maximal and minimal
eigenvalues of $T_L$; then it holds that
\begin{align}
  \delta`\Big(\widehat{\mathcal{N}\circ\mathcal{E}}(\proj{\phi_-}_L),
  \widehat{\mathcal{N}\circ\mathcal{E}}(\proj{\phi_+}_L))
  \geqslant \frac{ \Delta T_L }{2\,\norm{Z_{C'E}}_\infty}\ ,
\end{align}
where $\Delta T_L$ is the spectral range of $T_L$, i.e., the difference between
the maximal and minimal eigenvalue of $T_L$.  We assume here for simplicity that
the maximal and minimal eigenvalues of $T_i$ are equal in magnitude, such that
$\Delta T_i = 2\norm{T_i}_\infty$; hence
$2\,\norm{Z_{C'E}}_\infty = 2n\max_i \norm{T_i}_\infty = n \max_i\Delta T_i$.
On the other hand, if the environment's states are distinguishable for
different codewords, then the accuracy of the code is bad; specifically, we show
in the Appendix (\cref{lemma:aqecc-environtrdist}) that for any two logical
states $\ket{x}_L, \ket{x'}_L$, we have
\begin{align}
  \epsilon_{\mathrm{worst}}(\mathcal{N}\circ\mathcal{E})
  \geqslant
  \frac12
  \delta\Bigl( \widehat{\mathcal{N}\circ\mathcal{E}}(\proj{x}_L),
  \widehat{\mathcal{N}\circ\mathcal{E}}(\proj{x'}_L) \Bigr)\ .
\end{align}
Finally, we have proven our simplified main result.
\begin{theorem}
  \label{thm:simple-main-result-one-erasure}
  \noproofref%
  The performance of the covariant code $\mathcal{E}(\cdot)=V(\cdot)V^\dagger$ under the above
  assumptions, quantified by the worst-case entanglement fidelity, is bounded as
  follows:
  \begin{align}
  \epsilon_{\mathrm{worst}}(\mathcal{N}\circ\mathcal{E})
  \geqslant
  \frac1{2n} \frac{\Delta T_L}{\max_i \Delta T_i}\ . \label{eq:simple-bound}
\end{align}
\end{theorem}
A similar analysis leads to a bound for the figure of merit
$\epsilon_{\mathrm{e}}$ based on the average entanglement fidelity,
\begin{align}
  \epsilon_{\mathrm{e}}(\mathcal{N}\circ\mathcal{E})
  \geqslant
  \frac1{n} \frac{\norm{T_L - \tr(T_L) \Ident_L/d_L}_1/(2d_L)}{\max_i \Delta T_i} \ ,
  \label{eq:simple-bound-e}
\end{align}
The right hand side of~\eqref{eq:simple-bound-e} is simply a different measure
for the spread of eigenvalues; unlike $\Delta T_L$, it takes contributions from
all eigenvalues of $T_L$.  The argument of the norm is simply the charge
operator $T_L$ with a global shift that makes the operator traceless.
\Cref{eq:simple-bound-e} is proven as a special case of
\cref{thm:main-result-full}.

In \cref{appx:corrbound}, we provide an alternative proof for the
bound~\eqref{eq:simple-bound-e} using a different approach: We quantify the
information leaked to the environment by studying the connected correlation
functions between the subsystems.  In fact, we lower bound the sum of the
correlation functions between the logical qubit and individual physical
subsystems, and since this total correlation is non-zero, we deduce that the
environment is correlated with the logical information, which translates to an
upper bound on the fidelity of recovery.

In short, a covariant code with respect to a local charge may not perform
well for correcting a single erasure at a known location, unless it either
encodes the information into large physical systems, with a large range of
possible charge values ($\max_i \Delta T_{i} \to\infty$), or it encodes the
information into many physical systems ($n\to\infty$).

The following theorem generalizes \cref{thm:simple-main-result-one-erasure} in a
number of ways.  It allows for the code to only approximately conserve charge,
considers erasures affecting multiple systems with arbitrary erasure
probabilities, and does not require the charge to be strictly local; finally, it
can be applied in situations in which the codewords have most of their weight on
a finite charge range (but may have distribution tails extending to arbitrarily
large charge values).  The setting of \cref{thm:main-result-full} is depicted in
\cref{fig:main-result-thm-setting}.
\begin{figure}
  \centering
  \includegraphics{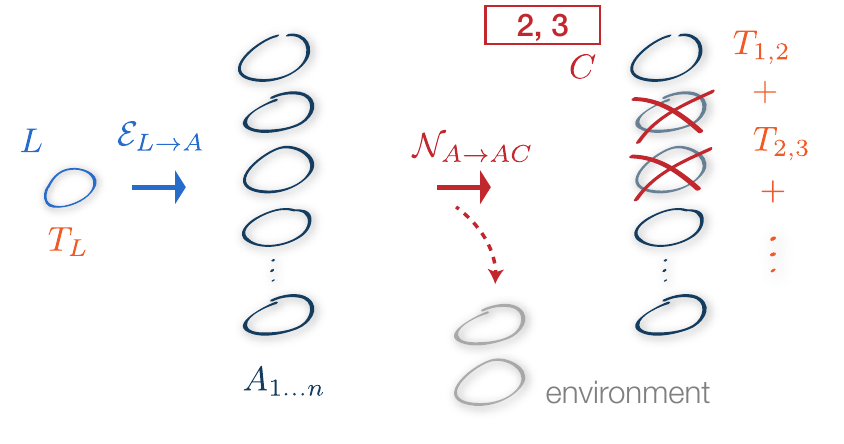}
  \caption{The general setting of \cref{thm:main-result-full}.  An
    approximately charge-conserving encoding maps a logical state onto
    several physical systems.  The noise acts by randomly erasing some
    subsystems, and storing which systems were erased in a register $C$.  Which
    combinations of subsystems can be lost and with which probability can be
    chosen arbitrarily.  The continuous symmetry is assumed to have a generator
    represented by $T_L$ on the logical system and by $T_A$ on the physical
    systems. We assume that $T_A$ can be written as a sum of terms
    $T_A = \sum T_\alpha$, where each $T_\alpha$ acts on a combinations of
    subsystems that could possibly be lost to the environment.  For instance,
    $T_A$ may include a term $T_{3,4,7}$ acting on systems $A_3 A_4 A_7$ only if
    the noise model is such that the systems $3,4,7$ have a nonzero probability
    of being simultaneously erased.}
  \label{fig:main-result-thm-setting}
\end{figure}
\begin{theorem}
  \label{thm:main-result-full}
  Let $L$ and $A=A_1\otimes\cdots\otimes A_n$ be the logical and physical
  systems, respectively, and let $\mathcal{E}_{L\to A}$ be any completely
  positive, trace-preserving map.  Consider logical and physical observables
  $T_L$ and $T_A$.  We assume that:
  \begin{enumerate}[label=(\alph*)]
  \item\label{item:main-thm-general-condition-approx-charge-conserving} There is
    a $\nu\in\mathbb{R}$ and a $\delta\geqslant 0$ such that
    $\norm{(T_L - \nu\Ident_L) - \mathcal{E}^\dagger(T_A)}_\infty \leqslant
    \delta$, i.e., the code is approximately charge-conserving up to a constant
    shift;
  \item We can write $T_A = \sum_{\alpha} T_\alpha$, where each term
    $T_\alpha$ acts on a subset of physical systems labeled by $\alpha$;
  \item Fixing cut-offs $t_{\alpha}^\pm$ for each $\alpha$, there is
    $\eta\geqslant 0$ such that for any state $\sigma_L$, we have
    \begin{align}
      \abs*{\tr`*( \sum (T_\alpha -
      t_\alpha)\Pi_\alpha^\perp\;\mathcal{E}(\sigma_L))} \leqslant \eta\ ,
      \label{eq:main-result-full-condition-charge-cutoffs}
    \end{align}
    where $\Pi_\alpha^\perp$ projects onto the eigenspaces of $T_\alpha$ whose
    eigenvalues are outside $\intervalc{t_\alpha^-}{t_\alpha^+}$, and where
    $t_\alpha = (t_\alpha^-+t_\alpha^+)/2$.  That is, when chopping off parts
    of the codewords exceeding charge $t_\alpha^\pm$ on term $\alpha$ and
    shifting the charge term to center it around zero, the total average
    charge chopped off does not exceed $\eta$;
  \item The noise acts as per~\eqref{eq:noise-map-general-alpha} by erasing
    subsystems labeled by $\alpha$ with probability $q_\alpha>0$, for each
    $\alpha$ for which there is a corresponding term $T_\alpha$ in the global
    generator $T_A$.
  \end{enumerate}
  Then the accuracy of the code $\mathcal{E}_{L\to A}$ against the noise
  $\mathcal{N}$ is bounded as
  \begin{subequations}
    \begin{align}
      \left.\begin{array}{r}
              \epsilon_{\mathrm{e}}(\mathcal{N}\circ\mathcal{E})
              \\[1ex]
              \avg[\big]{ \epsilon_{\mathrm{e}}(\mathcal{N}^{\alpha}\circ\mathcal{E}) }_\alpha
            \end{array}\right\}
      &\geqslant  \frac{\norm{T_L - \mu(T_L)\Ident_L}_1/d_L - \delta - \eta}%
        {\max_\alpha (\Delta T_\alpha / q_\alpha) }\ ;
        \label{eq:main-thm-bound-avg-entgl-fid}
      \\
      \epsilon_{\mathrm{worst}}(\mathcal{N}\circ\mathcal{E}) 
      &\geqslant \frac{\Delta T_L / 2 - \delta - \eta}%
        {\max_\alpha (\Delta T_\alpha / q_\alpha) }\ ,
        \label{eq:main-thm-bound-worst-entgl-fid}
    \end{align}
  \end{subequations}
  where $\Delta T_\alpha = t_\alpha^+ - t_\alpha^-$, where $d_L$ is the
  dimension of $L$, where $\avg{\cdot}_\alpha = \sum_\alpha q_\alpha (\cdot)$,
  and where $\mu(T_L)$ is a median eigenvalue of $T_L$. We define a median
  eigenvalue of $T_L$ to be a number $\mu$ such that the length-$d_L$ vector of
  eigenvalues of $T_L$ counted with multiplicity has at least
  $\lceil d_L/2\rceil$ components that are less than or equal to $\mu$, and at
  least $\lceil d_L/2\rceil$ components that are greater than or equal to $\mu$.

  Additionally, the first term in the numerator
  of\/~\eqref{eq:main-thm-bound-avg-entgl-fid} may be replaced by
  $\norm{T_L - \tr(T_L)\Ident/d_L}_1/(2d_L)$.
\end{theorem}

The bound~\eqref{eq:main-thm-bound-avg-entgl-fid} is intuitively sensitive to
the ``average amount of logical charge'' in absolute value, up to an arbitrary
charge offset; this makes sense since the average entanglement fidelity ``only
samples the average case.''  On the other hand, the worst-case entanglement
fidelity picks up the worst possible situation, noticing that there are two
states with maximally different charges; the
bound~\eqref{eq:main-thm-bound-worst-entgl-fid} reflects that the code will
perform the worst for those input states.
The median eigenvalue in~\eqref{eq:main-thm-bound-avg-entgl-fid} appears as an
optimal solution to the optimization $\min_\mu \norm{T_L - \mu \Ident_L}_1$.
For an operator $T$ that has the same number of positive eigenvalues as negative
ones (with multiplicity), such as a component of spin, we can set $\mu(T) = 0$.

For isometric encodings,
condition~\ref{item:main-thm-general-condition-approx-charge-conserving} really
means that the encoding is approximately covariant.  However our theorem holds
more generally for encodings that are not an isometry, as long as they
approximately conserve charge.  The latter condition is stricter than being
covariant.  However, an approximately covariant channel encoding that does not
approximately preserve charge can still fit in the context of
\cref{thm:main-result-full}, by explicitly considering instead its covariant
Stinespring dilation~\cite{Scutaru1979RMP_covariant,%
  Keyl1999JMP_cloning,PhDMarvian2012_symmetry,Faist2018arXiv_thcapacity} into an
ancilla system which is then erased by the environment with certainty as part of
the noise channel.

The proof of \cref{thm:main-result-full} is provided in
\cref{appx:ProofBounds}.  The proof is split into two parts.  A first part
shows that there exists an observable accessible to the environment which is
able to infer the global logical charge to a good approximation.  The second
part deduces from the existence of such an observable that the code must
necessarily have limited performance, as quantified by various entanglement
fidelity measures.


%% file: CriterionCode.tex
Here we introduce a criterion that allows us to certify a given encoding as
performing accurately as an approximate error-correcting code against any given
noise channel, as measured by the worst-case entanglement fidelity.
Proving that a code has a good average-case entanglement fidelity (i.e., showing
that $\epsilon_{\mathrm{e}}(\mathcal{N}\circ\mathcal{E})$ is small) is perhaps
comparatively easier, as one can attempt to guess a suitable recovery map for a
maximally entangled input state and directly compute the fidelity of recovery.
The method we present provides an upper bound to the stricter measure
$\epsilon_{\mathrm{worst}}(\mathcal{N}\circ\mathcal{E})$ and does not require us to
come up with explicit recovery procedures.

Intuitively, if we consider erasures at known locations, we can expect that if
all local reduced states of codewords look alike independently of the logical
information, then the code performs well.  That is, if for each individual
subsystem each codeword has the same reduced state, then because the environment
gets access only to those individual reduced states, it obtains no information
about the codeword and the erasure is thus correctable.  This intuition is
correct in the exact case, but in the approximate case the fact that the
entanglement fidelity is defined with a ``stabilization'' over a reference
system poses an additional challenge~\cite{CedricBenyEmail}.  Our solution is to
consider how logical operators of the form $\ketbra{x}{x'}$ are encoded, where
$\{ \ket{x} \}$ is any fixed basis of the logical system.  In the case of a
single erasure at a known location, we define
\begin{align}
  \label{eq:criterion-simple-reduced-state-i-codeword-x-xp}
  \rho_i^{x,x'} = \tr_{A\setminus A_i}(\mathcal{E}(\ketbra{x}{x'}))\ ,
\end{align}
noting that $\rho_i^{x,x'}$ is a quantum state if $x=x'$ but is not even
necessarily Hermitian for $x\neq x'$.  Our criterion then states the following:
If the states $\rho_i^{x,x}$ are approximately independent of $x$, and if each
$\rho_i^{x,x'}$ for $x\ne x'$ has a very small norm, then the code is a good approximate
error-correcting code against erasure of subsystem $i$.
\begin{theorem}
  \noproofref
  \label{prop:criterion-simple-main}
  For any encoding channel $\mathcal{E}$ and for any noise channel
  $\mathcal{N}$, let $\widehat{\mathcal{N}\circ\mathcal{E}}$ be a complementary
  channel of $\mathcal{N}\circ\mathcal{E}$.
  Fixing a basis of logical states $\{ \ket{x} \}$, we define
  \begin{align}
    \rho^{x,x'} = \widehat{\mathcal{N}\circ\mathcal{E}}(\ketbra{x}{x'})\ .
  \end{align}
  Assume that there exists a state $\zeta$, as well as constants
  $\epsilon,\nu \geqslant 0$ such that
  \begin{subequations}
    \begin{align}
      F(\rho^{x,x}, \zeta) &\geqslant \sqrt{1 - \epsilon^2} \\
      \norm{\rho^{x,x'}}_1 &\leqslant \nu \quad\text{ for $x\neq x'$}.
    \end{align}
  \end{subequations}
  Then, the code $\mathcal{E}$ is an approximate error-correcting code with an
  approximation parameter satisfying
  \begin{align}
    \epsilon_{\mathrm{worst}}(\mathcal{N}\circ\mathcal{E})
    \leqslant \epsilon + d_L\sqrt{\nu}\ ,
    \label{eq:criterion-certify-code-epsilon-worst}
  \end{align}
  where $d_L$ is the logical system dimension.

  If one of several noise channels is applied at random but it is known which
  one occurred, then~\eqref{eq:criterion-certify-code-epsilon-worst} holds for
  the overall noise channel if the assumptions above are satisfied for each
  individual noise channel.
\end{theorem}

Note that the criterion holds for any arbitrary noise channel, not only for
erasures at known locations.  The proof of \cref{prop:criterion-simple-main}
is given in \cref{appx:criterion-certify-code}.

Our criterion can be seen as an expression of the approximate Knill-Laflamme
conditions~\cite{Beny2010PRL_AQECC} in a particular basis, but where we provide
simple and practical conditions on how to bound the error parameter
$\epsilon_{\mathrm{worst}}$ of the code.

Our criterion is a sufficient condition for a code to be approximately
error-correcting, but the condition is not necessary. When the
criterion does not apply we cannot draw any conclusion about the code's
performance.

We note that our criterion does not make reference to individual Kraus operators
of the noise channel, as the Knill-Laflamme conditions or their approximate
verisons do~\cite{Knill1997PRA_correction,Beny2010PRL_AQECC}.  This property
eases its application to large-dimensional physical quantum systems.


%% file: ResultsExamplesCovariantCodes.tex
Here we study three classes of covariant codes that illustrate the behavior of
our bound in either regimes of large subsystem dimensions, or large number of
physical subsystems (\cref{tab:examples-codes}).

\subsection{Three-rotor secret-sharing code}
\label{sec:maintext-truncated-Hayden-code}

In this subsection, we apply our criterion to a truncated version of a code
introduced by Hayden \emph{et al.\@}~\cite{Hayden2017arXiv_frame}, linking that
code to the well-known three-qutrit secret-sharing quantum polynomial
code~\cite{Aharonov1997,Cleve1999PRL_secret}.  While illustrating how to use our
criterion, it also provides a covariant code which performs well in the limit of
codewords covering a large range of physical charge on the subsystems.

\subsubsection{Rotor version of the qutrit secret-sharing code}
\label{subsubsec:rotor-version}

\begin{table}
\begin{tabular}{lccc}
\hline\hline
 & Covariance & ~~Dimen.~~ & Error correction\tabularnewline
\hline
$[[3,1,2]]_{\mathbb{Z}}$ &  &  & \tabularnewline
sharp cutoff & $U(1)$ & Finite & Approximate\tabularnewline
smooth cutoff & $U(1)$ & Infinite & Approximate\tabularnewline
\hline
$[[5,1,3]]_{\mathbb{Z}}$ &  &  & \tabularnewline
qudit version & $\mathbb{Z}_{D}$ & Finite & Exact\tabularnewline
smooth cutoff & $U(1)$ & Infinite & Approximate\tabularnewline
\hline
\multicolumn{2}{l}{$[[n,a\log n,b\log n]]$} &  & \tabularnewline
finite $n$ & $U(1)$ & Finite & Approximate\tabularnewline
\hline\hline
\end{tabular}
\caption{Summary of the codes considered in \cref{sec:Uonecodes}: the
  three-rotor secret-sharing code, the five-rotor perfect code, and an $n$-qubit
  ``thermodynamic code'' with codewords consisting of Dicke states (and $a,b$
  chosen appropriately).}
\label{tab:examples-codes}
\end{table}

For our purposes, a quantum rotor (also, an $O(2)$ or planar quantum rotor) is simply a
system with a basis $`{\ket{x}}$ that is labeled by an integer $x\in\mathbb{Z}$
indexing representations of $U(1)$~\cite{Albert2017}.  Consider the three-rotor code given in
Ref.~\cite{Hayden2017arXiv_frame} defined by the isometry from $L$ to
$A=A_1\otimes A_2\otimes A_3$ given as
\begin{align}
  V_{L\to A} :\quad \ket{x}_L ~\to~ \sum_{y\in\mathbb{Z}} \, \ket{-3y,y-x,2(y+x)}_A
  \ ,
  \label{eq:threerotor}
\end{align}
where the states $`{\ket{x}}$ are eigenstates of the angular momentum operators
$T_L$ and $T_{A} = T_1 + T_2 + T_3$.  This code can correct against the loss of
any of the three subsystems~\cite{Hayden2017arXiv_frame}.  Moreover, the code is
covariant with respect to the charge $T$: A logical state $\ket{x}_L$ is mapped
onto a codeword with the same total charge $x$.

Interestingly, this code is a natural rotor generalization of the three-qutrit secret-sharing code~\cite{Aharonov1997,Cleve1999PRL_secret}.
The three-qutrit code maps the basis vectors $\ket{j}_L$ ($j=0,1,2$) of a
logical qutrit into the codewords $\sum_{k} \ket{k,k-j,k+j}$ where the addition
is modulo $3$.  Now, substitute each qutrit subsystem with a rotor. 
We obtain a code defined by the following encoding map:
\begin{align}
  \tilde{V}_{L\to A} :\quad
  \ket{x}_L ~\to~ \sum_{y\in\mathbb{Z}} \, \ket{y,y-x,y+x}\ .
  \label{eq:three-rotor-secret-sharing-notcovariant}
\end{align}

This code is not yet covariant with respect to the charge states $\ket{x}$, as
the charge of the codeword corresponding to $\ket{x}_L$ is not $x$.  However, we
may apply the isometry mapping $\ket{(.)}\to\ket{-3(.)}$ on the first rotor and
$\ket{(.)}\to\ket{2(.)}$ on the second, yielding the encoding
map~\eqref{eq:threerotor}.  (In fact, the
code~\eqref{eq:three-rotor-secret-sharing-notcovariant} is covariant
with respect to a different physical charge generator,
$T'_A = -3T_1 + T_2 + 2T_3$, whereas the code~\eqref{eq:threerotor} is covariant
with respect to the natural physical charge carried by three rotors,
$T_A = T_1+T_2+T_3$.)
In this sense, the code~\eqref{eq:threerotor} is a natural $\UU(1)$-covariant
generalization of the qutrit secret-sharing code.

In the following sections, we address the problem that the codewords
in~\eqref{eq:threerotor} are not normalizable, by building suitable wave packet
states.
We normalize the codewords in two different ways: the sharp cutoff selects a
range of charges to use for each rotor and discards the rest, while the smooth
cutoff imposes a Gaussian envelope on each rotor, thereby keeping all the states
but making them less prominent as the charge
increases~\cite{Gottesman2001PRA_oscillator} (see also related recent
work~\cite{AlvaroMischa-inprep}).  Our noise model is one single erasure at a
known location with probabilities $q_1,q_2,q_3=1/3$, as given
by~\eqref{eq:noise-map-one-erasure}.

\subsubsection{Sharp cutoff}
\label{subsec:sharp}

Let us now truncate the logical system $L$ to a dimension of $2h+1$ for some
fixed $h$, so the charge with respect to which the system is $\UU(1)$-covariant
becomes $T_L = \sum_{x=-h}^{h} x \proj{x}_L$. The physical subsystems are
truncated in turn to $2m+1$ dimensions, so there are in total two parameters
$\{h,m\}$ that determine the ranges of the logical and physical
charges. Normalizing the codewords, the isometry
becomes
\begin{align}
  V^{(m)}_{L\to A} :\quad
  \ket{x}_L ~\to~ \frac1{\sqrt{2 m +1}}
  \sum_{y=-m}^{+m}
  \ket{ -3y, y-x, 2(x+y)}\ ,
\end{align} 
for $x\in{-h,\cdots,h}$.

Since the code is covariant and finite-dimensional, it does not allow for
perfect error-correction.  In \cref{appx:AppendixCalcCodes}, we show that the
code has an accuracy parameter which satisfies
\begin{align}
    \epsilon_{\mathrm{worst}}(\mathcal{N}\circ\mathcal{E}^{(m)})
  \lesssim
  \sqrt2 \sqrt{\frac{h}{m}}\ .
  \label{eq:fijdskafnodajf}
\end{align}
By comparison, our bound~\eqref{eq:simple-bound} in this case reads
\begin{align}
  \epsilon_{\mathrm{worst}}(\mathcal{N}\circ\mathcal{E}^{(m)})
  \geqslant \frac12 \frac{\Delta T_L}{\max_i q_i^{-1} \Delta T_i}
  \approx \frac1{18}\,\frac{h}{m}\ .
\end{align}
There is a difference of a square root between the scaling of our actual code
performance and of our bound.  This is due to switching between the trace
distance and a fidelity-based distance in both of our bounds, and in the way we
have applied our criterion to derive~\eqref{eq:fijdskafnodajf}.

\subsubsection{Smooth cutoff}
\label{subsec:three-rotor-smooth-cutoff}

We now consider a different approach to normalizing the codewords: by using a
Gaussian envelope we can achieve a ``smoother'' cut-off in contrast to the sharp
cut-off considered above (such an envelope is known to be optimal for
finite-sized quantum clocks~\cite{Woods2019AHP_autonomous}).
We impose an envelope controlled by a parameter $w>0$ on the code states to make
them normalizable.  The encoding isometry $V_{L\to A}^{(w)}$ now acts as
\begin{align}
  \ket{x}_L \to \frac1{\sqrt{c_w}}
  \sum_{y=-\infty}^{\infty} \ee^{-\frac{y^2}{4w^2}} \,
  \ket{-3y, y-x, 2(x+y)}\ ,\label{eq:threerotorsmooth}
\end{align}
with a normalization factor $c_w = \sum_{y=-\infty}^\infty \ee^{-y^2/(2w^2)}$.
Note that the envelope does not disturb the symmetry---the code remains
covariant since all of the basis states used to write each logical state still
have the same charge.  We still consider a $(2h+1)$-dimensional logical system
$L$ in order to see how the bound scales in terms of $h/w$.  Deferring calculations
to \cref{appx:AppendixCalcCodes}, the present code has an accuracy parameter
satisfying
\begin{align}
  \epsilon_{\mathrm{worst}}(\mathcal{N}\circ\mathcal{E}^{(w)})
  \leqslant \sqrt{1 - \ee^{-\frac{h^2}{4w^2}}}
  \approx \frac{h}{2w}\ \label{eq:lowerbound3rot}.
\end{align} 
Our bound~\eqref{eq:simple-bound} in this case reads
\begin{align}
  \epsilon_{\mathrm{worst}}(\mathcal{N}\circ\mathcal{E}^{(w)})
  \gtrsim \frac{h/w}{12\sqrt{2\ln(w/h)}}\ ,
\end{align}
where we have kept only the first order in $h/w$, and where the logarithmic term
results from cutting off the infinite tails of our codewords.  Hence, we see
that the present code achieves approximately the scaling of our bound, as both
expressions scale as $h/w$ up to a logarithmic factor.

We may ask for the reason of the discrepancy in the accuracy between the
sharp and the smooth cut-off versions of our code.  For the sharp cut-off, the
error parameter scales as $\epsilon_{\mathrm{worst}} \sim \sqrt{h/w}$, while for
the smooth cut-off it scales approximately as
$\epsilon_{\mathrm{worst}} \sim h/w$.  This can be explained from the following
property of the infidelity.  Loosely speaking, the error parameter
$\epsilon_{\mathrm{worst}}$ is related to how much the local reduced state on a
single system varies as a function of the logical state, as measured in terms of
the infidelity [this can be seen from~\eqref{eq:Beny-Oreshkov-f}].  While in
both normalized versions of the above code, using either the sharp or the smooth
cut-off, we are careful to ensure that all codewords are close to each other, it
turns out that codewords with a sharp cut-off are in a regime where the
infidelity is more sensitive to differences than the smooth cut-off.  This is
because those codewords have incompatible supports.  More precisely, for any
state $\rho$, the infidelity $\sqrt{1 - F^2(\rho,\rho+\varepsilon X)}$, for
a small perturbation $\rho\to \rho+\varepsilon X$, can grow like the square root
of $\varepsilon$ if $\rho+\varepsilon X$ has overlap outside of the support of
$\rho$, while it grows linearly in $\varepsilon$ in well-behaved cases.  The
sharp cut-off belongs to the former regime, while in the case of the smooth
cut-off the infidelity is better behaved.

\subsection{Five-rotor perfect code}
\label{subsec:fivecode}

Here, we provide a rotor extension of the five-qubit perfect
code~\cite{Bennett1996PRA_MSEntglQECorr,Laflamme1996} that can be tiled to
construct holographic codes~\cite{Pastawski2015JHEP_holographic}.  While
qudit~\cite{Chau1997PRA_five} and oscillator~\cite{Braunstein1998PRL_CV}
extensions have been considered, a rotor extension is not as straightforward
because one has to take care of preserving the phases in the code states as
needed to error-correct erasures. Our rotor code is the limit of a sequence of
qudit codewords whose constituent phases approach multiples of an
\textit{irrational} number. This same trick has been used to obtain an
irrational magnetic flux via a sequence of rational fluxes in the
two-dimensional electron gas problem~\cite{Hofstadter1976} as well as rotor
limits of other Hamiltonians~\cite{Albert2017}.  This limit is meant to be an
idealization since there is not enough storage space to measure an irrational
number to infinite precision.

Let the dimension $D$ of each of the five physical subsystems be
finite for the qudit $[[5,1,3]]_{\mathbb{Z}_{D}}$ code and infinite
for the rotor $[[5,1,3]]_{\mathbb{Z}}$ code. The general form of
the unnormalized encoding for both codes is

\begin{align}
\ket x & \rightarrow\sum_{j,k,l,m,n\in\mathbb{Z}_{D}}T_{jklmnx}^{(D)}\ket{j,k,l,m,n}.\label{eq:fivecode}
\end{align}
We introduce the rotor code as a limiting case of the qudit code,
obtaining a concise expression for the qudit perfect tensor $T^{(D)}$
in the process. 

\subsubsection{Qudit version}

Consider first the known finite-$D$ case, for which\footnote{This formula was obtained by constructing the codespace projection
out of powers of products of the code stabilizers, applying it to
canonical basis states $\protect\ket{x,0,0,0,0}$, and calculating
the overlap of the resulting codeword with basis states $\protect\ket{j,k,l,m,n}$.} 
\begin{equation}
T_{jklmnx}^{(D)}=\delta_{x,j+k+l+m+n}^{(D)}\omega^{jk+kl+lm+mn+nj}\,,\label{eq:tensor}
\end{equation}
where $\delta_{a,b}^{(D)}=1$ if $a=b$ modulo $D$ and $\omega$
is a primitive $D$-th root of unity. Notice how the above expression
makes the cyclic permutation symmetry naturally manifest. The delta
function encodes the state label $x$ into the sum of the physical
qudit variables, with the key difference from the sharply-cutoff $[[3,1,2]]_{\mathbb{Z}}$
code being that the sum is modulo $D$. This property makes this code
exactly error-correcting and \textit{not} covariant with respect to
a $U(1)$ symmetry. Instead, this code is covariant with respect to
a $\mathbb{Z}_{D}$ symmetry generated by $Z^{\otimes5}$, where $Z=\sum_{k\in\mathbb{Z}_{D}}\omega^{k}\proj k$
is the qudit Pauli matrix.

\subsubsection{Smooth cutoff}

To take the qudit-to-rotor limit, pick $\omega=\exp(2\pi iL/D)$ with
incommensurate integers $L,D\rightarrow\infty$ such that $L/D$ approaches
a positive irrational number $\Phi$. The indices in \cref{eq:fivecode}
now range over $\mathbb{Z}$,
\begin{equation}
T_{jklmnx}^{(\infty)}=\delta_{x,j+k+l+m+n}e^{2\pi i\Phi(jk+kl+lm+mn+nj)}\,,\label{eq:tensor2}
\end{equation}
and $\delta$ is the usual Kronecker delta function. The final
ingredient is to normalize the states, which can be done via a sharp
or a smooth cutoff as in the $[[3,1,2]]_{\mathbb{Z}}$ code. We perform
the latter using a cyclically-symmetric Gaussian envelope with spread
$w$, prepending $\exp[-\frac{1}{4w^{2}}(j^{2}+k^{2}+l^{2}+m^{2}+n^{2})]$
to the tensor $T_{jklmn}^{(\infty)}$ in \cref{eq:fivecode} and then
normalizing the codewords. The resulting code is covariant with respect
to a $U(1)$ symmetry generated by the total physical charge $T_{A}=\sum_{i=1}^{5}T_{A_{i}}$,
analogous to the three-rotor code~\eqref{eq:threerotorsmooth}. With the addition of the envelope, the resulting tensor becomes approximately perfect. This rotor version can be stacked to form a approximately error-correcting $U(1)$-covariant holographic code in the same way as the qubit perfect tensors were connected in Ref.~\cite{Pastawski2015JHEP_holographic}.

One can apply the certification criteria to this code to yield the same scaling
as for the three-rotor code~\eqref{eq:lowerbound3rot} for the model of a single
erasure (see \cref{appx:AppendixCalcCodes} for details),
\begin{align}
  \epsilon_{\mathrm{worst}}(\mathcal{N}_{\text{1 erasure}}\circ\mathcal{E}^{(w)})
  \lesssim \frac{1}{\sqrt{160}}\frac{h}{w}\ . \label{eq:fiverotorepsone}
\end{align} 
However, this code is capable of correcting any single-subsystem error, so it
can correct for known erasure of any two subsystems. Calculating the bound for the noise channel 
$\cal{N}$ consisting of erasure of any two sites with equal probability yields
the same scaling,
\begin{align}
  \epsilon_{\mathrm{worst}}(\mathcal{N}_{\text{2 erasures}}\circ\mathcal{E}^{(w)})
  \lesssim \frac{1}{\sqrt{60}}\frac{h}{w}\ .
  \label{eq:fiverotorepstwo}
\end{align} 
The larger coefficient is sensible since a code approximately correcting at most
two erasures should be better at correcting only one. In both cases, there are
additional corrections of order $O(h e^{-c w^2})$ for $c>0$ arising from a
detailed application of our criterion.

\subsection{Thermodynamic codes for $n\to\infty$}
\label{sec:thermodynamic-codes}

We now investigate a class of covariant codes in the limit where the number of
subsystems $n$ grows large.  We exploit the codes developed in
Ref.~\cite{Brandao2017arXiv_chainAQECC}, relevant for quantum computing with
atomic ensembles~\cite{Saffman2010}.

For these codes the basis vectors for the code space can be chosen to be energy eigenstates of a many-body system, with the property that the reduced state on a subsystem appears to be thermal with a nonzero temperature; we therefore call them \textit{thermodynamic codes}. This thermal behavior of local subsystems is expected for closed quantum systems that satisfy the eigenstate thermalization hypothesis~\cite{Srednicki1994PRE_ETH} or dynamical typicality~\cite{Popescu2006NPhys_entanglement,Riera2012_thermalization}. Energy eigenstates with slightly different values of the total energy also have slightly different values of the locally measurable temperature; thus the identity of a codeword is imperfectly hidden from a local observer, and therefore erasure of a subsystem is imperfectly correctable.  

Consider a many-body system, such as a one-dimensional spin chain, and pick out
two global energy levels $\ket{E}_A, \ket{E'}_A$ in the middle of the spectrum,
with a given energy difference $\Delta E = E'-E$.  Assume, in the spirit of the
eigenstate thermalization hypothesis, that the reduced states of both
$\ket{E}_A$ and $\ket{E'}_A$ on each individual system $A_i$ are approximately
thermal.  The corresponding temperature scales as $T \propto E/n$ since the
temperature is an intensive thermodynamic variable.  Then, the temperature
difference vanishes for $n\to\infty$, and the resulting reduced thermal states
for these two states are very close.  Intuitively, this means that if a system
$A_i$ is provided to the environment, the latter cannot tell whether the global
state is $\ket{E}$ or $\ket{E'}$, and hence the two energy levels form a
two-dimensional code space that is approximately error-correcting against
erasures at known locations.

For example, consider the code developed
in~\cite[Appendix~D]{Brandao2017arXiv_chainAQECC}, in the context of a 1D
translation-invariant Heisenberg spin chain.  Here we consider as relevant
charge the total magnetization $M = \sum \sigma_Z^i$ of the spin chain.  The
codewords $\ket{h_m^n}$ in~\cite[Appendix~D]{Brandao2017arXiv_chainAQECC} are
Dicke states with respect to total magnetization---i.e., they are a
superposition of canonical $n$-spin basis states that all have some fixed
magnetization $m$:
\begin{align}
  \ket{h_m^n} = \begin{pmatrix} n \\ n/2+m/2 \end{pmatrix}^{-1/2}\,
  \sum_{\boldsymbol{s}:\; \sum s_j = m} \ket{\boldsymbol{s}}_n\ .
\end{align}
The code is covariant with respect to total magnetization by construction, by
defining the magnetization charge operator in the abstract logical system to
correspond to the magnetization of the corresponding codeword.  The values $m$
are spaced out by steps of $2d+1$, thus ensuring that any errors which change
the magnetization by at most $2d$ cannot cause logical bit flips. This
trick---using a sufficiently large spacing between codewords so that they are
not mapped into each other by errors---has analogues in CSS codes, related
multi-qubit codes~\cite{Ouyang2014}, and bosonic
error-correction~\cite{Albert2018PRA_bosonic}.  However, to show that such
errors are indeed correctable, one still has to make sure that expectation
values of errors with each codeword do not depend on the codeword in the
large-$n$ limit.

In \cref{appx:AppendixCalcCodes}, we show that this code's approximation
parameter as an approximate error-correcting code against the erasure of a
constant number of sites scales as
\begin{align}
  \epsilon_{\mathrm{worst}}(\mathcal{N}\circ\mathcal{E})
  =O(1/n)\ .
\end{align}
On the other hand, our bound~\eqref{eq:simple-bound} also displays the same 
scaling,
\begin{align}
  \epsilon_{\mathrm{worst}}(\mathcal{N}\circ\mathcal{E})
 = \Omega(1/n)\ .
\end{align}
In consequence, this code has an approximation parameter that displays
the same scaling as our bound, meaning that our bound is approximately tight in
the regime $n\to\infty$.


%% file: Repbound.tex
Our second main technical contribution is an approximate version of the
Eastin-Knill theorem.
The Eastin-Knill theorem states that it is not possible for an error-correcting
code to admit a universal set of transversal logical gates, imposing severe
restrictions on fault-tolerant quantum
computation~\cite{Eastin2009PRL_restrictions}.  In fact, an approximate version
of the Eastin-Knill theorem naturally follows from our bounds in
\cref{thm:simple-main-result-one-erasure}.  This is intuitively seen in the
setup of our main theorem depicted in \cref{fig:covariant-code}, by choosing the
transformation group to be the full unitary group $\UU(d_L)$ on the logical
system: To any logical unitary we require that there correspond an transversal
unitary on the physical system that achieves the same logical transformation.
Hence, our bound provides a limitation to the accuracy of codes that admit a
universal set of transversal logical gates.  The goal of this section is to
specialize our main bound~\eqref{eq:simple-bound} to this situation, in order to
obtain a limit expressed in terms of the dimensions of the local physical
subsystems.

There is a subtlety worth noting in the argument above.  In the setting of the
Eastin-Knill theorem, it is not necessarily required that the mapping of logical
to physical unitaries forms an actual representation, i.e., that it is
compatible with the group structure.  However, it turns out that we may assume
this without loss of generality.  Intuitively, as long as one can generate
logical unitaries that are close to the identity with a transversal physical
unitary, one can show that there are corresponding physical generators which
span a \emph{bona fide} representation (\cref{appx:ProofUdBounds}).  That is, if
a code admits a universal transversal gate set, then it is necessarily covariant
with respect to the full logical unitary group for some transversal
representation on the physical systems.

The bounds derived in
\cref{thm:simple-main-result-one-erasure,thm:main-result-full} cannot in general be directly
related to the dimension of the local physical subsystems.
Indeed, there is no restriction on how large $\Delta T_{i}$ can be.  The only
restriction that enters the statement of our main theorem is that a logical
charge eigenstate must be mapped onto a global physical eigenstate of the same
charge (up to a constant offset); the logical charge operator and the local
physical charge operators may otherwise be chosen arbitrarily.  For example, the
repetition code spanned by $\{\ket{000},\ket{111}\}$ with logical charge
$\delta\sigma_{z}$, physical charge
$M\sigma_{z}^{(1)}-M\sigma_{z}^{(2)}+\delta\sigma_{z}^{(3)}$, and $M\gg\delta$
can have a very large range $M$ of charges on each local physical subsystem
despite the systems having only two levels.  In the other extreme, a completely
degenerate local physical system will have zero charge range despite a possibly
huge dimension.

The above observation is an expression of the fact that the covariance is with
respect to an abelian symmetry group ($\UU(1)$).  In contrast, for non-abelian
Lie groups, one may no longer choose the generators arbitrarily because they
have to obey nontrivial commutation relations with each other.  Consider for
instance a code that is covariant with respect to spin, where the group is
$\SU(2)$.  The three generators of the corresponding Lie algebra, $J_x$, $J_y$,
and $J_z$, satisfy the commutation relations $[J_x, J_y] =i J_z$ along with the
corresponding cyclic permutations of $x,y,z$.  We know in the case of $\SU(2)$
that the irreducible representations are labeled by a spin quantum number $j$
that is a positive integer or half-integer, that the generator $J_z$ in this
representation has nondegenerate eigenvalues $m = -j,-j+1, \ldots, +j$, and
hence that the dimension of the irreducible representation labeled by $j$ is
$2j+1$.  By rotational symmetry, the same holds for any other standard generator
in that irreducible representation by choosing an appropriate basis.  In other
words, if the dimension of a physical subsystem is small, we cannot ``fit'' a
generator on that system with a large range of angular momentum values.
More precisely, if $T_i^z$ is the spin generator corresponding to $J_z$ on the
$i$th physical subsystem, the largest irreducible representation that can
appear in the action of $T_i^z$ must fit in the physical subsystem, that is, we
may not have any $j$ larger than $(d_i-1)/2$, where $d_i$ is the dimension of
the $i$th physical subsystem, or else the representation is too big.  In turn,
this bounds the range of $J_z$ charge values as $\Delta T_i^z \leqslant d_i-1$.
Hence, if we encode a qubit using a code that admits a universal set of
transversal logical unitaries, we may apply our bound~\eqref{eq:simple-bound},
choosing $T_L = J_z = \diag(1/2,-1/2)$ on the logical level with $\Delta T_L = 1$,
with the corresponding $\Delta T_i = \Delta T_i^z \leqslant d_i-1$; we then
obtain
\begin{equation}\label{eq:su2bound}
  \epsilon_{\mathrm{worst}}(\text{$\SU(2)$-covariant code})
  \geqslant \frac{1}{2n\,\max_{i}(d_{i}-1)}\ .
\end{equation}
Thus, remarkably, the non-abelian nature of the group $\SU(2)$ allows us to bound the
expression in~\eqref{eq:simple-bound} directly in terms of the dimensions of the
local physical subsystems.
This is because physically, the generators $J_{x,y,z}$ of $\SU(2)$ correspond to
rotations around different axes, and the Lie algebra commutation relations
require all of them to be of a similar scale.  No such requirement was present
for $\UU(1)$ since we were free to rotate around a chosen axis arbitrarily quickly.

In the case of a code that is covariant with respect to $\SU(d)$ with $d>2$, the
dependence on the physical subsystem dimensions
becomes considerably more restrictive.  We provide an overview of our argument,
leaving technical details to \cref{appx:ProofUdBounds}.  Irreducible
representations, or \emph{irreps}, of $\SU(d)$ are indexed by $d-1$ nonnegative
integers $(\lambda_{1},\lambda_{2},\cdots,\lambda_{d-1})\equiv\lambda$ arranged
in decreasing order. These integers determine the largest eigenvalues of the now
$d-1$ commuting generators of $\SU(d)$. For $\SU(2)$, only one generator $J_z$
is diagonal in the canonical basis and the integer $\lambda=\lambda_1=2j$
determines the highest spin attainable in that irrep. For the fundamental
representation $\lambda=(1,0)$ of $\SU(3)$, the two simultaneously
diagonalizable generators are the two Gell-Mann matrices that are diagonal in
the canonical basis. Since the entries in $\lambda$ are decreasing, the largest
eigenvalue that any generator could have is $\lambda_{1}$, i.e.,
$\norm[\big]{T_{\lambda}^{(i)}}_{\infty}\leq\lambda_{1}$.  It turns out that the
irrep that minimizes the dimension out of all irreps with fixed $\lambda_{1}$ is
the completely symmetric irrep $(\lambda_{1},0,0,\cdots,0)$.  The dimension of
this irrep is the dimension of the symmetric subspace on $\lambda_1$ number of
$d$-dimensional systems, which is a polynomial of degree $d-1$ in
$\lambda_1$. Therefore, in order to fit in a system of dimension $d_i$, the
largest possible $\lambda_{1}$ is of order $O`\big(d_{i}^{1/(d-1)})$.  Now, any
general representation can be decomposed into irreps, and a generator $T$ is
simply $T = \bigoplus T_\lambda$, where $T_\lambda$ is the corresponding
generator for each irrep.  We then have
$\norm{T}_\infty = \max_\lambda \norm{T_\lambda}_\infty$.  So, if a
representation fits in the system dimension $d_i$, then it cannot contain any
irrep $\lambda$ with $\lambda_1$ larger than $O`\big(d_i^{1/(d-1)})$.
For a code that is covariant with respect to the full unitary group on the
logical space, we have $d=d_L$, and picking a simple standard generator for our
earlier bound~\eqref{eq:simple-bound}, we obtain the following theorem.
\begin{theorem}[Approximate Eastin-Knill theorem]
  \label{thm:approximate-EK-main}
  The performance of an $SU(d_{L})$-covariant code, quantified by the worst-case
  entanglement fidelity, is bounded as follows:
  \begin{multline}
    \label{eq:appek1}
    \epsilon_{\mathrm{worst}}(\text{\upshape $\SU(d_L)$-covariant code})
    \\ \geqslant \quad
    \frac {1}{2n}\frac{1}{\max_{i}\ln d_{i}}+O`*(\frac{1}{n\,d_L})\ .
  \end{multline}
  The following bounds also hold:
  \begin{subequations}
    \label{eq:approximate-EK-di-geq-expressions}
    \begin{align}
      \max_i \ln d_i
      & \geqslant \frac{\ln `*(d_L-1)}{2n\epsilon_{\mathrm{worst}}}
        - \frac{\ln`*(1 + (2n\epsilon_{\mathrm{worst}})^{-1})}{2n\epsilon_{\mathrm{worst}}}\ ;
        \label{eq:approximate-EK-di-geq-poly-dL}
      \\
      \max_i \ln d_i
      &\geqslant 
        (d_L-1) \ln`*(
        \frac{1}{2\epsilon_{\mathrm{worst}} n d_L})
        \ .
        \label{eq:approximate-EK-di-geq-exp-dL}
    \end{align}
  \end{subequations}
  Similar bounds can be obtained for the figure of merit
  $\epsilon_{\mathrm{e}}$, based on the average entanglement fidelity, by making
  in~\eqref{eq:appek1} and~\eqref{eq:approximate-EK-di-geq-exp-dL} the replacement
  $\epsilon_{\mathrm{worst}} \rightarrow d_L\epsilon_{\mathrm {e}}/2$.
\end{theorem}

In other words, any code that (a)~stores a large amount of quantum information,
and (b)~admits universal transversal gates, has severe restrictions on its
ability to recover from erasure errors.

The bound~\eqref{eq:appek1} is useful to determine the precision limit of a code
that has a universal set of transversal gates.  If we imagine that each physical
subsystem is composed of $m_i$ qubits lumped together, then the error parameter
of the code scales at least inversely in the largest number of qubits $m_i$ that
were lumped together.  If we set for instance $d_L\sim 10^3$ (10 logical qubits)
that are encoded into $n$ systems consisting of 10 qubits each, i.e.,
$d_i \sim 10^3$, we obtain the rather prohibitive error parameter
$\epsilon_{\mathrm{worst}} \gtrsim 0.14/n$.  (This estimate can be improved to
$\epsilon_{\mathrm{worst}} \geqslant 0.5/n$ using a tighter bound given in
\cref{appx:ProofUdBounds}.)

The bound~\eqref{eq:approximate-EK-di-geq-poly-dL} shows that if
$\epsilon_{\mathrm{worst}}$ is kept constant and for $d_L\to\infty$, we must
have that $d_i$ grows polynomially in $d_L$, where the exponent is
$1/(2n\epsilon_{\mathrm{worst}})$.  If, for instance, we wish to achieve a
precision of $\epsilon_{\mathrm{worst}}\sim 10^{-3}$, then we must have the
scaling $d_i\sim (d_L)^{500/n}$.  Concretely, for $d_L\sim 10^3$ (10 logical
qubits) encoded into $n=10$ subsystems, the physical subsystems need to be of a
respectable dimension $d_i\sim 10^{65} \sim 2^{216}$, i.e., the physical
subsystem must comprise $216$ qubits lumped together.

Our third bound is interesting in the regime of extremely high accuracy.
Suppose we wish to accurately resolve individual logical basis states of a
highly mixed logical state.  The logical information might, for instance, be
entangled with a large reference system.  In such a situation, we require
$\epsilon_{\mathrm{worst}}\lesssim d_L^{-1}$.
Bound~\eqref{eq:approximate-EK-di-geq-exp-dL} then asserts that the physical
subsystem dimension must grow \emph{exponentially} in the logical system
dimension.

Finally, we note that
\cref{eq:appek1,eq:approximate-EK-di-geq-poly-dL,eq:approximate-EK-di-geq-exp-dL}
are obtained from a more general, tighter bound on $\max d_i$ which is expressed
as a binomial coefficient (see \cref{appx:ProofUdBounds} for details).  In some
cases, this bound allows to obtain tighter estimates on the physical dimension
of the subsystems.


%% file: Random_Code_Main_Text.tex
The bounds of \cref{thm:approximate-EK-main} severely limit the error correction
capability of the unitary $\SU(d_L)$-covariant codes.  We now show that it is
possible to find good $\SU(d_L)$-covariant codes in regimes of large physical
systems that are not excluded by \cref{thm:approximate-EK-main}.

The constructions we present are randomized as well as asymptotic in the
dimension of the physical subsystems. More precisely, we consider the encoding
of one $d_L$-dimensional Hilbert space $\Hil_L$ in a physical space which is a
tensor product of three Hilbert spaces
$\Hil_A=\Hil_{A_1} \otimes \Hil_{A_2} \otimes \Hil_{A_3}$. The encoding is done
via an isometry $V_{L\to A}$, which is $\UU(d_L)$ covariant: For all
$U\in U(d_L)$,
\begin{equation}
\label{eqn:cov_def}
V U =r_1(U) \ot r_2(U) \ot r_3 (U)\,V\ .
\end{equation} 
Here $r_1, r_2,\text{and } r_3$ are three irreps of $\UU(d_L)$. Our
constructions are randomized in the following way:
\begin{itemize}
\item $V$ is chosen randomly from all possible isometries satisfying the covariance condition~\eqref{eqn:cov_def}.
\item The irreps $r_1, r_2,\text{and } r_3$ are chosen randomly, or at least
  \emph{generically}. In fact, we only need that the irreducible representation
  does not belong to a small subset of all possible irreducible representations.
\end{itemize}
In \cref{appx:Random_Consts}, we use randomized constructions to prove existence
of $U(d_L)$-covariant codes with small error (measured by $\epsilon_e$ based on
the average entanglement fidelity), as summarized in the following theorem:
\begin{theorem}
  \label{thm:randomcode}\noproofref
  For $d_L\geqslant 4$ and every $\epsilon>0$, there exists a
  $\UU(d_L)$-covariant code with error
  $\epsilon_{\mathrm{e}} \leqslant \epsilon$ and physical dimensions $d_i$,
  $i\in \{1,2,3\}$, such that
  \begin{equation}
    \label{eq:randcode}
    \max_i \ln d_i \leqslant
    d_L(d_L-1) \ln`*(\frac{1}{\epsilon_{\mathrm{e}}}) +C_2,
  \end{equation}
  for some $C_2$ which is only a function of $d_L$.
\end{theorem}
It is not clear how to compare the performance of our code given
by~\eqref{eq:randcode} to our
bounds~\eqref{eq:approximate-EK-di-geq-expressions} because our nonconstructive
proof does not specify the behavior of $C_2$ as a function of $d_L$, which is
given by details of the representation theory of $\UU(d_L)$.  It remains open
whether the lower bound can be strengthened or the constructions can be
improved.

Our proof technique does not immediately work for $\UU(2)$-covariant codes, as
it is harder to bound the fluctuations of the fidelity of recovery when the
logical Hilbert space is too small.  For $\UU(3)$-covariant codes, our methods
lead to codes with a slightly different scaling from \cref{eq:randcode}.  In
fact, for the $\UU(3)$-case, one can provide randomized and
\emph{non-asymptotic} constructions (which work for known finite physical
dimensions) using the explicit formulas for the Littlewood-Richardson
coefficients~\cite{rassart2004polynomiality}. These constructions are not
included in the present paper, as there is little specific interest in the
$d_L=3$ case.

The proof of \cref{thm:randomcode} is technical and relies on the
representation theory of the unitary group (see \cref{appx:Random_Consts} for
details).  The proof starts by connecting the average fidelity
recovery of erasure of a fixed subsystem to the smoothness of the
\emph{Littlewood-Richardson coefficients}. 
Littlewood-Richardson coefficients are representation theory quantities that
count the degeneracy of a particular irrep of $\UU(d_L)$ in the tensor product
of two other irreps, and their smoothness follows from modern results in
representation theory of the unitary group~\cite{rassart2004polynomiality}
(\cref{fig:chamberc}).
\begin{figure}
  \centering
  \includegraphics[width= 6cm]{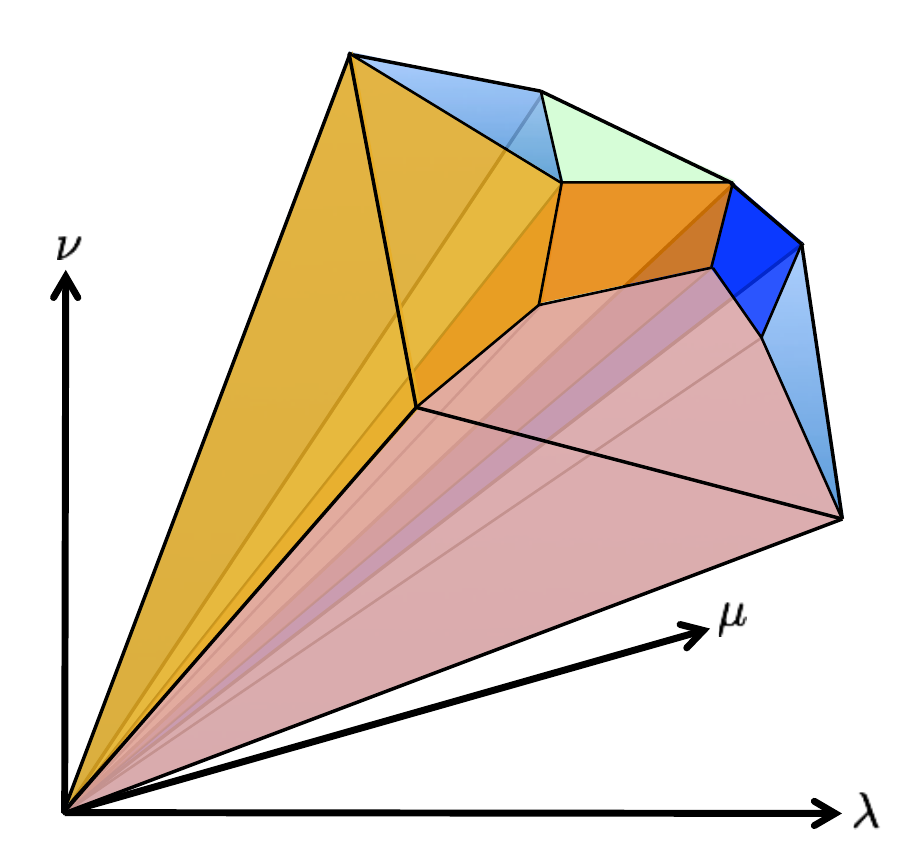}
  \caption{Smoothness of the Littlewood-Richardson coefficients, required for
    our proof that random covariant codes can asymptotically correct against
    errors.  A Littlewood-Richardson coefficient $c_{\mu\nu}^\lambda$ is the
    coefficient that counts the degeneracy of the $\UU(d_L)$-irrep labeled by
    the Young diagram $\lambda$ in the tensor product of two other irreps
    labeled by $\mu$ and $\nu$. The Littlewood-Richardson coefficients are
    non-zero in a convex cone in the space of $(\mu,\nu,\lambda)$. This
    cone---or \emph{chamber complex}---is divided into several smaller convex
    cones---or chambers---in which $c_{\mu\nu}^\lambda$ is a polynomial of
    $\mu,\nu$, and $\lambda$.  Hence, a generic choice of irreps on which we
    chose a random code will have corresponding coefficients that are smooth,
    which we show implies good asymptotic performance of the code.}
  \label{fig:chamberc}
\end{figure}


%% file: WStateExample.tex
Here we consider another example of an approximate quantum error-correcting
code, covariant with respect to the full unitary group on the logical system.
It is based on the $W$ state, and achieves an arbitrarily small
$\epsilon_{\mathrm{worst}}$. in the limit of a large number of subsystems,
$n\to\infty$. The logical system $L$ of dimension $d_L$ is encoded into a
physical system composed of $n$ copies of a $(d_L+1)$-dimensional space, where
each subsystem is a copy of the logical system with an additional basis vector
$\ket\perp$.  The encoding is
\begin{multline}
  \ket\psi_L \to \frac1{\sqrt n}\bigl( \ket{\psi,\perp,\ldots,\perp} + \ket{\perp,\psi,\perp,\ldots} + {}
  \\
  \cdots\  + \ket{\perp,\ldots,\perp,\psi}\bigr)\ .
  \label{eq:W-state-codeword-psi}
\end{multline}
Any logical unitary $U$ can be carried out on the encoded state transversally by
applying the unitary $U\otimes U\otimes\cdots \otimes U$, where we let $U$ act
trivially on the extra state $\ket\perp$.

Remarkably, aside from being $\UU(d_L)$-covariant, this trivial code is also effective against random erasures.
Intuitively, this is because the environment will only receive access to the logical state $\ket\psi$ with
probability $1/n$ if it gets access to a single
subsystem; that is, the environment is unlikely to learn anything about
the logical information.  This can be formalized with a direct application of
our criterion (\cref{prop:criterion-simple-main}).  Given a basis $\ket{x}_L$
of $L$, the reduced state $\rho_1^{x}$ on any single physical subsystem of the
codeword~\eqref{eq:W-state-codeword-psi} corresponding to $\ket{x}_L$ is
\begin{align}
  \rho_1^x = \frac1n\proj{x} + `*(1-\frac1n)\proj\perp\ ,
\end{align}
and thus $F^2(\rho_1^x,\proj\perp) = 1 - (1/n)$ for all $x$, and it follows that
$\sqrt{1 - F^2(\rho_1^x,\proj\perp)} \leqslant \sqrt{2/n} =: \epsilon$.
For $x\neq x'$ we have according
to~\eqref{eq:criterion-simple-reduced-state-i-codeword-x-xp},
\begin{align}
  \rho_1^{x,x'} = \frac1{n}\ketbra{x}{x'}\ ,
\end{align}
and thus $\norm{\rho_1^{x,x'}}_1 = 1/n =: \nu$.  The corresponding reduced
states on the other physical subsystems are the same by symmetry of the
codeword~\eqref{eq:W-state-codeword-psi}.  Then
\cref{prop:criterion-simple-main} asserts that this code has an error
parameter that is at most
\begin{align}
  \epsilon_{\mathrm{worst}} \leqslant \frac{\sqrt{2}+d_L}{\sqrt n}\ .
\end{align}
That is, for fixed $d_L$, the code becomes a good error-correcting code in the
limit $n\to\infty$.

In contrast to the thermodynamic codes presented above, this $W$-state code does
not saturate our bound on $\epsilon_{\mathrm{worst}}$, which is inversely
proportional to $n$ rather than the square root of $n$.  The reason for this
discrepancy is the same as for the difference between a sharp and a smooth
cut-off for the three-rotor code, discussed in
\cref{subsec:three-rotor-smooth-cutoff}.  Again, here, as $n$ becomes large, the
local reduced state grows close to the rank-deficient state $\proj\perp$, which
is a regime where the infidelity is particularly sensitive to small
perturbations.  In contrast, for instance, our thermodynamic codes of
\cref{sec:thermodynamic-codes} have reduced states that are full-rank, allowing
the code to achieve the same scaling as our accuracy bound as $n\to\infty$.
While this code does not achieve the same $1/n$ scaling as the thermodynamic
codes, it does exhibit covariance with repect to the full logical unitary group
$\UU(d_L)$.


%% file: Gcodes.tex
In this section, we develop a framework for constructing codes that are
covariant with respect to any group $G$ admitting a left- and right-invariant
Haar measure, encompassing in particular codes that are based on rotors,
oscillators, and qudits.
Our construction is based on quantum systems that transform as the \emph{regular
  representation} of  $G$. Orthonormal basis states $`{\ket{g}}_{g\in G}$ for this representation are labeled by group elements; if the group has an infinite number of elements, then the
quantum system is infinite-dimensional.

A qubit can transform as the regular representation of
the group $\mathbb{Z}_2$, and a qudit as the regular representation of
$\mathbb{Z}_d$.  An oscillator provides a regular representation of the (noncompact) group
$\mathbb{R}$, with the group acting by translation in either its position basis
$`{\ket{x}}$ or its momentum basis $`{\ket{p}}$.  Similarly, a rotor provides a
regular representation of the group $\UU(1)$, with orthonormal basis states $`{\ket{\ee^{i\phi}}}$; when Fourier transformed, it can transform as a regular representation of
$\mathbb{Z}$, where the basis states are the eigenstates of angular momentum 
$`{\ket{\ell}}_{\ell\in\mathbb{Z}}$.

For ease of presentation, we will consider codes whose logical system
and whose physical subsystems transform as the regular representation of any compact group $G$, commenting on noncompact groups in \cref{subsec:last}. 
Well-known qubit codes such as the bit-flip, phase-flip, and
$[[4,2,2]]$ codes, naturally extend to this setting. More generally,
we will also discuss extensions of the $[[m^2,1,m]]$ and $[[2m,2m-2,2]]$ qubit codes.

\subsection{Bit- and phase-flip codes}

For simplicity, let us review bit-flip and phase-flip codes first.  An $M$-qubit
bit-flip encoding copies the logical basis state index $x\in\mathbb{Z}_{2}$ in
each of the $M$ subsystems. An $M$-qubit phase-flip encoding hides the logical
index in the sum of the physical qubit states.  Taking $M=3$ for concreteness,
the two encodings are
\begin{subequations}
\label{eq:quditcodes}
\begin{align}
\ket x_{L}^{\text{bit}} & \rightarrow\ket{x,x,x}\\
\ket x_{L}^{\text{phs}} & \rightarrow \frac{1}{2} \sum_{y_{1},y_{2},y_{3}\in\mathbb{Z}_{2}}\delta_{x,y_{1}+y_{2}+y_{3}}\ket{y_{1},y_{2},y_{3}}\,,
\end{align}
\end{subequations}
where $\delta_{x,y}=1$ if $x=y$ modulo 2. Bit-flip codes protect
against single-qubit shifts $x\rightarrow x+1$ while phase-flip codes
protect against single-qubit operators which are diagonal in the canonical basis. 

By viewing a qubit as a regular representation of the group $G=\mathbb{Z}_{2}$, we can see how to generalize this construction to other groups. For a finite group $G$ with order $|G|$, consider the $|G|$-dimensional Hilbert space $V$ spanned by $\{\ket g\,|\,g\in G\}$ with inner product $\braket gh=\delta_{g,h}$, where $\delta_{g,h}=1$ if $g$ and $h$ are the same group element and zero otherwise. For compact continuous groups, the Hilbert space is infinite-dimensional, and $\delta_{g,h}$ becomes the Dirac delta function---infinite when $g=h$ and zero otherwise---and sums $\frac{1}{\left|G\right|}\sum_{h\in G}$ are replaced by integrals $\int dg$, where $dg$ is the group's normalized Haar measure~\cite{sternberg_book,Arovas}. We'll write sums below for simplicity, with the understanding that the sum is to be replaced by an integral when $G$ is a compact Lie group. 

The respective $M=3$-subsystem bit- and phase-flip generalizations
of \cref{eq:quditcodes} for finite groups are
\begin{subequations}
\label{eq:groupcodes}
\begin{align}
\ket g_{L}^{\text{bit}} & \rightarrow\ket{g,g,g}\label{eq:bitflip}\\
\ket g_{L}^{\text{phs}} & \rightarrow\frac{1}{\left|G\right|}\sum_{h_{1},h_{2},h_{3}\in G}\delta_{g,h_{1}h_{2}h_{3}}\ket{h_{1},h_{2},h_{3}}\,.\label{eq:phaseflip}
\end{align}
\end{subequations}
The bit-flip encoding records a group element redundantly, while the phase-flip
encoding hides $g$ in a product of three group elements. The error-correction
properties of these codes are analogous to those for $G=\mathbb{Z}_{2}$:
the bit-flip codes correct against errors which take individual subsystems
into states orthogonal to $\ket g$ while phase-flip codes correct
against single-subsystem errors diagonal in the $\ket g$-basis. 

To perform an $X$-type gate on these codes, introduce left and right-multipliers,
$\xl_{g}$ and $\xr_{g}$, which act as 
\begin{align}
  \xl_{g}\ket h &=\ket{gh}\qquad\text{and}\qquad
  \xr_{g}\ket h =\ket{hg}\ .
\end{align}
The sets $\{\xl_{g}\}_{g\in G}$ and $\{\xr_{g}\}_{g\in G}$ are permutation matrices forming
the left and right regular representations of $G$. Note that the
arrow points towards $h$ from the side that $g$ acts. Since multiplying
from the left commutes with multiplying from the right, the two sets
commute with each other. 

For the bit-flip code~\eqref{eq:bitflip}, the logical left multiplication gate
\begin{equation}
\xl_{L,k}^{\text{bit}}: \ket g_{L}^{\text{bit}}\rightarrow\ket{kg}_{L}^{\text{bit}}
\end{equation}
can be implemented transversally:
\begin{equation}
\xl_{L,k}^{\text{bit}} = \xl_{k}\otimes \xl_{k}\otimes \xl_{k}.
\end{equation}
For the phase-flip code, which provides no protection against bit-flips at all, logical left multiplication is implemented by acting on a single subsystem:
\begin{equation}
\xl_{L,k}^{\text{phs}} = \xl_{k}\otimes I\otimes I,
\end{equation}
where $I$
is the subsystem identity. Similar constructions hold for logical right multipliers.

For continuous $G$, the code states become nonnormalizable, but the
gates work the same way. Therefore, the logical operators $\xl_{L,k}$ define exact continuous symmetries of theses codes. However, these codes do not correct erasure of a subsystem; rather, each code corrects only a limited set of single-subsystem errors. The same is true for the qubit codes that inspired this construction.

We can concatenate the bit-flip code and the phase-flip code for qubits to
obtain Bacon-Shor codes~\cite{Ralph2005,Bacon2006PRA_subsystems}, which have the
parameters $[[m^{2},1,m]]_{\mathbb{Z}_{2}}$. This notation means that one
logical qubit is encoded in a code block of $m^2$ physical qubits, and that the
code distance is $m$; hence erasure of any $m-1$ of the qubits can be
corrected. Of the codes in this family, the best known are the
$[[4,1,2]]_{\mathbb{Z}_{2}}$ error-detecting code~\cite{Steane1996,Steane1996a}
and Shor's nine-qubit $[[9,1,3]]_{\mathbb{Z}_{2}}$ error-correcting
code~\cite{Shor1995}.

Likewise, by concatenating the $G$-covariant bit-flip and phase-flip codes, we obtain the $G$-covariant $[[m^{2},1,m]]_{G}$ code. For finite $G$, this is a $G$-covariant encoding of a $|G|$-dimensional logical system in $m^2$ $|G|$-dimensional subsystems, protected against erasure of any $m-1$ of the subsystems. If $G$ is a compact Lie group, this code has continuous $G$ symmetry. In that case, as the Eastin-Knill theorem requires, the encoding is infinite-dimensional. 

Rather than discussing this generalized Bacon-Shor code construction more explicitly here, in \ref{subsec:[[4,2,2]]} we'll provide a more detailed discussion of a related code, with two rather than just one $|G|$-dimensional logical subsystems.

\subsection{The $[[4,2,2]]_{G}$ code and its generalizations}
\label{subsec:[[4,2,2]]}
There is also a $[[4,2,2]]_{\mathbb{Z}_{2}}$ qubit code~\cite{Grassl1997}, which
can be extended to a covariant $[[4,2,2]]_{G}$ code, with encoding map
\begin{equation}
\ket{g_{1},g_{2}}_{L}\rightarrow\frac{1}{\sqrt{\left|G\right|}}\sum_{g\in G}\ket{g,\,g^{-1}g_1,\, gg_2,\, g^{-1}g_1g_2}\,.\label{eq:422}
\end{equation}
In fact, the $[[4,2,2]]_{\mathbb{Z}_{2}}$ code can be viewed as a mimimal
version of Kitaev's toric code~\cite{Kitaev2003AoP_anyons}, defined by just one
plaquette operator and one star operator, and~\eqref{eq:422} defines the
corresponding quantum double code with group $G$.

Given $l\in G$, the physical operator $I\otimes I\otimes \xr_{l}\otimes \xr_{l}$ has the effect of replacing $g_2$ by $g_2l$ in  \cref{eq:422}, hence mapping the logical state to $|g_1,g_2 l\rangle_L$. The physical operator $\xl_{l}\otimes I\otimes \xl_{l}\otimes I$, after a redefinition of the summation variable ($g \rightarrow l^{-1} g'$), has the effect of replacing $g_1$ by $lg_1$, hence mapping to the logical state to
$|lg_1,g_2\rangle_L$. Since the left and right multipliers commute, and both logical operations are transversal, the code is covariant with respect to the group $G\times G$.

Using the quantum error-correction
conditions~\cite{Bennett1996PRA_MSEntglQECorr,Knill1997PRA_correction}, we can
check that this code corrects one erasure. Let $O_1$ be an operator acting on
the first subsystem, and consider its matrix element between code states.
Plugging into \cref{eq:422} and contracting indices we find
\begin{equation}
{}_L\bra{g_{1},g_{2}}O_{1}\ket{g_{1}^{\prime},g_{2}^{\prime}}_{L}=\delta_{g_{1},g_{1}^{\prime}}\delta_{g_{2},g_{2}^{\prime}}\tr(O_{1})/\left|G\right|\,.\label{eq:ercond}
\end{equation}
This means that the code satisfies the condition for correctability of erasure of the first subsystem.
A similar calculation can be performed for operators acting on any of the other subsystems; therefore erasure is correctable for each of the four subsystems.

The $[[4,2,2]]_{\mathbb{Z}_{2}}$ qubit code can be generalized to a $[[2m,2m-2,2]]_{\mathbb{Z}_{2}}$ code, which can also be extended to a covariant $[[2m,2m-2,2]]_{G}$ code for any group $G$. To understand this construction, first consider a different $[[4,2,2]]_{G}$ code, which has a smaller covariance group than the code described above. Now we use the encoding map
\begin{equation}
\ket{g_{1},g_{2}}_{L}\rightarrow\frac{1}{\sqrt{\left|G\right|}}\sum_{g\in G}
\ket{g,\,gg_1,\, g g_2g_1 ,\, gg_2}\,.\label{eq:422-1}
\end{equation}
Unlike the previously considered code, this code has the property of being invariant under the action of a ``stabilizer'' operator $S_{l}=\xl_{l}\otimes \xl_{l}\otimes \xl_{l}\otimes \xl_{l}$
for each $l\in G$. 
The price we pay for this invariance property is a reduction in the number of independent transversal operations which act nontrivially on the code space. There is no nontrivial symmetry of the code acting from the left, but the operator $I\otimes \xr_{l}\otimes \xr_{l}\otimes I$ maps $|g_1,g_2\rangle_L$ to $|g_1l,g_2 \rangle_L$. Therefore, this code is $G$-covariant. We can also check that it satisfies the condition for correctability of erasure for each one of the four subsystems.

To illustrate how this code generalizes to a higher-length code with more physical subsystems, we will, to be concrete, describe the corresponding $[[2m,2m-2,2]]_{G}$ code with $m=4$. This code has the stabilizer $S_{l}=\xl{}_{l}^{\otimes 8}$ for each $l\in G$, and the encoding map
\begin{align}
&|g_1, g_2, g_3, g_4, g_5, g_6\rangle_L \notag\\
\rightarrow 
&\frac{1}{\sqrt{\left|G\right|}}\sum_{g\in G} S_g 
\ket{1,\,g_1,\,  g_2g_1 ,\, g_2 g_3, \, g_4g_3,\,g_4 g_5,\, g_6 g_5,\, g_6 }\,.
\end{align}
Aside from being invariant under the action of $S_l$, the code has another important property: each codeword is a superposition of states $\ket{h_1,h_2,h_3,h_4,h_5,h_6,h_7,h_8}$ of the eight physical subsystems having the property $h_1^{-1} h_2 h_3^{-1} h_4 h_5^{-1} h_6 h_7^{-1} h_8 = 1$ (for this to work the code has to have even length). These two properties together suffice to ensure that erasure of each subsystem is correctable.

This code is covariant under the group $G^3$. The operator
\begin{equation}
I\otimes \xr_{h_1}\otimes \xr_{h_1}\otimes \xr_{h_3}\otimes \xr_{h_3}\otimes \xr_{h_5}\otimes \xr_{h_5}\otimes I
\end{equation}
acts on the code's basis states according to
\begin{equation}
|g_1, g_2, g_3, g_4, g_5, g_6\rangle_L \rightarrow
|g_1h_1, g_2, g_3h_3, g_4, g_5h_5, g_6\rangle_L .
\end{equation}
In general, the $[[2m,2m-2,2]]_{G}$ code has a transversal $G^{m-1}$ symmetry, acting similarly.

\subsection{Further extensions and some limitations}\label{subsec:last}

One can extend these constructions to noncompact groups. For example, the
oscillator $[[9,1,3]]_{\mathbb{R}}$ code was noticed early
on~\cite{Lloyd1998,Braunstein1998PRL_CV} (see
also~\cite{Barnes2004,Bermejo-Vega2016}). Another example is the rotor
$[[4,2,2]]_{\mathbb{Z}}$ encoding
\begin{equation}
\ket{a,b}\rightarrow\sum_{j,k,l\in\mathbb{Z}}\delta_{a,j+k}\delta_{b,l}\ket{j,k,j+l,k+l}~.
\end{equation}
As done in \cref{sec:maintext-truncated-Hayden-code}, one can impose an envelope
so that the codewords are normalizable. In general, a bi-invariant Haar measure
is sufficient to perform the left- and right-multiplier transversal gates as
well as the error-correction, but one would have to approximate the codewords to
avoid infinities due to non-normalizable Haar measures. For the oscillator code
$[[4,2,2]]_{\mathbb{R}}$, for which the above is an integral over oscillator
position states, we additionally need to approximate the position states with a
displaced and finitely squeezed vacuum~\cite{Gottesman2001PRA_oscillator}. In
other words, noncompactness \textit{and} the continuous nature of the group may
each require approximations to achieve normalizability of the codewords.

One may also ask if it is possible to extend the secret-sharing code
$[[3,1,2]]_{\mathbb{Z}_{3}}$ from \cref{sec:maintext-truncated-Hayden-code} to a
more general group $G$. An extension does indeed work for
$G\in\{\mathbb{R},\mathbb{Z},\mathbb{Z}_{2D+1},U(1)\}$, but the code breaks down
at, e.g., $\mathbb{Z}_{2D}$ due to there being a non-measure-zero set of
order-two elements in the group. Writing a natural guess for the encoding,
\begin{equation}
  \ket g_{L}\rightarrow\frac{1}{\sqrt{\abs{G}}}\sum_{h\in G}\ket{h,gh^{-1},gh^{-2}}\,,
\end{equation}
we see that the third subsystem stores the logical index ``in plain sight''
whenever $h^{2}=1$.  Roughly speaking, for groups with too many such elements,
the environment can extract logical information from the code.


%% file: Holography.tex
The interplay between continuous symmetries and quantum error correction has implications for holography and quantum gravity. 
The AdS/CFT correspondence~\cite{Maldacena1999IJTP_largeN,Witten1998ATMP_AdSCFT}
is a duality between quantum gravity in Anti-de Sitter (AdS) space, and a
conformal field theory (CFT) in one fewer spatial dimensions, where the CFT
resides on the boundary of the AdS space.  It was recently discovered that the
duality map from bulk operators to boundary operators may be regarded as the
encoding map of a quantum error-correcting code, where the code space is spanned
by low-energy states of the CFT. Specifically, local operators deep inside the
bulk AdS are encoded as highly nonlocal operators in the boundary CFT which are
robust against erasure errors in the boundary
theory~\cite{Almheiri2015JHEP_bulk,Harlow2016RMP_Jerusalem,%
  Harlow2018TASI_emergence}.  Here, we discuss symmetries of this AdS/CFT
code. First, we reprise a recent analysis
from~\cite{Harlow2018arXiv_constraints,Harlow2018arXiv_symmetries}, which rules
out exact global symmetries for quantum gravity in the bulk AdS space. Then we
explain how our results in this paper clarify the correspondence between time
evolution in the bulk and boundary theories.

\subsection{No bulk global symmetries}
\label{subsec:no-global}
A longstanding conjecture holds that quantum gravity is incompatible with global symmetry. One argument supporting this claim goes as follows~\cite{preskill1992black,Kallosh1995PRD_global}. According to semiclassical theory, which should be reliable for large black holes, the Hawking radiation emitted by a black hole is not affected by the amount of global charge the black hole might have previously consumed. Therefore, a process in which a black hole arises from the gravitational collapse of an object with large charge, and then evaporates completely, will not obey charge conservation.

This argument may not be trustworthy if the symmetry group is a small finite group, in which case the total charge cannot be ``large,'' and any missing charge might reappear in the late stages of black hole evaporation when semiclassical theory does not apply. 
But recently, Harlow and Ooguri used AdS/CFT technology to show that even discrete global symmetries are disallowed in the bulk~\cite{Harlow2018arXiv_symmetries}. Here we will reprise their argument, expressing it in language that emphasizes the conceptual core of the proof, and that may be more accessible for those familiar with the formalism of quantum error correction. 

To quantum coding theorists, it sounds strange to hear that the AdS/CFT code
cannot have discrete symmetries, because typical quantum codes do. To illustrate
this point we'll revisit a simple quantum-error correcting code that is often
used to exemplify the structure of the AdS/CFT code: the three-qutrit
code~\cite{Almheiri2015JHEP_bulk}, which we already discussed in
\cref{subsubsec:rotor-version}. This encodes a single logical qutrit in a block
of three physical qutrits, and protects against the erasure of any one of the
three qutrits.

The three-qutrit code is an example of a stabilizer code---the code space may be defined as the simultaneous eigenspace of a set of generalized Pauli operators. For a qutrit with basis states $\{|j\rangle, j =0,1,2\}$, the generalized Pauli group is generated by operators $X$ and $Z$ defined by
\begin{equation}
X|j\rangle = |j+1~({\mathrm{mod}}~3)\rangle, \quad Z|j\rangle = \omega^j|j\rangle, 
\end{equation}
where $\omega = e^{2\pi i / 3}$, which obey the commutation relations
\begin{equation}
ZX = \omega XZ, \quad Z^{-1}X = \omega^{-1} XZ^{-1}.
\end{equation}
The code space of the three-qutrit code is the simultaneous eigenspace with eigenvalue 1 of the operators
\begin{equation}
S_X = X\otimes X\otimes X, \quad S_Z = Z\otimes Z\otimes Z
\end{equation}
acting on the three qutrits in the code block. Note that, although $X$ and $Z$ do not commute, $S_X$ and $S_Z$ do commute, and can therefore be simultaneously diagonalized. Any nontrivial weight-one Pauli operator (supported on a single qutrit and distinct from the identity) must fail to commute with at least one of $S_X$ or $S_Z$. Therefore no nontrivial weight-one operator preserves the code space, which is why erasure of a single qutrit is correctable. 

However, there are weight-two Pauli operators that commute with both $S_X$ and $S_Z$, and therefore preserve the code space; for example, 
\begin{equation}
 X_L = X\otimes X^{-1} \otimes I, \quad  Z_L = Z\otimes I\otimes Z^{-1}.
\end{equation}
Because they preserve the code space, and act nontrivially on the code space, we say that $X_L$ and $Z_L$ are nontrivial logical operators for this code. Furthermore, $X_L$ and $Z_L$ obey the same commutation relations as $X$ and $Z$; they generate the logical Pauli group acting on an encoded qutrit. Note that because $S_X$ acts trivially on the code space, the operator $X_L$, which is supported on the first two qutrits, acts on the code space in the same way as $X_LS_X$, which is supported on the first and third qutrit, and also in the same way as $X_L S_X^{-1}$, which is supported on the second and third qutrit. A similar observation also applies to $Z_L$ and $S_Z$. This feature illustrates a general property: if $O_L$ is a logical operator, and $A$ is a subset of the qutrits in the code block such that erasure of $A$ is correctable, then we may represent $O_L$ as a physical operator supported on the complementary set $A^c$.

Our purpose in describing this code is just to point out that the transversal logical operators $X_L$ and $Z_L$ may be viewed as global symmetries of the code. The action of each of these operators on the logical system can be realized as a tensor product of single-qutrit operators. Such a symmetry is what Harlow and Ooguri rule out. We need to understand why their argument applies to the AdS/CFT code, but not to the qutrit code or to other stabilizer codes.

Harlow and Ooguri use special properties of AdS/CFT in two different ways, and their argument proceeds in two steps. The first step (explained in more detail below) appeals to \textit{entanglement wedge reconstruction},
together with the structure of global symmetries in quantum field theory, to show that any global symmetry acting on the bulk acts transversally on the boundary. That is, the boundary can be expressed as a union of disjoint subregions $\{A_k\}$ such that erasure of each $A_k$ is correctable, and any bulk global symmetry operator $U_L$, when reconstructed on the boundary, can be expressed as a tensor product $\bigotimes_k W_k$, where $W_k$ is supported on $A_k$. (Here we ignore a correction factor supported only where the regions touch, which is inessential to the argument.) This is just the property that we have assumed throughout this paper, and which is exemplified by the three-qutrit code discussed above. 

The second step of the argument (also explained further below) is the crucial one, which invokes a property of the AdS/CFT code which is not shared by the typical quantum codes which arise in work on fault-tolerant quantum computation. Harlow and Ooguri argue that each $W_k$ is itself a logical operator; that is, each $W_k$ maps the code space to the code space. The essence of this part of the argument is that the code space is the span of low-energy states in the CFT, and the $W_k$'s, perhaps after suitable smoothing, can be chosen so that they do not increase the energy of the CFT by very much. As we have already emphasized, this property does not apply to the three-qutrit code, where $X_L$ is a logical operator, yet its weight-one factors $X\otimes I\otimes I$ and $I\otimes X^{-1}\otimes I$ are not logical. Indeed, because $X\otimes I\otimes I$ changes the eigenvalue of the unitary operator $S_Z$ by the multiplicative factor $\omega$, it maps the code space (the simultaneous eigenspace of $S_Z$ and $S_X$ with eigenvalue 1) to a subspace orthogonal to the code space (the eigenspace of $S_Z$ with eigenvalue $\omega$).

A logical operator supported on a region $A$, where erasure of $A$ is correctable, must be the logical identity. We can easily see that's true, because otherwise an adversary could steal region $A$ and apply a nontrivial logical operator, altering the encoded state and therefore introducing an uncorrectable error. Now the conclusion of Harlow and Ooguri follows easily. The bulk global symmetry operator $U_L$ is a product of logical operators, each of which is trivial; therefore $U_L$ must be the identity. 

As Harlow and Ooguri note (Footnote~68 in~\cite{Harlow2018arXiv_symmetries}),
their argument, which excludes discrete symmetries of the AdS/CFT code as well
as continuous symmetries, is quite different than the Eastin-Knill argument,
which excludes only continuous symmetries of a code. Both arguments apply in a
framework where the symmetry of the code can be applied transversally, as a
product of local operators. But for the Eastin-Knill argument, there is no need
to assume that these local operators preserve the codespace, and therefore the
argument applies to general codes. In contrast, Harlow and Ooguri assert that
for the AdS/CFT code in particular, the local operators \textit{do} preserve the
code space. Therefore, their argument excluding discrete symmetries applies to
the AdS/CFT code, but not to the typical codes studied by quantum information
theorists.

For completeness, we'll now sketch the two key steps of the Harlow-Ooguri argument in slightly greater detail, starting with the step which shows that a bulk global symmetry acts transversally on the boundary.  We begin by noting that a global symmetry in the bulk AdS space implies a corresponding symmetry acting on the boundary; to see this we need only consider the action of the bulk global symmetry on bulk local operators in the limit where the support of the bulk local operators approaches the boundary. Furthermore, a global symmetry operator of the boundary CFT is \textit{splittable}; that is, it can be expressed as a tensor product of many operators, each supported on a small region. 
In coding theory language, the encoding isometry $V$ which maps bulk to boundary  has the property 
\begin{equation}
V U_L = U_{\mathrm{CFT}} V,
\end{equation}
where $U_L$ is the bulk symmetry operator and $U_{\mathrm{CFT}}$ is the corresponding CFT symmetry operator. Because the CFT symmetry is splittable, we may consider decomposing the CFT into small spatial subregions $\{A_k\}$, and infer that
\begin{equation}\label{eq:UCFT-Wk}
U_{\mathrm{CFT}} = \bigotimes_k W_k,
\end{equation}
where $W_k$ is a CFT operator supported on $A_k$.

Next we would like to see that the boundary subregions can be chosen so that erasure of any $A_k$ is correctable. This point is most naturally discussed using the language of operator algebra quantum error correction~\cite{Almheiri2015JHEP_bulk}. We consider the subalgebra $\mathcal{A}$ of logical operators which are supported on a subregion of the bulk. Each logical operator $O_L\in \mathcal{A}$ can be ``reconstructed'' as a physical operator $O_{\mathrm{CFT}}$ acting on the boundary using the encoding isometry $V$:
\begin{equation}
V O_L = O_{\mathrm{CFT}} V.
\end{equation}
What we wish to show is that, for any $O_L$ in $\mathcal{A}$, and for each boundary subregion $A_k$, the reconstructed boundary operator $O_{\mathrm{CFT}}$ can be chosen to have support on the complementary boundary subregion $A_k^c$. This property ensures that, for the bulk subalgebra $\mathcal{A}$, erasure of boundary region $A_k$ is correctable. 

The argument showing that erasure of boundary subregion $A_k$ is correctable is illustrated in \cref{AdSCFTDecomposition}. 
Associated with each boundary subregion $A_k$ is a bulk subregion $a_k$ which is called the entanglement wedge of $A_k$. The AdS/CFT code has these important properties~\cite{Almheiri2015JHEP_bulk}: (1) A bulk operator supported in bulk subregion $a_k$ can be reconstructed as a boundary operator supported in boundary subregion $A_k$. This property is called \textit{subregion duality}. (2) Furthermore, a bulk operator supported in the bulk complement $a_k^c$ of bulk subregion $a_k$ can be reconstructed as a boundary operator supported in the boundary complement $A_k^c$ of boundary subregion $A_k$. This property is called \textit{complementary recovery}. 

It follows from complementary recovery that if the bulk subalgebra $\mathcal{A}$
is supported in $a_k^c$, then erasure of boundary subregion $A_k$ is correctable
for the subalgebra $\mathcal{A}$. This is the key fact that we need. As in
\cref{AdSCFTDecomposition}, for any fixed subregion $a_0$ of the bulk, we can
choose the decomposition of the boundary into subregions $\{A_k\}$ such that
$a_0$ lies outside the entanglement wedge of each $A_k$. Therefore, the algebra
$\mathcal{A}$ of bulk operators supported on $a_0$ has the feature that erasure
of each $A_k$ is correctable for the algebra $\mathcal{A}$. This completes the
first step of the Harlow-Ooguri argument, showing that a bulk global symmetry
operator $U_L$ must be transversal in the sense we have assumed in this
paper---it factorizes as a tensor product of boundary operators, each of which
is supported on a correctable boundary subregion.

Actually, so far we have ignored a subtlety in this argument associated with general covariance in the bulk~\cite{Harlow2018arXiv_symmetries}. Operators acting in the bulk are not really strictly local; rather a bulk ``local'' operator is accompanied by gravitational \textit{dressing} which connects it to the boundary. This dressing is needed in order to enforce invariance under bulk diffeomorphisms. Because the dressing extends to the boundary, it has support on at least one of the $a_k$, and its reconstructed counterpart has support on at least one boundary subregion. However, this complication does not invalidate the argument, because the dressing is purely gravitational, and is therefore oblivous to the global charge defined within the bulk subalgebra $\mathcal{A}$.

\begin{figure}
\centering
\includegraphics{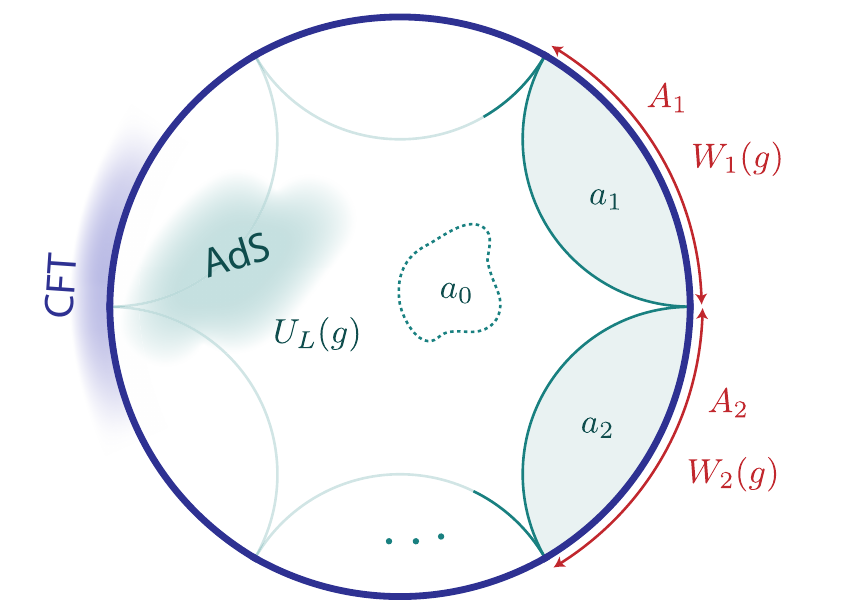}
\caption{\label{AdSCFTDecomposition} Nontrivial bulk global symmetries are
  incompatible with the AdS/CFT quantum error-correcting code. A bulk global
  symmetry operator $U_L(g)$ corresponds to a boundary global symmetry operator
  $U_{\mathrm{CFT}}(g)$ which is transversal with respect to the decomposition
  of the boundary into subregions $\{A_k\}$:
  $U_{\mathrm{CFT}}(g) = \bigotimes_k W_k(g)$, where $W_k(g)$ is supported on
  $A_k$.  The bulk subregion $a_0$ is outside the entanglement wedge $a_k$ of
  each boundary subregion $A_k$; therefore erasure of each $A_k$ is correctable
  for the algebra $\mathcal{A}$ of bulk local operators on $a_0$. Furthermore,
  each $W_k(g)$ maps low-energy states of the boundary CFT to low-energy
  states. This means that $W_k(g)$ is a logical operator which preserves the
  code subspace of the CFT. A logical operator $W_k(g)$ supported on a
  correctable boundary subregion $A_k$ must be the logical identity. Therefore
  the global symmetry operator $U_L(g)$ acts trivially on bulk local operators.}
\end{figure}

Now we come to the second part of the Harlow-Ooguri argument, which establishes that the operator $W_k$ supported on boundary subregion $A_k$ is actually a \textit{logical} operator. In the holographic correspondence, the choice of code space is actually rather flexible. One possible procedure~\cite{Almheiri2015JHEP_bulk} is to pick a set of local operators deep in the bulk, corresponding to highly nonlocal operators when reconstructed in the CFT. Then the code space is spanned by polynomials of bounded degree in these operators acting on the CFT vacuum state. The motivation for this choice is that each of the highly nonlocal CFT operators raises the energy of the CFT by only a small amount, hence producing only very weak back reaction on the bulk geometry. Logical operators are those that preserve this low-energy sector of the CFT, and Harlow and Ooguri assert that each operator $W_k$ can be chosen to have this property. 
Since $W_k$ preserves the code space, and is supported on the correctable boundary subregion $A_k$, it must act trivially on the code space. This assertion is affirmed if the code's logical operators may be regarded as bulk operators which are supported in a bulk region which is outside the entanglement wedge of the $A_k$ (such as the region $a_0$ in \cref{AdSCFTDecomposition}), since in that case each logical operator can be reconstructed on the complementary boundary region $A_k^c$, where $W_k$ acts trivially. Therefore, since each $W_k$ is a trivial logical operator, we conclude that the global symmetry operator $U_L$ is the identity acting on the code space. 

In this argument, we assumed that subregion duality and complementary recovery are \textit{exact} properties of the AdS/CFT code, and thus inferred that erasure of boundary region $A_k$ is \textit{exactly} correctable. 
In fact, though, these properties of the code hold precisely only in the leading order of a systematic expansion in Newton's gravitational constant $G_N$, and can be modified when corrections higher order in $G_N$ are included. Nevertheless, the conclusion that bulk global symmetries are disallowed continues to hold even when these higher-order corrections are taken into account, assuming the corrections are small. A nontrivial global symmetry operation (if one were allowed), acting on a bulk local operator $\phi$, should modify $\phi$ by an amount $\delta \phi$ which is $O(1)$, independent of $G_N$. But we have argued that $\delta \phi = 0$ to leading order in $G_N$ (since exact correctability of $A_k$ holds to this order).  Higher-order corrections might make an $O(G_N)$ contribution  to $\delta \phi$, but these small corrections do not suffice to restore the proper nontrivial action on $\phi$ of the putative global symmetry.

Now we have found that exact bulk local symmetries cannot occur in AdS/CFT. But
what can we say about whether \textit{approximate} discrete global symmetries
are allowed? As we've discussed, finite-dimensional quantum error-correcting
codes \textit{can} have exact discrete symmetries, even though the AdS/CFT code
does not.  In this respect, discrete symmetries are essentially different than
continuous symmetries, which are disallowed by the Eastin-Knill theorem for any
finite-dimensional quantum code that can correct erasure of subsystems exactly.
Therefore, we can't expect to make general statements which are directly
analogous to \cref{thm:main-result-full} about limitations on approximate
\textit{discrete} symmetries that apply to general codes.

Nevertheless, it may be instructive to study further the properties of approximate quantum error-correcting codes which are approximately covariant with respect to a discrete symmetry. In the setting of AdS/CFT, it is of particular interest to consider the case where the local transformations $\{W_k\}$ in \cref{eq:UCFT-Wk} are either precisely or approximately logical.

\subsection{Bulk time evolution}
\label{subsec:time-evolution}

A natural symmetry arising in AdS/CFT is the time-translation invariance of the
boundary CFT, which is governed by a local Hamiltonian. Time evolution in the
bulk AdS space is a bit subtle because of the general covariance of the bulk
theory, but if we fix the gauge by choosing a preferred sequence of bulk time
slices, then time evolution in the bulk corresponds to time evolution on the
boundary. From the perspective of quantum error correction, this correspondence
is puzzling, because covariance of the AdS/CFT code with respect to time
evolution seems to be incompatible with perfect correctability of erasure on the
boundary~\cite{Preskill2000arXiv_synchronization,Hayden2017arXiv_frame}.
Indeed, the analysis of bulk global symmetries in \cref{subsec:no-global}, which
is applicable to both discrete and continuous symmetries, builds on the
observation that a boundary global symmetry operator, when restricted to a
correctable boundary subregion, preserves the code space and therefore must be a
trivial logical operator. Why can't we apply similar reasoning to the action of
the boundary Hamiltonian, concluding (incorrectly) that bulk time evolution is
trivial?

The answer hinges on a crucial distinction, emphasized
in~\cite{Harlow2018arXiv_constraints,Harlow2018arXiv_symmetries}, between global
symmetry and \textit{long-range gauge symmetry} in the bulk.
As we've noted, a ``local'' operator in the bulk is not truly local; it requires
gravitational dressing connecting it to the boundary. For the analysis of bulk
global symmetries, this dressing could be ignored, because the dressing
transforms trivially under the global symmetry. For the analysis of bulk time
evolution, the dressing cannot be ignored, because the dressing depends on the
energy-momentum of a bulk quantum state. It is the nontrivial action of the
boundary Hamiltonian on the asymptotic gravitational dressing of bulk ``local''
operators which is responsible for the bulk time evolution. Furthermore, because
the dressing can be detected by localized boundary observers, erasure of
boundary subregions can really be corrected only approximately rather than
exactly.

Our \cref{thm:main-result-full} clarifies the situation by quantifying the
incompatibility between continuous symmetries and error correction.  In the
regime of sufficiently large physical subsystems, or for a large enough number
of subsystems, covariant codes can provide arbitrarily good protection against
erasure errors.  The AdS/CFT setting fulfills both of these criteria. The
boundary theory is a field theory, which formally has an unbounded number of
local physical subsystems. Furthermore, in the ``large $N$'' limit of the CFT,
which corresponds to semiclassical gravity in the bulk, the Hilbert space
dimension of each local subsystem is very large~\cite{Maldacena1999IJTP_largeN}.

We note that holographic quantum codes, toy models of the bulk which capture
some of the properties of full blown AdS/CFT, have been constructed in which
local Hamiltonian evolution in the bulk is realized approximately by a local
Hamiltonian in the boundary theory~\cite{kohler2018complete}. However, although
these codes are approximately covariant, the boundary Hamiltonian is far from
uniform.


%% file: Discussion.tex
In this paper we have studied quantum error-correcting codes that are exactly or
approximately covariant with respect to a continuous symmetry group. A special
case of our main result applies if the logical charge operator $T_L$ which
generates a continuous symmetry is a \textit{transversal} logical operator of
the code. This means that the logical system $L$ is encoded in a physical system
$A$ which can be decomposed as a tensor product of physical subsystems $\{A_i\}$
such that erasure of each $A_i$ is correctable, and that the physical symmetry
generator $T_A$ is a sum $T_A = \sum_i T_i$, such that $T_i$ is a \textit{local}
charge operator supported only on subsystem $A_i$.

The Eastin-Knill theorem~\cite{Eastin2009PRL_restrictions,%
  Zeng2011IEEETIT_transversality,Chen2008PRA_subsystem} asserts that no quantum
error-correcting code can be covariant with respect to a continuous symmetry if
the number of physical subsystems is finite, each subsystem is
finite-dimensional, and erasure of each subsystem is exactly
correctable. However, it was shown in~\cite{Hayden2017arXiv_frame} that this
conclusion can be evaded by infinite-dimensional codes. Our main objective here
has been to clarify the properties of covariant quantum codes in which the
dimension of each physical subsystem is large but finite, and in which the
number of subsystems is large but finite.

In \cref{thm:simple-main-result-one-erasure}, we consider codes that can correct
erasure of a subsystem only \textit{approximately}, and we derive a lower bound
on the worst case entanglement infidelity $\epsilon_{\mathrm{worst}}$ that can
be achieved by the best recovery map after an erasure. In keeping with the
findings of~\cite{Hayden2017arXiv_frame}, this lower bound approaches zero when
the number $n$ of subsystems approaches infinity, or when the fluctuations of
the local charge of individual subsystems grow without bound. The idea behind
the lower bound is that, if the number of subsystems and the local charge
fluctuations are both finite, then some information about the value of the
global logical charge is available to an adversary who takes possession of a
single physical subsystem, resulting in irreversible decoherence of the logical
state. In Theorem \ref{thm:main-result-full}, we extend the result by relaxing
the assumptions. This more general theorem applies when the code is not exactly
covariant, when the logical charge operator is not exactly transversal, and when
more than one subsystem is erased.

While originally derived in the context of fault-tolerant quantum computing, the
Eastin/Knill theorem has a variety of other applications, for example to quantum
reference frames and quantum
clocks~\cite{Preskill2000arXiv_synchronization,Hayden2017arXiv_frame} (cf.\@
also recent related work~\cite{AlvaroMischa-inprep}), and to the holographic
dictionary relating bulk and boundary physics in the AdS/CFT
correspondence~\cite{Hayden2017arXiv_frame}.  When applied to these settings,
our results provide limitations on transmission of reference frames over noisy
channels, and help to clarify the relationship between bulk and boundary time
evolution for the AdS/CFT quantum code. Our lower bounds on infidelity also
apply to the recently discovered quantum codes arising in one-dimensional
translation-invariant spin chains~\cite{Brandao2017arXiv_chainAQECC}.

Our main result hinges on an interplay between the noise model and the structure
of the local charge observables.  Specifically, \Cref{thm:main-result-full}
applies under the following condition: For any term $T_\alpha$ that appears in
the physical charge $T_A=\sum_\alpha T_\alpha$, there is a nonzero probability
that \textit{all} physical subsystems supporting $T_\alpha$ are
\textit{simultaneously} lost to the environment.  One may wonder whether this
condition is really necessary---\textit{e.g.}, would a code with a $2$-local
charge operator be allowed if it could correct only a single erasure?  It turns
out that such codes do exist, showing that our condition is necessary.
As a simple example, the erasure of a single qubit is correctable for the
$[[4,2,2]]$ quantum code, but there is also a nontrivial logical operator
$Q =X\otimes X \otimes I\otimes I$ supported on the first two
qubits~\cite{Gottesman2016arXiv_fault}.  We can exponentiate this 2-local
operator to generate a logical rotation of the first logical qubit.  This
provides an example of a code that is exactly error-correcting against a single
located erasure and that is nevertheless exactly covariant with respect to a
two-local charge.

In the lower bound~\eqref{eq:overview-main-eps-worst-bound}, the range
$\Delta T_L$ of the logical charge operator and corresponding range $\Delta T_i$
of the physical charge are not directly related to the corresponding system
dimensions if the symmetry is abelian.  The situation is different when we apply
our bound~\eqref{eq:simple-bound} to codes that are covariant with respect to
the full unitary group $\UU(d_L)$. In that case, there is a minimal subsystem
dimension for each value of $\Delta T_i$, and \cref{thm:approximate-EK-main}
therefore follows from \cref{thm:simple-main-result-one-erasure}.

While \cref{thm:simple-main-result-one-erasure,thm:main-result-full} pertain to correction of erasure errors, similar conclusions should apply for more general errors. For dephasing errors in particular, the information leaked to the environment can be explicitly characterized in accord with recent results~\cite{Beny2018arXiv_constraints,dephasing_inprep}.

Our work builds on Ref.~\cite{Hayden2017arXiv_frame}, where covariant quantum codes arose in the study of reference frames; {\it i.e.}, asymmetric states which convey ``physical''
information~\cite{Bartlett2007_refframes,Kitaev2004PRA_superselection}. As shown in~\cite{Hayden2017arXiv_frame}, exact error correction of reference frames is impossible for finite-dimensional systems, yet in the real world reference frames are always finite-dimensional and communication channels are always imperfect. Nevertheless, in practice we routinely share reference frames over noisy channels, easily reaching agreement about which direction is ``up'' or what time it is; furthermore quantum technologists can distribute entanglement between nodes of a quantum network, which is possible only if the nodes share a common phase reference. Our results clarify, quantitatively, why accurate communication of reference information is achievable in practice. A quantum reference frame of sufficiently high dimensionality becomes effectively classical, quite robust against the ravages of environment noise. Examples of such systems include highly excited oscillators and rotors, Bose-Einstein condensates, superconductors, and other macroscopic phases of quantum matter. 

In metrology, quantum error correction provides a promising tool for
improving sensitivity by protecting a probe system against a noisy
environment~\cite{Preskill2000arXiv_synchronization,%
  Kessler2014PRL_qec4metrology,Arrad2014PRL_increasing,%
  Dur2014PRL_improved,Ozeri2013arXiv_Heisenberg}. However, there is a delicate
balance to achieve between error-correcting against the noise while still
being sensitive to the physical observable $H$ one wishes to measure. In order
to correct against errors, one needs to encode in an appropriate
codespace. Furthermore, in order to measure $H$, it needs to act nontrivially
within that codespace. The ability to measure $H$ directly by local
observations corresponds, in the language of this paper, to covariance of the
code with respect to the physical charge $H$. In other words, adapting our
setup to one from quantum metrology is straightforward: the goal now is to
estimate the continuous parameter $\omega$ in $H=\omega T_A$ as accurately as
possible while at the same time being able to correct against relevant
noise. Recent efforts have determined that it is possible to measure $H$ at
the Heisenberg limit using an error-correcting code if $H$ is not a sum of the
operators characterizing the correctable
noise~\cite{Zhou2018NatComm_heisenberg,%
  Demkowicz-Dobrzanski2017,layden2018spatial,Gorecki2019arXiv_multiparameter}. But 
  if the physical charge $T_A$ is a sum of local charges, 
  the Eastin-Knill
theorem 
poses a challenge
to the application of error-correcting techniques;
namely, we cannot measure what we can correct. The infinite dimensional
counterexamples of~\cite{Hayden2017arXiv_frame} show that it is nonetheless
possible to correct against local noise \textit{and} admit a charge that is a
sum of such noise operators, granted one has non-normalizable codewords. The
bounds and example codes of this paper provide a quantitative version of this
infinite-dimensional limit.

Our results suggest the possibility that one could sacrifice some error
correction precision to achieve better sensitivity with physical, i.e.\@
normalizable, states (cf.\@ also~\cite{AlvaroMischa-inprep}).  However, to
properly apply our results to quantum metrology, there are some additional
steps that need to be taken, which is the subject of ongoing follow-up work.
First, since we are trying to measure an unknown parameter (and not necessarily to protect quantum information \emph{per se}), we should account for the fact that a code is only required to
reconstruct a logical state that would yield a precise reading of said parameter.  
Second, our results are stated in terms of the worst-case entanglement fidelity,
but for applications to metrology one would prefer different figures of merit, such as the precision
at which the probe can sense magnetic fields, or the ability of a quantum clock to tell time accurately.  
Finally, it would be desirable to consider noise models that are more relevant to
quantum metrology, such as fluctuating background magnetic fields that induce
dephasing errors.  
B\'eny's characterization of approximate quantum error correction of
algebras~\cite{Beny2009TQC_conditions} provides a promising tool for addressing these challenges because one can specify precisely
which observables need to be faithfully reproduced after action by the noise and a possible recovery operation.

Approximate quantum error-correcting codes also arise naturally in many-body
quantum
systems~\cite{Brandao2017arXiv_chainAQECC,Gschwendtner2019arXiv_lowenergies}. We
anticipate that constraints on correlation functions of many-body quantum states
can be derived from the covariance properties of the corresponding codes.

Finally, the interplay of symmetry and quantum error correction has a prominent
role in the AdS/CFT holographic correspondence.  Although covariance with
respect to a continuous symmetry is incompatibile with perfect correctability of
erasure of physical subsystems for any finite-dimensional quantum code,
nevertheless we expect that in the AdS/CFT code continuous time evolution of the
boundary system corresponds to continuous time evolution of the encoded logical
bulk system. Our results relieve the tension between these two observations,
because near perfect correctability can be achieved if either the number of
physical subsystems, or the dimension of each physical subsystem, becomes very
large. Both these provisos apply to the continuum limit of a regulated
holographic boundary conformal field theory, as the number of lattice sites per
unit volume is very large in this limit, and the number of degrees of freedom
per site is also very large if semiclassical gravity accurately describes the
bulk geometry (the ``large-$N$ limit'').

Recent results indicate that not just exact continuous symmetries, but also
exact discrete symmetries, are incompatible with the quantum error correction
properties of the AdS/CFT
code~\cite{Harlow2018arXiv_constraints,Harlow2018arXiv_symmetries}. An
intriguing topic for further research will be investigation of approximate
symmetries, both continuous and discrete, in the context of quantum gravity.

\begin{acknowledgments}
During the preparation of this work, the authors became aware of an
independent effort by \'Alvaro Alhambra and Mischa Woods to analyze how well
the Eastin-Knill theorem can be evaded by allowing for a small recovery
error~\cite{AlvaroMischa-inprep}. We thank them for collegially agreeing to
synchronize our arXiv posts.

The authors also thank
C\'edric B\'eny,
Fernando Brand\~ao,
Elizabeth Crosson,
Steve Flammia,
Daniel Harlow,
Liang Jiang,
Tomas Jochym-O'Connor,
Iman Marvian,
Hirosi Ooguri,
Burak \c{S}ahino\u{g}lu, and
Michael Walter
for discussions.
PhF acknowledges support from the Swiss National Science Foundation (SNSF)
through the Early PostDoc.Mobility fellowship No.\@ P2EZP2\_165239 hosted by the
Institute for Quantum Information and Matter (IQIM) at Caltech, from the IQIM
which is a National Science Foundation (NSF) Physics Frontiers Center (NSF Grant
PHY-1733907), and from the Department of Energy (DOE) Award DE-SC0018407.
VVA acknowledges support from the Walter Burke Institute for Theoretical Physics at Caltech.
GS acknowledges support from the IQIM at Caltech, and the Stanford Institute for Theoretical Physics.
PH acknowledges support from CIFAR, DOD and the Simons Foundation.
JP acknowledges support from ARO, DOE, IARPA, NSF, and the Simons Foundation.
Some of this work was done during the 2017 program on ``Quantum Physics of
Information'' at the Kavli Institute for Theoretical Physics (NSF Grant
PHY-1748958).
\end{acknowledgments}


%% file: AppendixProofGeneralBound.tex
The proof of \cref{thm:main-result-full} is split into two lemmas.  A first
lemma deduces that the environment has access to the logical charge, to a good
approximation.
\begin{lemma}
  \label{lemma:main-result-thm-exists-environ-observable}
  Under the assumptions of \cref{thm:main-result-full}, and following the
  latter's notation, there exists an observable $Z_{C'E}$ satisfying
  \begin{align}
    \norm[\big]{ \widehat{\mathcal{N}\circ\mathcal{E}}{}^\dagger(\Ident_F\otimes Z_{C'E})
    - (T_L - \nu'\Ident_L) }_\infty
    &\leqslant \delta + \eta\ ;
    \\
    \norm{ Z_{C'E} }_\infty
    &\leqslant \max_\alpha \frac{\Delta T_\alpha}{2q_\alpha}\ ,
  \end{align}
  where $\nu' = \nu + \sum (t_\alpha^- + t_\alpha^+)/2$ and where the
  complementary channel $\widehat{\mathcal{N}\circ\mathcal{E}}{}_{L\to C'EF}$ to
  the combined encoding and noise is given
  by~\eqref{eq:compl-channel-N-E-general-alpha}.
\end{lemma}

\begin{proof}[*lemma:main-result-thm-exists-environ-observable]
  Let $\Pi_\alpha = \Ident - \Pi_\alpha^\perp$ be the projector which projects
  onto the eigenspaces of $T_\alpha$ whose corresponding eigenvalues are in the
  range $\intervalc{t_\alpha^-}{t_\alpha^+}$.  Recall that
  $t_\alpha = (t_\alpha^- + t_\alpha^+)/2$ is the midpoint of the interval
  $\intervalc{t_\alpha^-}{t_\alpha^+}$.
  Define the observables
  $\tilde{T}_\alpha = \Pi_\alpha (T_\alpha - t_\alpha\Ident)$, and observe that
  $\tilde{T}_\alpha$ has eigenvalues between $-\Delta T_\alpha/2$ and
  $+\Delta T_\alpha/2$, and hence
  $\norm{\tilde{T}_\alpha}_\infty \leqslant \Delta T_\alpha/2$.  Define the
  observable
  \begin{align}
    Z_{C'E} =
    \sum_\alpha \proj{\alpha}_{C'} \otimes (q_\alpha^{-1}\, \tilde{T}_{\alpha})\ .
    \label{eq:full-proof-Q_CpE-obs-environ}
  \end{align}
  Then, for any logical state $\sigma_L$, and writing
  $\rho_A = \mathcal{E}(\sigma_L)$,
  \begin{align}
    \tr`\big(Z_{C'E}\, \widehat{\mathcal{N}\circ\mathcal{E}}(\sigma_L))
    &= \sum \tr`\big( \tilde{T}_{\alpha} \,
    \tr_{A\setminus A_\alpha}(\rho_A))
    = \sum \tr`\big( \tilde{T}_{\alpha} \, \rho_A)
    = \sum \tr`\big( \Pi_\alpha (T_\alpha - t_\alpha\Ident) \, \rho_A)
      \nonumber\\
    &=
      \sum \tr`\big( (\Ident-\Pi_\alpha^\perp)\, (T_\alpha - t_\alpha\Ident) \, \rho_A)
    \nonumber\\
    &= \sum `\big( \tr(T_\alpha \, \rho_A) - t_\alpha
    - \tr`\big(\Pi_\alpha^\perp\,(T_\alpha-t_\alpha\Ident)\,\rho_A)
    )
      \nonumber\\
    &= \tr( T_A \rho_A) - \sum t_\alpha
    - \sum \tr(\Pi_\alpha^\perp \,(T_\alpha-t_\alpha\Ident)\,\rho_A) \ ,
  \end{align}
  thus
  \begin{align}
    \tr`\Big{\widehat{\mathcal{N}\circ\mathcal{E}}{}^\dagger(\Ident_F\otimes Z_{C'E})
    \, \sigma_L} -
    `\Big(\tr( T_A \rho_A) - \sum t_\alpha)
    = - \sum \tr(\Pi_\alpha^\perp \,(T_\alpha-t_\alpha\Ident)\,\rho_A)\ .
  \end{align}
  Noting that
  $\tr`\big(T_A\mathcal{E}(\sigma_L)) - \sum t_\alpha =
  \tr`\big(`\big[\mathcal{E}^\dagger(T_A) - \sum t_\alpha\Ident]\,\sigma_L)$, we
  have
  \begin{align}
    \abs[\Big]{
    \tr`\Big{`\Big[\widehat{\mathcal{N}\circ\mathcal{E}}{}^\dagger(\Ident_F\otimes Z_{C'E})
    - `*(\mathcal{E}^\dagger(T_A) - \sum t_\alpha\Ident)] \sigma_L}
    }
    = \abs*{\sum \tr`\big(\Pi^\perp_\alpha(T_\alpha - t_\alpha\Ident)\rho_A) }
    \leqslant \eta\ ,
    \label{eq:jkiaugyrtw9aufosjbw}
  \end{align}
  where we have used
  condition~\eqref{eq:main-result-full-condition-charge-cutoffs}.  Recall that
  for any Hermitian operator $X$, we have
  $\norm{X}_\infty = \max_\sigma \abs{\tr(X\sigma)}$ with an optimization over
  all density matrices $\sigma$.  Since~\eqref{eq:jkiaugyrtw9aufosjbw} holds for
  all $\sigma_L$, we have
  \begin{align}
    \norm[\Big]{
    \widehat{\mathcal{N}\circ\mathcal{E}}{}^\dagger(\Ident_F\otimes Z_{C'E}) -
    `*(\mathcal{E}^\dagger(T_A) - \sum t_\alpha\Ident)
    }_\infty \leqslant \eta\ .
  \end{align}
  Using the approximate charge conservation condition
  $\norm{(T_L - \nu\Ident) - \mathcal{E}^\dagger(T_A)}_\infty \leqslant \delta$
  and the triangle inequality for the infinity norm, we finally obtain
  \begin{align}
    \norm[\Big]{
    \widehat{\mathcal{N}\circ\mathcal{E}}{}^\dagger(\Ident_F\otimes Z_{C'E}) -
    `*(T_L - \nu'\Ident)
    }_\infty \leqslant \delta + \eta\ ,
  \end{align}
  setting $\nu' = \nu + \sum t_\alpha$.

  Since the infinity norm picks out the largest eigenvalue in absolute value, we
  see from~\eqref{eq:full-proof-Q_CpE-obs-environ} that
  $\norm{Z_{C'E}}_\infty = \max_\alpha q_\alpha^{-1}
  \norm{\tilde{T}_\alpha}_\infty \leqslant \max_\alpha q_\alpha^{-1}
  \Delta{T}_\alpha/2$.
\end{proof}

The second part of the proof of \cref{thm:main-result-full} is to deduce from
the environment's access to the global charge that the code performs poorly
with respect to the various entanglement fidelity measures.
We phrase this statement as a more general lemma that applies in fact to any
noise model, and can be used to bound the fixed-input entanglement fidelity for
any given fixed input state $\ket\phi_{LR}$, as long as the environment has
access to an observable which yields some information about the logical state.
In analogy with $\epsilon_{\mathrm{e}}$ and $\epsilon_{\mathrm{worst}}$, we
define for any $\ket\phi_{LR}$ and for any channel $\mathcal{N}'$,
\begin{align}
  \epsilon_{\ket\phi}(\mathcal{N}')
  = \sqrt{ 1 - F_{\ket\phi}^2(\mathcal{N}', \IdentProc[]{}) } \ .
\end{align}

This lemma can be seen as a refinement of B\'eny's characterization of
approximate error correction using operator
algebras~\cite{Beny2009TQC_conditions}.  To formulate the lemma, we define two
auxiliary quantities that depend on a state $\sigma$ and an observable $T$:
\begin{subequations}
  \begin{align}
    C^0_{\sigma,T}
    &=
      \norm[\big]{ \sigma^{1/2} \,`(T - \tr`(T\sigma)\Ident)\, \sigma^{1/2} }_1\ ,
      \label{eq:unifying-bound-auxiliary-CsigmaT-0}
    \\
    C_{\sigma,T}
    &=
      \min_\mu \, \norm[\big]{ \sigma^{1/2} \,`(T - \mu\Ident)\, \sigma^{1/2} }_1\ ,
      \label{eq:unifying-bound-auxiliary-CsigmaT}
  \end{align}
\end{subequations}
where in the second line the optimization ranges over all $\mu\in\mathbb{R}$.
Intuitively, both these quantities $C_{\sigma,T}$ pick up the average charge
absolute value (where $T$ is the charge and according to the state $\sigma$), up
to a constant charge offset $\mu$ or $\tr(T\sigma)$.  Special cases of these
quantities will be discussed in the proof of \cref{thm:main-result-full}.

\begin{lemma}
  \label{lemma:main-result-thm-environ-obs-implies-poor-code}
  Let $(\mathcal{N}\circ\mathcal{E})_{L\to A'}$ be the combined encoding and
  noise channel with total output system(s) $A'$, where both encoding and noise
  channels may be any completely positive, trace-preserving maps.  Let
  $\widehat{\mathcal{N}\circ\mathcal{E}}_{L\to E'}$ be a complementary channel
  with combined output system(s) $E'$. (In the context of
  \cref{thm:main-result-full}, we set $A'=A\otimes C$ and
  $E'=E\otimes C'\otimes F$, but this lemma holds more generally.)  Suppose that
  there exists observables $T_L$ and $Z_{E'}$ on the input and environment
  systems respectively, as well as $\nu'\in\mathbb{R}$, $\delta'\geqslant 0$,
  such that:
  \begin{align}
    \norm[\big]{ \widehat{\mathcal{N}\circ\mathcal{E}}{}^\dagger(Z_{E'})
    - (T_L - \nu'\Ident_L) }_\infty
    &\leqslant \delta'\ .
  \end{align}
  Then, for any $\ket\phi_{LR}$, both
  $\epsilon_{\mathrm{worst}}(\mathcal{N}\circ\mathcal{E})$ and
  $\epsilon_{\ket\phi}(\mathcal{N}\circ\mathcal{E})$ are lower bounded by two
  different independent bounds:
  \begin{subequations}
    \begin{align}
      \epsilon_{\mathrm{worst}}(\mathcal{N}\circ\mathcal{E})
      &\geqslant
        \epsilon_{\ket\phi}(\mathcal{N}\circ\mathcal{E})
        \geqslant
        \frac{C_{\phi_L,T_L} - \delta'}{2\norm{Z_{E'}}_\infty}
        \label{eq:unifying-bound-eps-phi}
      \\
      \epsilon_{\mathrm{worst}}(\mathcal{N}\circ\mathcal{E})
      &\geqslant
        \epsilon_{\ket\phi}(\mathcal{N}\circ\mathcal{E})
        \geqslant
        \frac{\frac12\,C^0_{\phi_L,T_L} - \delta'}{2\norm{Z_{E'}}_\infty}
        \label{eq:unifying-bound-eps-phi-C0}
    \end{align}
  \end{subequations}
  Finally, if
  $\widehat{\mathcal{N}\circ\mathcal{E}}(\cdot) = \sum q_\alpha \,
  \widehat{\mathcal{N}_\alpha\circ\mathcal{E}}(\cdot)$ for a probability
  distribution $`{q_\alpha}$ and a set of noise channels
  $`{ \mathcal{N}_\alpha }$, then for any $\ket\phi_{LR}$, the same bounds apply
  to the average of the individual error parameters corresponding to each
  erasure event:
  \begin{align}
    \sum q_\alpha \epsilon_{\ket\phi}(\mathcal{N}_\alpha\circ\mathcal{E})
    &\geqslant
      \left\{\begin{array}{l}
               \displaystyle \frac{C_{\phi_L,T_L} - \delta'}{2\norm{Z_{E'}}_\infty}
               \\[3ex]
               \displaystyle \frac{\frac12\,C^0_{\phi_L,T_L} - \delta'}{2\norm{Z_{E'}}_\infty}
             \end{array}\right.
    \label{eq:unifying-bound-avgalpha-epsphi}
  \end{align}
\end{lemma}

In summary: There are two figures of merit we are interested in,
$\epsilon_{\ket\phi}(\mathcal{N}\circ\mathcal{E})$ and
$\avg[\big]{\epsilon_{\ket\phi}(\mathcal{N}_\alpha\circ\mathcal{E})}_\alpha$,
and both are bounded from below by the same bound expressed in terms of the
auxiliary quantities~\eqref{eq:unifying-bound-auxiliary-CsigmaT-0}
and~\eqref{eq:unifying-bound-auxiliary-CsigmaT}.

\begin{proof}[*lemma:main-result-thm-environ-obs-implies-poor-code]
  We start by showing the following two statements: For any $\ket\phi_{LR}$, and
  for any state $\zeta_{E'}$, it holds that
  \begin{align}
    \delta`*( \widehat{\mathcal{N}\circ\mathcal{E}}(\phi_{LR}),
    \zeta_{E'}\otimes\phi_R )
    &\geqslant
      \frac{C_{\phi_L,T_L} - \delta'}{2\norm{Z_{E'}}_\infty}\ ;\text{ and}
      \label{eq:unifying-bound-anyphi-tr-dist-gtr-Cphi}
    \\
    \delta`*( \widehat{\mathcal{N}\circ\mathcal{E}}(\phi_{LR}),
    \rho_{E'}\otimes\phi_R )
    &\geqslant
      \frac{C^0_{\phi_L,T_L} - 2\delta'}{2\norm{Z_{E'}}_\infty}\ ,
      \label{eq:unifying-bound-anyphi-tr-dist-phiL-gtr-Cphi0}
  \end{align}
  where $\rho_{E'} = \widehat{\mathcal{N}\circ\mathcal{E}}(\phi_L)$.

  We recall the following expressions for the one-norm of any Hermitian operator $A$:
  \begin{subequations}
    \label{eq:one-norm-HermitianA-optimization-primal-dual}
    \begin{align}
      \norm{A}_1
      &= \max_{\norm{X}_\infty \leqslant 1} \tr`(XA)
      \label{eq:one-norm-HermitianA-optimization-maxX}
      \\
      &= \min_{\substack{\Delta_\pm\geqslant 0\\A = \Delta_+ - \Delta_-}}
      \tr(\Delta_+) + \tr(\Delta_-)\ ,
      \label{eq:one-norm-HermitianA-optimization-minDeltas}
    \end{align}
  \end{subequations}
  where the first optimization ranges over operators Hermitian $X$, and the
  second over positive semidefinite operators $\Delta_{\pm}$.
  We start form the left-hand side
  of~\eqref{eq:unifying-bound-anyphi-tr-dist-gtr-Cphi}.  By choosing a candidate
  $X$ in~\eqref{eq:one-norm-HermitianA-optimization-maxX} of the form
  $(Z/\norm{Z}_\infty)\otimes X'$ with $\norm{X'}_\infty\leqslant 1$, then for
  any $\ket\phi_{LR}$ and for any $\zeta_{E'}$ we have that
  \begin{align}
    \hspace*{2em}
    &\hspace*{-2em}
    \frac12\norm[\big]{
      \widehat{\mathcal{N}\circ\mathcal{E}}(\phi_{LR}) - \zeta_{E'}\otimes \phi_R
      }_1
    \nonumber\\
      &\geqslant
        \max_{\norm{X'_R}_\infty \leqslant 1}
        \frac1{2\norm{Z_{E'}}_\infty}
        \tr`*{ (Z_{E'}\otimes X_R') `*(
        \widehat{\mathcal{N}\circ\mathcal{E}}(\phi_{LR}) - \zeta_{E'}\otimes \phi_R
        ) }
    \nonumber\\
    &=
      \max_{\norm{X'_R}_\infty \leqslant 1}
      \frac1{2\norm{Z_{E'}}_\infty}
      \tr`*{
      `\big(\widehat{\mathcal{N}\circ\mathcal{E}}{}^\dagger(Z_{E'})\otimes X_R')\,
      \phi_{LR}
      - (Z_{E'} \zeta_{E'})\otimes (X_R'\,\phi_R)
      }\ ,
      \label{eq:y8gfueu}
  \end{align}
  where the optimization ranges over Hermitian operators $X_R'$ on the $R$
  system.  Making use of the main assumption of this lemma, and restricting the
  optimization to $X_R'$ such that $\tr(X_R'\phi_{R}) = 0$ yields
  \begin{align}
    \text{\eqref{eq:y8gfueu}}
    &\geqslant
      \max_{\substack{\norm{X_R'}_\infty \leqslant 1 \\ \tr(X'_R\phi_{R}) = 0}}
    \frac1{\norm{Z_{E'}}_\infty}
    `\Big[ \tr`\big{ (T_L - \nu\Ident_L)\, \tr_R(X'_R\,\phi_{LR})
    }  - \delta']\ ,
    \label{eq:rueifbosd}
  \end{align}
  using the fact that if $\norm{A-B}_\infty\leqslant \delta'$, then
  $\tr(A Y) \geqslant \tr(B Y) - \delta'\tr(Y)$ for any Hermitian $A,B$ and
  positive semidefinite $Y$, and that furthermore here
  $\tr(Y) = \tr`(X'_R\phi_{LR}) \leqslant \tr`(\phi_{LR}) = 1$.  Without loss of
  generality, we may assume that $R\simeq L$ (if $R$ is smaller, then embed it
  trivially in a larger system of same dimension as $L$; if $R$ is larger, then
  remove unused dimensions on which $\phi_R$ has no support, noting that the
  support of $\phi_R$ may not exceed the dimension of $L$).  Let
  $`*{ \ket k_L }$, $`*{\ket k_R }$ be Schmidt bases of $L$ and $R$
  corresponding to $\ket\phi_{LR}$, and recall that we have the relations
  $\ket{\phi}_{LR} = \phi_L^{1/2}\,\ket{\Phi}_{L:R} =
  \phi_R^{1/2}\,\ket{\Phi}_{L:R}$, where
  $\ket{\Phi}_{L:R} = \sum \ket k_L\otimes\ket k_R$ and where as before
  $\phi_L = \tr_R(\phi_{LR})$ and $\phi_R = \tr_L(\phi_{LR})$.  Note that for
  any operator $X'_R$, we have $X'_R\ket\Phi_{L:R} = X_L\ket\Phi_{L:R}$ where
  $X_L$ is related to $X'_R$ by a transpose with respect to the bases used to
  define $\ket\Phi_{L:R}$, which implies also
  $\norm{X_L}_\infty=\norm{X_R}_\infty$.  Consequently,
  $\tr_R(X'_R\phi_{LR}) = \tr_R(X'_R\, \phi_{L}^{1/2}\,\Phi_{L:R}\,\phi_L^{1/2})
  = \phi_L^{1/2}\,X_L\,\phi_L^{1/2}$.  Finally, note that
  $\tr(X_R'\phi_R) = \tr(X_R'\phi_{LR}) = \tr(\phi_L^{1/2}\,X_L\,\phi_L^{1/2}) =
  \tr(X_L\phi_L)$.  So we obtain
  \begin{align}
    \text{\eqref{eq:rueifbosd}}
    &= \max_{\substack{\norm{X_L}_\infty \leqslant 1 \\ \tr(X_L\phi_{L}) = 0}}
    \frac1{2\norm{Z_{E'}}_\infty}
    `\Big[ \tr`\big(\phi_L^{1/2}\,`(T_L - \nu'\Ident_L) \phi_L^{1/2}\,X_L) - \delta' ]\ .
    \label{eq:relkr92txyhoirgheuhfa}
  \end{align}
  The optimization~\eqref{eq:relkr92txyhoirgheuhfa} is a semidefinite program,
  and we proceed to compute its dual program~\cite{Watrous2009_sdps}.  In terms
  of the variables $X_L=X_L^\dagger$, $A,B\geqslant 0$, and $\mu\in\mathbb{R}$,
  and writing for short $T_L' = T_L - \nu'\Ident_L$, we have
  \begin{subequations}
  \begin{align}
    \hspace*{5em}
    &\hspace*{-5em}
    \max_{\substack{\norm{X}_\infty\leqslant 1\\ \tr(X\phi_L) = 0}}
    \tr`\big[ \phi_L^{1/2} \, T_L' \, \phi_L^{1/2} \, X_L ]
    \nonumber\\
    &=\qquad
      \def\arraystretch{1.2}
      \begin{array}[t]{@{}r@{}l@{}}
        \mathrm{maximize:}\quad
        &\tr`\big[ \phi_L^{1/2} \, T_L' \, \phi_L^{1/2} \, X_L ] \\
        \textcolor{sdpdualvar}{A:}\quad
        &X_L \leqslant \Ident_L\\
        \textcolor{sdpdualvar}{B:}\quad 
        &X_L \geqslant -\Ident_L\\
        \textcolor{sdpdualvar}{\mu:}\quad
        &\tr(X_L\phi_L) = 0
      \end{array} \\
    &=\qquad
      \def\arraystretch{1.2}
      \begin{array}[t]{@{}r@{}l@{}}
        \mathrm{minimize:}\quad
        & \tr(A) + \tr(B) \\
        \textcolor{sdpdualvar}{X_L:}\quad
        & \phi_L^{1/2} T_L' \phi_L^{1/2} = \mu\phi_L + A - B\ .
      \end{array}
          \label{eq:oijghrueysfdjk}
  \end{align}
  \end{subequations}
  Strong duality holds because of Slater's conditions. Indeed $X_L=0$ is
  strictly feasible in the primal problem; the dual is actually also strictly
  feasible by choosing (say) $\mu=0$ and $A$ and $B$ to be the positive and
  negative parts respectively of the Hermitian operator
  $\phi_L^{1/2}T_L'\phi_L^{1/2}$ plus a constant times the identity.
  For fixed $\mu$ in~\eqref{eq:oijghrueysfdjk}, we recognize the dual
  semidefinite program for the one-norm of a Hermitian
  matrix~\eqref{eq:one-norm-HermitianA-optimization-minDeltas}, and hence we
  actually obtain the same expression as
  in~\eqref{eq:unifying-bound-auxiliary-CsigmaT},
  \begin{align}
    \max_{\substack{\norm{X}_\infty\leqslant 1\\ \tr(X\phi_L) = 0}}
    \tr`\big[ \phi_L^{1/2} \, T_L' \, \phi_L^{1/2} \, X_L ]
    = \min_{\mu\in\mathbb{R}} \,
    \norm[\big]{ \phi_L^{1/2}`*(T_L' - \mu\Ident) \phi_L^{1/2} }_1
    = C_{\phi_L,T_L'}\ .
    \label{eq:unifying-bound-aux-CphiT-equivalent-optimizations}
  \end{align}
  Then
  \begin{align}
    \text{\eqref{eq:relkr92txyhoirgheuhfa}}
    &= \frac1{2\norm{Z_{E'}}_\infty}
      `\Big( \min_\mu\, \norm[\big]{\phi_L^{1/2}\,`\big(T_L - \nu'\Ident_L -
      \mu\Ident_L)\,\phi_L^{1/2}}_1
      - \delta' ) 
       = \frac{C_{\phi_L,T_L} - \delta'}{2\norm{Z_{E'}}_\infty}\ ,
  \end{align}
  noting that the constant shift $\nu'\Ident_L$ can be absorbed into the
  optimization over $\mu$.  This
  proves~\eqref{eq:unifying-bound-anyphi-tr-dist-gtr-Cphi}.

  Now we show~\eqref{eq:unifying-bound-anyphi-tr-dist-phiL-gtr-Cphi0}.
  Similarly to how we started above, we write
  \begin{align}
    \hspace*{3em}
    &\hspace*{-3em}
    \frac12\norm[\big]{
      \widehat{\mathcal{N}\circ\mathcal{E}}(\phi_{LR}) - \rho_{E'} \otimes \phi_R
      }_1
    \nonumber\\
      &\geqslant
        \max_{\norm{X'_R}_\infty \leqslant 1}
        \frac1{2\norm{Z_{E'}}_\infty}
        \tr`*{ (Z_{E'}\otimes X_R') `*(
        \widehat{\mathcal{N}\circ\mathcal{E}}(\phi_{LR} - \phi_L \otimes \phi_R)
        ) }
    \nonumber\\
      &=
        \max_{\norm{X'_R}_\infty \leqslant 1}
        \frac1{2\norm{Z_{E'}}_\infty}
        \tr`*{
        `\big(\widehat{\mathcal{N}\circ\mathcal{E}}{}^\dagger(Z_{E'})\otimes X_R')\,
        `\big(\phi_{LR} - \phi_L\otimes\phi_R) }
        \label{eq:lkuisguhkndhjsvbh}
  \end{align}
  Define $Z_L = \widehat{\mathcal{N}\circ\mathcal{E}}{}^\dagger(Z_{E'})$, and
  using the same procedure to define $\ket\Phi_{L:R}$ as above with $X_L$ in
  one-to-one correspondence with $X_R'$ via the transpose operation and with
  $\tr(X_L\phi_L) = \tr(X_R\phi_R)$, we obtain
  \begin{align}
    \text{\eqref{eq:lkuisguhkndhjsvbh}}
      &=
        \max_{\norm{X_L}_\infty \leqslant 1}
        \frac1{2\norm{Z_{E'}}_\infty} `*[
        \tr`\big{ Z_{L} \phi_L^{1/2} X_L \phi_L^{1/2} } - \tr`\big{ Z_L\phi_L\,\tr(X_L\phi_L) } ]\ .
        \label{eq:kljigfyfbkjldskglasd}
  \end{align}
  By assumption, we have
  $Z_L = \widehat{\mathcal{N}\circ\mathcal{E}}{}^\dagger(Z_{E'})
  = T'_L + \Delta_L$ with $T_L' = T_L - \nu'\Ident_L$ and
  $\norm{\Delta_L}_\infty\leqslant \delta'$, so this implies that
  \begin{align}
    \text{\eqref{eq:kljigfyfbkjldskglasd}}
    &\geqslant
      \max_{\norm{X_L}_\infty \leqslant 1}
      \frac1{2\norm{Z_{E'}}_\infty} `*[
      \tr`\big{ T'_{L} \phi_L^{1/2} X_L \phi_L^{1/2} } - \tr`\big{ T'_L\phi_L\,\tr(X_L\phi_L) }
      - 2\delta']
         \nonumber\\
       &=
         \frac1{2\norm{Z_{E'}}_\infty} `*[
         \norm[\big]{ \phi_L^{1/2} `\big(T'_L - \tr`(T'_L\phi_L)\Ident_L) \phi_L^{1/2} }_1
         - 2\delta' ]
         \nonumber\\
       &=
         \frac1{2\norm{Z_{E'}}_\infty} `*[
         \norm[\big]{ \phi_L^{1/2} `\big(T_L - \tr`(T_L\phi_L)\Ident_L) \phi_L^{1/2} }_1
         - 2\delta' ]\ ,
  \end{align}
  where in the last line we use the fact that
  $\tr(T'_L\phi_L) = \tr(T_L\phi_L) - \nu'$.  This
  proves~\eqref{eq:unifying-bound-anyphi-tr-dist-phiL-gtr-Cphi0}.

  Now, following B\'eny and Oreshkov~\cite{Beny2010PRL_AQECC}, we have the
  duality also for a fixed input state, and there exists a state $\zeta_{E'}$
  such that%
  \footnote{The statement with fixed input state is only briefly stated towards
    the end of their paper, as that claim is in fact easier to prove than their
    main theorem for the worst-case entanglement fidelity.\par}
  \begin{align}
    F_{\ket\phi}(\mathcal{N}\circ\mathcal{E}, \IdentProc{})
    = F(\widehat{\mathcal{N}\circ\mathcal{E}}(\phi_{LR}), \zeta_{E'}\otimes \phi_R)\ ,
  \end{align}
  and thus
  \begin{align}
    \epsilon_{\ket\phi}(\mathcal{N}\circ\mathcal{E})
    = P(\widehat{\mathcal{N}\circ\mathcal{E}}(\phi_{LR}), \zeta_{E'}\otimes \phi_R)\ ,
    \label{eq:yiuhofjkbpoai}
  \end{align}
  where $P(\sigma,\rho) = \sqrt{1 - F^2(\sigma,\rho)}$ denotes the ``purified
  distance'' or ``root infidelity'' between the two
  states~\cite{Tomamichel2010IEEE_Duality,PhDTomamichel2012}.  Now, using known
  inequalities between this distance measure and the trace
  distance~\cite{Tomamichel2010IEEE_Duality}, we have
  \begin{align}
    \epsilon_{\ket\phi}(\mathcal{N}\circ\mathcal{E})
    &=
      P`\big(\widehat{\mathcal{N}\circ\mathcal{E}}(\phi_{LR}), \zeta_{E'}\otimes \phi_R)
      \geqslant
      \delta`\big(\widehat{\mathcal{N}\circ\mathcal{E}}(\phi_{LR}), \zeta_{E'}\otimes \phi_R)
      \ ,
  \end{align}
  which in combination with~\eqref{eq:unifying-bound-anyphi-tr-dist-gtr-Cphi}
  proves~\eqref{eq:unifying-bound-eps-phi}.  The first part
  of~\eqref{eq:unifying-bound-eps-phi} trivially follows from the fact that
  $\epsilon_{\mathrm{worst}}(\cdot) = \max_{\ket\phi}
  \epsilon_{\ket\phi}(\cdot)$.

  From~\eqref{eq:yiuhofjkbpoai}, and using the fact that the purified distance
  cannot increase under partial trace, we find with
  $\rho_{E'} = \widehat{\mathcal{N}\circ\mathcal{E}}(\phi_{L})$,
  \begin{align}
    P(\rho_{E'}, \zeta_{E'}) \leqslant \epsilon_{\ket\phi}(\mathcal{N}\circ\mathcal{E})\ .
  \end{align}
  By triangle inequality, and using again the known inequality between trace
  distance and purified distance, we obtain
  \begin{align}
    \delta`\big(\widehat{\mathcal{N}\circ\mathcal{E}}(\phi_{LR}), \rho_{E'}\otimes\phi_R)
    &\leqslant 
    P`\big(\widehat{\mathcal{N}\circ\mathcal{E}}(\phi_{LR}), \rho_{E'}\otimes\phi_R)
      \nonumber\\
    &\leqslant 
    P`\big(\widehat{\mathcal{N}\circ\mathcal{E}}(\phi_{LR}), \zeta_{E'}\otimes\phi_R)
    +
    P`\big(\zeta_{E'}\otimes\phi_R, \rho_{E'}\otimes\phi_R)
      \nonumber\\
    &\leqslant
    2 \, \epsilon_{\ket\phi}(\mathcal{N}\circ\mathcal{E})\ .
  \end{align}
  Combining this with~\eqref{eq:unifying-bound-anyphi-tr-dist-phiL-gtr-Cphi0}
  proves~\eqref{eq:unifying-bound-eps-phi-C0}.

  Now we further assume that
  $\widehat{\mathcal{N}\circ\mathcal{E}} = \sum q_\alpha\,
  \widehat{\mathcal{N}_\alpha\circ\mathcal{E}}$ for some set of $\alpha$'s and a
  probability distribution $`{q_\alpha}$.  Then as above, invoking B\'eny and
  Oreshkov for each $\alpha$ with corresponding optimal states $\zeta_{E'}^\alpha$,
  we have
  \begin{multline}
    \sum q_\alpha
    \epsilon_\phi(\mathcal{N}_\alpha\circ\mathcal{E})
    = 
    \sum q_\alpha
    P`\big(\widehat{\mathcal{N}_\alpha\circ\mathcal{E}}(\phi_{LR}),
    \zeta_{E'}^\alpha\otimes\phi_R)
    \geqslant 
    \sum q_\alpha
    \delta`\big(\widehat{\mathcal{N}_\alpha\circ\mathcal{E}}(\phi_{LR}),
    \zeta_{E'}^\alpha\otimes\phi_R)
    \\
    \geqslant 
    \delta`*(\sum q_\alpha
    \widehat{\mathcal{N}_\alpha\circ\mathcal{E}}(\phi_{LR}), \sum
    q_\alpha\zeta_{E'}^\alpha\otimes\phi_R)
    =
    \delta`*(\widehat{\mathcal{N}\circ\mathcal{E}}(\phi_{LR}),
    \zeta'_{E'}\otimes\phi_R)\ ,
    \label{eq:ioiugrehbkjhois}
  \end{multline}
  using the joint convexity of the trace distance and defining
  $\zeta'_{E'} = \sum q_\alpha \zeta_{E'}^\alpha$.  Directly
  invoking~\eqref{eq:unifying-bound-anyphi-tr-dist-gtr-Cphi} then proves the
  first bound in~\eqref{eq:unifying-bound-avgalpha-epsphi}.  We also have
  $\text{\eqref{eq:ioiugrehbkjhois}} \geqslant \delta`*(\rho_{E'},
  \zeta'_{E'})$, and hence by triangle inequality
  \begin{align}
    \delta`*(\widehat{\mathcal{N}\circ\mathcal{E}}(\phi_{LR}),
    \rho_{E'}\otimes\phi_R)
    \leqslant 2\, \sum q_\alpha
    \epsilon_\phi(\mathcal{N}_\alpha\circ\mathcal{E})\ .
  \end{align}
  Combining with~\eqref{eq:unifying-bound-anyphi-tr-dist-phiL-gtr-Cphi0} then
  yields the second bound in~\eqref{eq:unifying-bound-avgalpha-epsphi}.
\end{proof}

We may now combine these two lemmas to finally prove
\cref{thm:main-result-full}.

\begin{proof}[*thm:main-result-full]
  Thanks to \cref{lemma:main-result-thm-exists-environ-observable} there
  exists $Z_{C'E}$ and $\nu'\in\mathbb{R}$ such that
  \begin{subequations}
    \begin{gather}
      \norm[\big]{\widehat{\mathcal{N}\circ\mathcal{E}}{}^\dagger(Z_{C'E}) - (T_L -
      \nu'\Ident_L)}_\infty
      \leqslant \delta + \eta\ ;
      \\
      \norm{Z_{C'E}}_\infty 
      \leqslant \max_\alpha \frac{\Delta T_\alpha}{2 q_\alpha}\ .
      \label{eq:tyfiguihfwklkad}
    \end{gather}
  \end{subequations}
  We may directly plug this observable into
  \cref{lemma:main-result-thm-environ-obs-implies-poor-code} to deduce that
  the bound~\eqref{eq:unifying-bound-eps-phi} applies to our approximately
  covariant code.  We now need to compute the form of the bound for the
  particular quantities $\epsilon_{\mathrm{e}}(\mathcal{N}\circ\mathcal{E})$,
  $\avg{\epsilon_{\mathrm{e}}(\mathcal{N}^\alpha\circ\mathcal{E})}_\alpha$ and
  $\epsilon_{\mathrm{worst}}(\mathcal{N}\circ\mathcal{E})$.

  First, let $\ket\phi_{LR} = \ket{\hat\phi}_{LR}$ be the maximally entangled
  state between $L$ and $R\simeq L$.  Then by definition, and recalling the
  alternative expression
  in~\eqref{eq:unifying-bound-aux-CphiT-equivalent-optimizations} for
  $C_{\phi_L,T_L}$ with a maximization, we have
  \begin{align}
    C_{\phi_L,T_L} =
    C_{\Ident_L/d_L,T_L} =
    \frac1{d_L} \min_{\mu} \, \norm[\big]{T_L - \mu\Ident_L}_1
    =
    \frac1{d_L} \,
    \max_{\substack{\norm{X}_\infty \leqslant 1 \\ \tr(X)=0}}
    \tr`\big[ T_L X ]\ .
    \label{eq:yihkjfdlsjfhpiuovuq}
  \end{align}
  Let $\mu(T_L)$ denote a median eigenvalue of $T_L$ counted with multiplicity,
  which implies the following.  Let $`{ \ket{k}_L }$ for $k=1,\ldots,d_L$ be an
  eigenbasis of $T_L$ with its elements arranged such that the eigenvalues of
  $T_L$ are nonincreasing in $k$,
  $\dmatrixel{1}{T_L} \geqslant \dmatrixel{2}{T_L} \geqslant \cdots \geqslant
  \dmatrixel{d_L}{T_L}$.  Let
  \begin{align}
    P_+ &= \sum_{k=1}^{\lfloor d_L/2\rfloor} \proj{k}_L\ ;
    &
    P_- &= \sum_{k=\lceil d_L/2\rceil+1}^{d_L} \proj{k}_L\ ,
  \end{align}
  noting that the two projectors are orthogonal and that
  $\rank(P_+) = \tr(P_+) = \tr(P_-) = \rank(P_-)$.  That is, we divide all basis
  vectors into two sets of equal size, corresponding to the smallest eigenvalues
  and the largest eigenvalues respectively, possibly leaving out the middle
  basis vector if the space dimension is odd.  Then, the eigenvalues
  corresponding to the eigenbasis vectors included in $P_+$ (respectively,
  $P_-$) are all greater than or equal to (respectively less than or equal to)
  $\mu(T_L)$.  If $d_L$ is odd, then the basis vector that was left out
  corresponds to the eigenvalue $\mu(T_L)$.

  Now set $X = P_+ - P_-$, satisfying $\norm{X}_\infty\leqslant 1$.  We have
  $\norm{T_L - \mu(T_L)\Ident}_1 = \tr`\big[X (T_L - \mu(T_L)\Ident)]$: Indeed,
  the one-norm is equal to the sum of the absolute values of the eigenvalues of
  its argument, which is precisely taken care of by our careful choice of $X$.
  Then
  $\norm{T_L - \mu(T_L)\Ident}_1 = \tr`\big(X T_L) - \mu(T_L)\tr(X) = \tr`\big(X
  T_L)$ because $\tr(X)=0$ by construction.  Now because both $\mu(T_L)$ and $X$
  are optimization candidates in~\eqref{eq:yihkjfdlsjfhpiuovuq}, we have
  \begin{align}
    \frac1{d_L}\,\norm[\big]{T_L - \mu(T_L)\Ident}_1
    \geqslant
    C_{\Ident/d_L,T_L}
    \geqslant
    \frac1{d_L}\, \tr`*(X  T_L)
    = 
    \frac1{d_L}\,\norm[\big]{T_L - \mu(T_L)\Ident}_1\ ,
  \end{align}
  which implies that
  $C_{\Ident/d_L,T_L} = d_L^{-1}\,\norm[\big]{T_L - \mu(T_L)\Ident}_1$.
  \Cref{lemma:main-result-thm-environ-obs-implies-poor-code} states that
  $`\big[C_{\Ident/d_L,T_L} - (\delta+\eta)]/(2\norm{Z_{C'E}}_\infty)$ is a
  lower bound both to $\epsilon_{\mathrm{e}}(\mathcal{N}\circ\mathcal{E})$ and
  to $\avg{\epsilon_{\mathrm{e}}(\mathcal{N}^\alpha\circ\mathcal{E})}_\alpha$,
  which proves~\eqref{eq:main-thm-bound-avg-entgl-fid} as we recall the
  property~\eqref{eq:tyfiguihfwklkad}.
  
  That the norm term in~\eqref{eq:main-thm-bound-avg-entgl-fid} can be replaced
  by $\norm{T_L - \tr(T_L)\Ident_L/d_L}_1/(2d_L)$ follows from the alternative
  bound in \cref{lemma:main-result-thm-environ-obs-implies-poor-code}, stating
  that
  $`\big[(C^0_{\Ident/d_L,T_L}/2) - (\delta+\eta)]/(2\norm{Z_{C'E}}_\infty)$
  [cf.~\eqref{eq:unifying-bound-auxiliary-CsigmaT-0}] is also a lower bound to
  both error parameters considered in~\eqref{eq:main-thm-bound-avg-entgl-fid}.

  For $\epsilon_{\mathrm{worst}}(\mathcal{N}\circ\mathcal{E})$, we get to pick
  $\ket\phi_{LR}$ freely and this will yield a valid bound.  Let
  $\ket{\psi^\pm}_L$ be eigenstates of $T_L$ corresponding to the maximal and
  minimal eigenvalues $T_L$, respectively, with
  $\dmatrixel{\psi^+}{T_L} - \dmatrixel{\psi^-}{T_L} = \Delta T_L$.  Now choose
  two arbitrary orthogonal states $\ket{\pm}_R$ on $R$ and set
  \begin{align}
    \ket\phi_{LR} = \frac1{\sqrt 2}`*[ \ket{\psi^+}_L \ket{+}_R + \ket{\psi^-}_L \ket{-}_R ]\ ,
  \end{align}
  with $\phi_L = \Pi_L/2$, where we write
  $\Pi_L = \proj{\psi^+} + \proj{\psi^-}$.  Recall the alternative expression
  in~\eqref{eq:unifying-bound-aux-CphiT-equivalent-optimizations} for
  $C_{\phi_L,T_L}$ with a maximization.  We can choose as candidate
  $X_L = \proj{\psi^+}_L - \proj{\psi^-}_L$, since we have indeed
  $\tr`(\phi_L X_L) = 0$ and $\norm{X_L}_\infty\leqslant 1$, and we obtain
  \begin{align}
    C_{\phi_L,T_L} \geqslant \frac12\tr`\big(\Pi_L X_L \Pi_L T_L)
    = \frac{\Delta T_L}{2}\ .
  \end{align}
  \Cref{lemma:main-result-thm-environ-obs-implies-poor-code} then asserts that
  \begin{align}
    \epsilon_{\mathrm{worst}}(\mathcal{N}\circ\mathcal{E})
    \geqslant
    \frac{\Delta T_L/2 - \delta - \eta}%
    {\max_\alpha \Delta T_\alpha/q_\alpha}\ ,
  \end{align}
  where we recall~\eqref{eq:tyfiguihfwklkad}. This
  proves~\eqref{eq:main-thm-bound-worst-entgl-fid}.
\end{proof}

At this point we comment on
Condition~\eqref{eq:main-result-full-condition-charge-cutoffs} in the statement
of \cref{thm:main-result-full}.  It may look a bit awkward, but its meaning
is intuitively simple: First, we need to shift the charge values to center them
at zero for each $\alpha$ for our proof.  Second, we need to make sure that if
we project any codeword into the given range of physical charge values for each
$\alpha$, then the total error we make when attempting to determine the
expectation value of the actual (possibly unbounded) charge observable $T_A$ is
small.  In practice, this just means that the part of the codewords outside of
the given range of charge values only has a small contribution to the total
expectation value of charge.  For convenience we may use the following
simplified criterion, where we simply fix a charge cut-off value $t$:
\begin{proposition}
  \label{prop:criterion-phys-charge-codewords}
  Consider $V_{L\to A}$, $T_L$, $K$ and $T_A=\sum_{\alpha\in K} T_\alpha$ as in
  \cref{thm:main-result-full}.  Let $t>0$.  Set
  $t_\alpha^+ = - t_\alpha^- = t$ and define $\Pi_\alpha, \Pi_\alpha^\perp$ as
  in the statement of \cref{thm:main-result-full}.  Let
  $\{ \ket{\phi_\alpha^{t',j}} \}$ be an eigenbasis of $T_\alpha$ corresponding
  to eigenvalues $t'$ with a possible degeneracy index $j$.
  Suppose that there is an $\eta'\geqslant 0$ such that for any logical state
  $\psi_L$ and for any $\alpha$,
  \begin{align}
    \sum_{t',j:\; \abs{t'}>t}
    \, \abs{t'}\, \dmatrixel{\phi_{\alpha}^{t',j}}{ \rho_\alpha } \leqslant \eta'\ ,
    \label{eq:criterion-phys-charge-codewords-sumtp}
  \end{align}
  where we write $\rho_\alpha = \tr_{A\setminus A_\alpha}(V \psi_L V^\dagger)$
  and where the sum ranges over the eigenstate labels $(t',j)$ such that
  $\abs{t'}>t$.
  Then, condition~\eqref{eq:main-result-full-condition-charge-cutoffs} is
  satisfied with $\eta=\abs{K}\,\eta'$, and furthermore $\Delta T_\alpha = 2t$
  for all $\alpha$.
\end{proposition}

\begin{proof}[*prop:criterion-phys-charge-codewords]
  We have $t_\alpha=(t_\alpha^+ + t_\alpha^-)/2 = 0$. 
  For any $\psi_L$, calculate
  \begin{align}
    \abs[\Big]{\sum \tr`\big(\Pi_\alpha^\perp T_\alpha\, V\psi_L V^\dagger) }
    &\leqslant \sum \abs{  \tr`\big(\Pi_\alpha^\perp T_\alpha\, V\psi_L V^\dagger) }
      \nonumber\\
    &\leqslant \sum_\alpha
    \abs*{ \sum\nolimits_{t',j:\;\abs{t'}>t} t'\,
      \dmatrixel{\phi_\alpha^{t',j}}{\tr_{A\setminus A_\alpha}(V\psi_L V^\dagger)} }
      \nonumber\\
    &\leqslant \sum_\alpha
      \sum\nolimits_{t',j:\;\abs{t'}>t} \abs{t'}\,
      \dmatrixel{\phi_\alpha^{t',j}}{\tr_{A\setminus A_\alpha}(V\psi_L V^\dagger)}
      \nonumber\\
    &\leqslant \sum_\alpha \eta' \leqslant \abs{K}\,\eta'\ .
  \end{align}
  Note by the way that the left hand side
  of~\eqref{eq:criterion-phys-charge-codewords-sumtp} is exactly
  $\tr(\Pi_\alpha^\perp\,\abs{T_\alpha} V \psi_L V^\dagger)$.
\end{proof}


%% file: CorrelationBound.tex
In this section we present an alternative strategy for proving the
bound~\eqref{eq:simple-bound-e}, by studying the connected correlation functions
between the physical subsystems and the logical information.

The covariance of the codes can be seen as a linear constraint, which can be
easily employed to obtain a second order constraints. To start, we again assume
the simpler case of isometric encoding.  We construct the state corresponding to
the encoding isometry $V_{L\rightarrow A}$ by injecting a maximally entangled
state $\ket{\hat \phi}_{LR}$ to $V_{L\rightarrow A}$ (\cref{fig:corrstate}):
\begin{align}
  \ket \Psi_{AR} = V\ket{\hat \phi}_{LR}\ .
\end{align}

\begin{figure}
  \centering
  \includegraphics[width=3.5in]{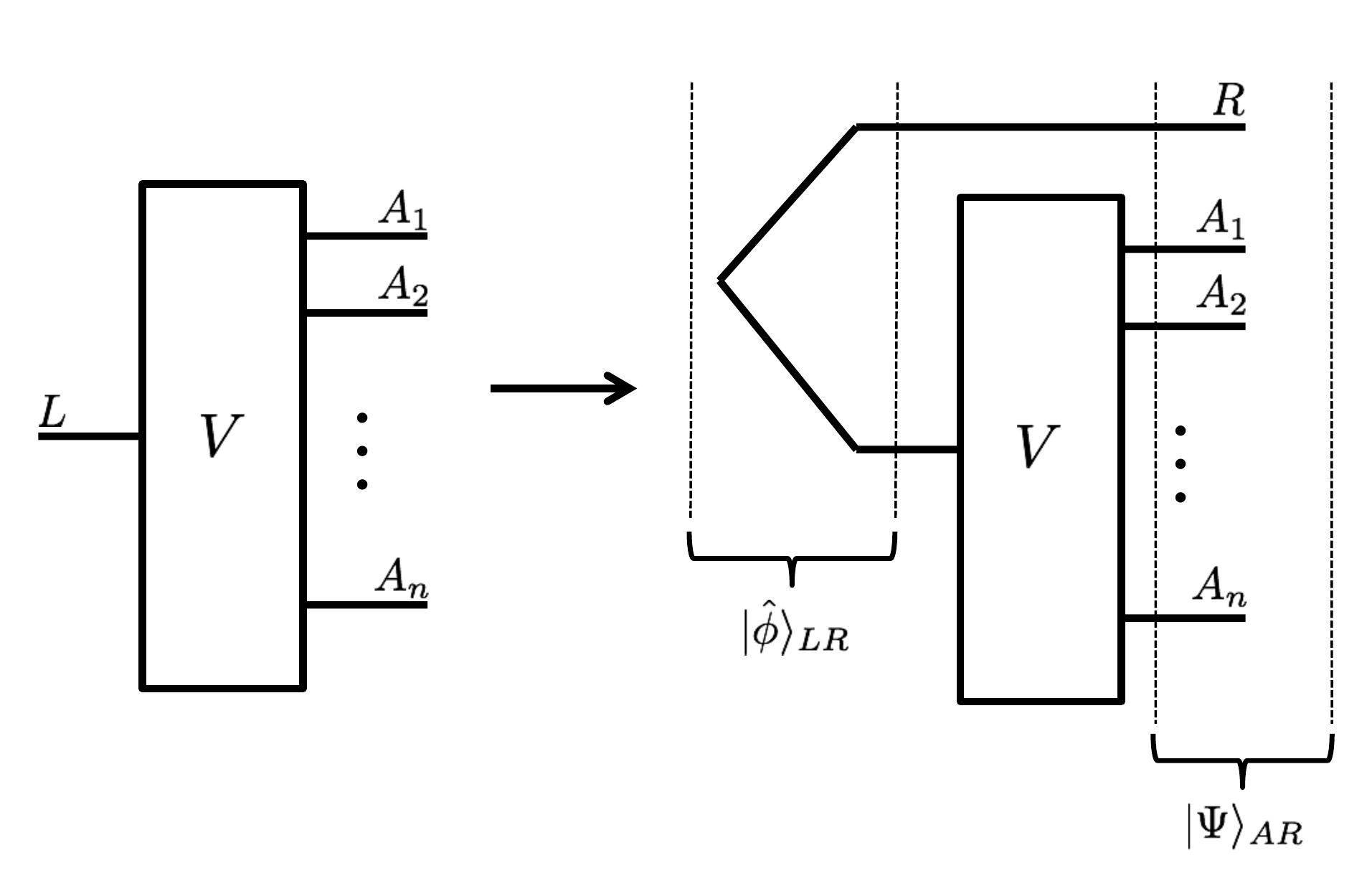}
  \caption{Depiction of the construction of the state $\ket{\Psi}_{AR}$ by injecting the maximally entangled state $\ket{\hat \phi}_{LR}$ into the encoding isometry $V_{L\rightarrow A}$.}
  \label{fig:corrstate}
\end{figure}

We have
$T_A \ket \Psi_{LA}= T_A V \ket {\hat \phi}_{LR}= V (T_L-\nu\Ident_L)
\ket{\hat\phi}_{LR}$ for some constant $\nu$.  Define
$T_R = (T_L-\nu\Ident_L)^T$ where the transpose is taken as a matrix ignoring
the Hilbert space label; this ensures that
$(T_L-\nu\Ident_L) \ket{\hat \phi}_{LR} = T_R \ket{\hat \phi}_{LR}$.  Therefore,
the covariance of $V$ translates to the invariance of $\ket \Psi$:
\begin{align}
  \label{eq:linear_invariance}
  \left(\sum_{i=1}^n T_{A_i}\right) \ket{\Psi}_{RA}
  = T_A \ket{\Psi}_{RA}=T_R\ket{\Psi}_{RA}\ .
\end{align}
We define the \emph{connected correlator} between two operators $A,B$ as
\begin{align}
  \langle A, B \rangle := \tr{(A B\Psi)}-\tr{(A\Psi)}\tr{(B\Psi)}\ .
\end{align}
Consider an arbitrary operator $X_R$.  It be seen
from~\eqref{eq:linear_invariance} that
\begin{align*}
  \langle X_R,T_R \rangle =\sum_{i=1}^n {\langle X_R,T_{A_i} \rangle }\ .
\end{align*}
Using the triangle inequality, we obtain
\begin{align}
  \label{eq:correlation_function}
  \abs{\langle X_R,T_R \rangle} \leq \sum_{i=1}^n \abs{\langle X_R,T_{A_i} \rangle}
  \qquad \text{ for all }X_R.
\end{align}
Although the derivation of~\eqref{eq:correlation_function} is very simple, it
provides a general lower bound to the amount of correlations between the
reference system and the physical subsystems, from which we can draw physical
consequences.
The correlation functions measure how close the state $\Psi_{RA_i}$ is to the
product state $\Psi_R \ot\Psi_{A_i}$:
\begin{align}
  \label{eq:corr_func_rhs}
  \abs*{ \langle X_R,T_{A_i} \rangle}
  &=\abs*{ \tr[X_R T_{A_i} (\Psi_{RA_i}-\Psi_R \ot\Psi_{A_i})] } \nonumber\\
  &\leqslant \norm{X_R}_\infty \, \norm{T_{A_i}}_\infty \,
    \norm{ \Psi_{RA_i}-\Psi_R \ot\Psi_{A_i} }_1\ ,
\end{align}
where we used H\"older's inequality.  We can replace
$T_{A_i}\to T_{A_i} - t\Ident$ in~\eqref{eq:corr_func_rhs} without changing the
left hand side of the inequality as
$\langle X_R, T_{A_i} - t\Ident \rangle = \langle X_R, T_{A_i} \rangle$:
\begin{align}
  \abs*{ \langle X_R,T_{A_i} \rangle}
  &\leqslant  \norm{X_R}_\infty \, \norm{ T_{A_i} - t\Ident}_\infty \,
    \norm{ \Psi_{RA_i}-\Psi_R \ot\Psi_{A_i} }_1
  \nonumber\\
  &= \frac12 \norm{X_R}_\infty \, \Delta T_{A_i} \,
    \norm{ \Psi_{RA_i}-\Psi_R \ot\Psi_{A_i} }_1\ ,
  \label{eq:corr_func_rhs-2-DeltaTAi}
\end{align}
where the second line follows by a suitable choice of $t$, and where
$\Delta T_{A_i}$ is the difference between the maximal and minimal eigenvalue of
$T_{A_i}$.

The accuracy to which the code $V$ can correct against errors is precisely
determined by how close $\Psi_{RA_i}$ is to a product state.  Indeed, consider
the noise channel $\mathcal{N}^{i}_{A\to A}$
in~\eqref{eq:noise-map-general-alpha--Nalpha} that erases the system $A_i$.  By
B\'eny and Oreshkov~\eqref{eq:Beny-Oreshkov-f-fixedinput}, we have
\begin{align}
  \epsilon_{\mathrm{e}}(\mathcal{N}^i\circ\mathcal{E})
  &= \min_{\zeta} \sqrt{ 1 - F^2`\big(
    \widehat{\mathcal{N}^i\circ\mathcal{E}}(\hat\phi_{LR}), \zeta\otimes\Psi_R
    ) } 
    \nonumber\\
  &\geqslant  \min_{\zeta}\frac12\norm[\big]{
    \widehat{\mathcal{N}^i\circ\mathcal{E}}(\hat\phi_{LR}) - \zeta\otimes\Psi_R
    }_1
    \nonumber\\
  &=  \min_{\zeta}\frac12\norm[\big]{ \Psi_{RA_i} - \zeta_{A_i}\otimes\Psi_R }_1
\end{align}
where $\widehat{\mathcal{N}^i\circ\mathcal{E}}(\hat\phi_{LR}) = \Psi_{RA_i}$ and
$\Psi_R = \Ident_R/d_L$, and where we have used the known relation
$\delta(\cdot,\cdot)\leqslant\sqrt{1 - F^2(\cdot,\cdot)}$ between the trace
distance and the fidelity.  Because the trace distance cannot increase under the
partial trace, and if we set $\zeta_{A_i}$ to be the optimal state in the
expression above, also have
$(1/2)\norm{ \Psi_{A_i} - \zeta_{A_i} }_1 \leqslant
\epsilon_{\mathrm{e}}(\mathcal{N}^i\circ\mathcal{E})$ and thus by triangle
inequality,
\begin{align}
  \frac12\norm[\big]{ \Psi_{RA_i} - \Psi_{A_i}\otimes\Psi_R }_1
  \leqslant
  2\epsilon_{\mathrm{e}}(\mathcal{N}^i\circ\mathcal{E})\ .
  \label{eq:trdist-PsiARi-product-inequality-epsilon-e}
\end{align}

It remains to combine~\eqref{eq:trdist-PsiARi-product-inequality-epsilon-e}
with~\eqref{eq:corr_func_rhs} and~\eqref{eq:correlation_function} and to choose
the best possible $X_R$ to get our final result.

\begin{theorem}
  \label{thm:corr-thm}
  The individual entanglement fidelities of recovery of a covariant code
  $\mathcal{E}(\cdot) = V (\cdot)V^\dagger$ against single erasures at known
  locations satisfy the following inequality:
  \begin{align}
    \label{eq:corrbound1}
    \frac{1}{2d_L} \norm*{ T_L - \tr(T_L)\frac{\Ident}{d_L}}_1
    \leq  
    \sum_{i=1}^n \Delta T_{i} \, \epsilon_{\mathrm{e}}(\mathcal{N}^i\circ\mathcal{E})\ .
  \end{align}
  Furthermore, this can be used to show that
  \begin{align}
    \label{eq:corrbound2}
    \epsilon_{\mathrm{e}}(\mathcal{N}\circ\mathcal{E})
    \geq
    \frac{1}{2d_L} \frac{\norm{T_L - \tr(T_L)\Ident/d_L }_1 }{\max_i q_i^{-1} \Delta T_i} \ .
  \end{align}
\end{theorem}
Note that $T_L-\tr(T_L)\Ident/d_L$ is just a shift of $T_L$ by a multiple of
identity to make it traceless.  Therefore, $\norm{ T_L - \tr(T_L)\Ident/d_L}_1$
is a 1-norm measure for the spread of eigenvalues of $T_L$.  The bounds of
\cref{thm:corr-thm} and \cref{eq:simple-bound-e} have a very similar nature.

\begin{proof}[*thm:corr-thm]
  We start with the correlator in the left hand side
  of~\eqref{eq:correlation_function}:
  \begin{align}
    \langle X_R, T_R \rangle
    &= \tr(X_R T_R \Psi_R) - \tr(X_R \Psi_R) \tr(T_R\Psi_R)
    = \frac1{d_L}\tr`*(X_R `*[ T_R - \frac{\tr(T_R)}{d_L}\Ident ]) \ .
  \end{align}
  Now, choose the optimal $X_R$ such that $\norm{X_R}_\infty\leqslant 1$ and
  that
  $\norm{T_R - \tr(T_R) \Ident/d_L}_1 = \tr`*[X_R`*(T_R-\tr(T_R)\Ident/d_L)]$.
  Plugging into~\eqref{eq:correlation_function}, and combining
  with~\eqref{eq:corr_func_rhs-2-DeltaTAi}
  and~\eqref{eq:trdist-PsiARi-product-inequality-epsilon-e}, immediately
  gives~\eqref{eq:corrbound1}.

  Furthermore from~\eqref{eq:corrbound1} we have
  \begin{align}
    \frac1{d_L}\norm*{ T_L - \tr(T_L) \frac\Ident{d_L} }_1
    &\leqslant
    \sum (q_i^{-1} \Delta T_i) (q_i \epsilon_{\mathrm{e}}(\mathcal{N}^i\circ\mathcal{E}))
    \nonumber\\
    &\leqslant `*(\max_i \, (q_i^{-1} \Delta T_i) )
      \sum q_i \epsilon_{\mathrm{e}}(\mathcal{N}^i\circ\mathcal{E})\ .
      \label{eq:fkdiugyfhodklsajbhsi}
  \end{align}
  By convexity of $x\mapsto x^2$, and by \cref{lem:inv-fid-to-global-fid}, we
  have
  \begin{align}
    \sum q_i \epsilon_{\mathrm{e}}(\mathcal{N}^i\circ\mathcal{E})
    &\leqslant
      \sqrt{ \sum q_i \epsilon^2_{\mathrm{e}}(\mathcal{N}^i\circ\mathcal{E}) }
      = \epsilon_{\mathrm{e}}(\mathcal{N}\circ\mathcal{E})\ .
      \label{eq:uioekbodsuid}
  \end{align}
  Combining~\eqref{eq:uioekbodsuid} with~\eqref{eq:fkdiugyfhodklsajbhsi}
  proves~\eqref{eq:corrbound2}.
\end{proof}


%% file: AppendixCriteriaCodes.tex
When we come up with a new code, how can we show that it forms an
$\epsilon$-approximate error-correcting code against erasures at known
locations?
Here we provide a criterion that, when it can be applied, certifies that a given
code performs well.

Let $L$ be the logical space and $A$ be the physical space, and consider an
encoding operation $\mathcal{E}_{L\to A}$ that can be any completely positive,
trace-preserving map.  Note that in the case of a more general noise model, $A$
does not necessarily have to be composed of several subsystems.  Consider a
collection of noise channels $\{ \mathcal{N}^\alpha \}$ and probabilities
$\{ q_\alpha \}$.  We assume that the environment applies a random noise channel
from this set with the corresponding probability, while providing a record of
which noise channel was applied in a separate register $C$. The overall noise
channel that is applied by the environment is then
\begin{align}
  \mathcal{N}_{A\to AC}(\cdot) =
  \sum q_\alpha \proj\alpha_C \otimes \mathcal{N}^\alpha_{A\to A}(\cdot)\ .
\end{align}
Given complementary channels $\widehat{\mathcal{N}^\alpha\circ\mathcal{E}}$ of
$\mathcal{N}^\alpha\circ\mathcal{E}$, we can construct a complementary channel
of $\mathcal{N}\circ\mathcal{E}$ as
\begin{align}
  \widehat{\mathcal{N}\circ\mathcal{E}}_{A\to C'E}(\cdot)
  = \sum q_\alpha \proj{\alpha}_{C'} \otimes
  \widehat{\mathcal{N}^\alpha\circ\mathcal{E}}(\cdot)\ ,
\end{align}
with an additional register $C'$ and where the outputs of the individual
complementary channels for each $\alpha$ are embedded into a system $E$.

We fix any basis $\{\ket{x}_L\}$ of $L$, and we define for each $\alpha$ the
operators
\begin{align}
  \rho_\alpha^{x,x'} = \widehat{\mathcal{N}^\alpha \circ\mathcal{E}}
  (\ketbra{x}{x'}_L)\ .
  \label{eq:criterion-cert-aqecc-rho-xxp}
\end{align}
Note that $\rho_\alpha^{x,x}$ is a quantum state for each $\alpha$ and for each
$x$, but that $\rho_\alpha^{x,x'}$ is not necessarily even Hermitian for $x\neq x'$.

For an isometric encoding $\mathcal{E}$, and in the noise $\mathcal{N}$
acts by erasing a collection of subsystems labeled by $\alpha$ and chosen with
probability $q_\alpha$, the operators $\rho_\alpha^{x,x'}$ are simply the
reduced operators on the sites labeled by $\alpha$ of the logical operator
$\ketbra{x}{x'}$:
\begin{align}
  \rho_\alpha^{x,x'} = \tr_{A\setminus A_\alpha}`*(\mathcal{E}(\ketbra{x}{x'}))\ .
\end{align}

\begin{proposition}
  \label{prop:new-criterion-cert-aqecc-2}
  Assume that there exists $\nu,\epsilon' \geqslant 0$, and that there exists a
  quantum state $\zeta_\alpha$ for each $\alpha$, such that for all $\alpha$,
  \begin{subequations}
    \label{eq:new-criterion-cert-2-assumption-rho-x-xp}
    \begin{alignat}{2}
      F(\rho_\alpha^{x,x},\zeta_\alpha) &\geqslant \sqrt{1 - {\epsilon'}^2}
      &
      &\quad \text{for all}~x \ ;\quad\text{and}
        \label{eq:new-criterion-cert-2-assumption-rho-x-x-purif-close}
        \\
      \norm{\rho_\alpha^{x,x'}}_1 &\leqslant \nu
      &
      &\quad \text{for all}~x\neq x'\ .
        \label{eq:new-criterion-cert-2-assumption-rho-x-xp-small}
    \end{alignat}
  \end{subequations}
  Then $\mathcal{E}_{L\to A}$ is an approximate error-correcting code against
  the noise $\mathcal{N}$, with approximation parameter
  \begin{align}
    \epsilon_{\mathrm{worst}}(\mathcal{N}\circ\mathcal{E})
    &\leqslant \epsilon' + d_L\sqrt{\nu} \ ,
  \end{align}
  where $d_L$ is the dimension of the logical system $L$.
\end{proposition}

\begin{proof}[*prop:new-criterion-cert-aqecc-2]
  Let
  \begin{align}
    \zeta_{C'E} =
    \sum_{\alpha} q_\alpha \proj{\alpha}_{C'}\otimes \zeta_\alpha\ .
  \end{align}
  Using the B\'eny-Oreshkov property~\eqref{eq:Beny-Oreshkov-f}, the proof
  strategy is to find a lower bound to the entanglement fidelity of the channel
  $\widehat{\mathcal{N}\circ\mathcal{E}}$ to the constant channel
  $\mathcal{T}_\zeta$ outputting the state $\zeta_{C'E}$ defined above.

  Consider a reference system $R\simeq L$, and let $\{\ket{x}_R\}$ be any fixed basis
  of $R$.  Let $\ket{\Phi}_{L:R} = \sum_x \ket{x}_L\otimes\ket{x}_R$.
  For any state $\ket{\sigma}_{LR}$, there exists a complex matrix $B_R$ such
  that $\ket{\sigma}_{LR} = B_R\,\ket{\Phi}_{L:R}$ and
  $\sigma_R = \tr_L(\sigma_{LR}) = B_R B_R^\dagger$ (choose
  $B_R = \sum_{x,x'} \braket{x,x'}{\sigma}_{LR} \ketbra{x'}{x}_R$). Note that
  $\norm{B_R B_R^\dagger}_\infty = \norm{B_R^\dagger B_R}_\infty \leqslant 1$.
  We have
  \begin{align}
    (\widehat{\mathcal{N}\circ\mathcal{E}}\otimes\IdentProc[R]{})(\sigma_{LR})
    &= B_R\, \widehat{\mathcal{N}\circ\mathcal{E}}(\Phi_{L:R})\,
      B_R^\dagger
      \nonumber\\
    &= \sum_\alpha q_\alpha \proj\alpha_{C'} \otimes
      (B_R\; \widehat{\mathcal{N}^\alpha\circ\mathcal{E}}(\Phi_{L:R}) \, B_R^\dagger)
      \nonumber\\
    &= \sum_{\alpha,x,x'} q_\alpha\,\proj{\alpha}_{C'}\otimes
      `\big( B_R\,  \rho^{\alpha}_{ER} \, B_R^\dagger ) \ ,
  \end{align}
  where we have defined for each $\alpha$ the positive semidefinite operator
  \begin{align}
    \rho_{ER}^\alpha
    = \widehat{\mathcal{N}^\alpha\circ\mathcal{E}}(\Phi_{L:R})
    = \sum_{x,x'} \rho_\alpha^{x,x'}\otimes\ketbra{x}{x'}_R
    \ .
  \end{align}
  While the $\rho_{ER}^\alpha$'s are positive semidefinite, they are not
  normalized to unit trace as proper quantum states.  Recalling that the
  fidelity is jointly concave, we have
  \begin{align}
    F`\big(\widehat{\mathcal{N}\circ\mathcal{E}}(\sigma_{LR}),
    \zeta_{C'E}\otimes\sigma_R)
    &= F`*( \sum\nolimits_\alpha q_\alpha \proj{\alpha}_{C'}
    \otimes (B_R \rho^\alpha_{ER} B_R^{\dagger}),
    \sum\nolimits_\alpha q_\alpha \proj{\alpha}_{C'}\otimes
    \zeta_\alpha\otimes \sigma_R )
    \nonumber\\
    &\geqslant
      \sum_\alpha\, q_\alpha\,
      F`\big( B_R \rho^\alpha_{BR} B_R^\dagger, \zeta_\alpha\otimes\sigma_R)\ .
      \label{eq:prop-new-criterion-cert-aqecc-2-calc-1}
  \end{align}
  At this point, we define for each $\alpha$ the positive semidefinite
  operator
  \begin{align}
    \tilde{\rho}_{ER}^\alpha
    &= \sum_{x} \rho_\alpha^{x,x}\otimes\proj{x}_R\ .
  \end{align}
  Note that $B_R \tilde\rho_{ER}^\alpha B_R^\dagger$ is a quantum state, because
  $\tr(B_R \tilde\rho_{ER}^\alpha B_R^\dagger) = \sum_x \tr(B_R \proj{x}
  B_R^\dagger) = \tr(B_R B_R^\dagger) = 1$.  In fact, the quantum states
  $B_R\tilde\rho_{ER}^\alpha B_R^\dagger$ and $B_R \rho_{ER}^\alpha B_R^\dagger$
  are close in trace distance:
  \begin{align}
    \norm[\big]{ B_R \, (\rho^\alpha_{ER} - \tilde\rho^\alpha_{ER} )
    \, B_R^\dagger }_1
    &= \norm*{ B_R \,
    `*( \sum\nolimits_{x\neq x'} \rho_\alpha^{x,x'}\otimes\ketbra{x}{x'} )
    \, B_R^\dagger }_1
    \nonumber\\
    &= \norm*{
      \sum\nolimits_{x\neq x'} \rho_\alpha^{x,x'}\otimes
      (B_R \,\ketbra{x}{x'} \, B_R^\dagger) }_1
    \nonumber\\
    &\leqslant \sum\nolimits_{x\neq x'}
      \norm[\big]{\rho_\alpha^{x,x'}}_1\cdot
      \norm{B_R \,\ketbra{x}{x'} \, B_R^\dagger}_1
    \nonumber\\
    & \leqslant \sum\nolimits_{x\neq x'}
      \norm[\big]{\rho_\alpha^{x,x'}}_1  \leqslant d_L^2 \, \nu\ ,
  \end{align}
  using our
  assumption~\eqref{eq:new-criterion-cert-2-assumption-rho-x-xp-small}, and
  noting that
  $\norm{B_R \,\ketbra{x}{x'} \, B_R^\dagger}_1 \leqslant \norm{B_R\ket{x}}_1 \,
  \norm{\bra{x'}B_R^\dagger}_1 = \norm{\bra{x}B_R^\dagger}_1 \,
  \norm{\bra{x'}B_R^\dagger}_1 \leqslant 1$ because
  $\tr\sqrt{\dmatrixel{x}{B_R^\dagger B_R}} \leqslant 1$.  Recalling the
  relation
  $P(\cdot,\cdot) \leqslant \sqrt{2\delta(\cdot,\cdot)} =
  \sqrt{\norm{(\cdot)-(\cdot)}_1}$ between the purified distance
  $P(\cdot,\cdot) = \sqrt{1 - F^2(\cdot,\cdot)}$ and the trace distance, we have
  \begin{align}
    P( B_R \, \rho^\alpha_{ER} \, B_R^\dagger ,
     B_R \, \tilde\rho^\alpha_{ER} \, B_R^\dagger ) \leqslant  d_L \sqrt{\nu}\ .
  \end{align}
  On the other hand, using again the joint concavity of the fidelity, we have
  \begin{align}
    F`\big( B_R\, \tilde\rho^\alpha_{ER} \, B_R^\dagger ,
    \zeta_\alpha\otimes\sigma_R )
    &= F`*( \sum\nolimits_x \rho^{x,x}_\alpha \otimes (B_R\, \proj{x} \, B_R^\dagger) ,
      \sum\nolimits_x \zeta_\alpha\otimes (B_R\,\proj{x}\, B_R^\dagger) )
      \nonumber\\
    & \geqslant
      \sum_x  \dmatrixel{x}{B^\dagger B}_R \, 
      F`*( \rho^{x,x}_\alpha \otimes
      \frac{B_R\, \proj{x} \, B_R^\dagger}{\dmatrixel{x}{B^\dagger B}_R} ,
      \zeta_\alpha\otimes
      \frac{B_R\, \proj{x} \, B_R^\dagger}{\dmatrixel{x}{B^\dagger B}_R}
      )
      \nonumber\\
    & =
      \sum_x  \dmatrixel{x}{B^\dagger B}_R \, 
      F`*( \rho^{x,x}_\alpha , \zeta_\alpha )
      \nonumber\\
    & \geqslant
      \sum_x  \dmatrixel{x}{B^\dagger B}_R \, 
      \sqrt{1 - \epsilon^{\prime 2}}
      \nonumber\\
    & \geqslant \sqrt{1 - \epsilon^{\prime 2}}\ ,
  \end{align}
  recalling our
  assumption~\eqref{eq:new-criterion-cert-2-assumption-rho-x-x-purif-close} and
  using the fact that $\tr(B^\dagger B) = \tr(BB^\dagger) = 1$; hence
  \begin{align}
    P`\big( B_R\, \tilde\rho^\alpha_{ER} \, B_R^\dagger ,
    \zeta_\alpha\otimes\sigma_R ) \leqslant \epsilon' \ .
  \end{align}
  By triangle inequality for the purified distance, we have
  \begin{align}
    P`\big( B_R \rho^\alpha_{BR} B_R^\dagger, \zeta_\alpha\otimes\sigma_R)
    &\leqslant
      P( B_R \, \rho^\alpha_{ER} \, B_R^\dagger ,
      B_R \, \tilde\rho^\alpha_{ER} \, B_R^\dagger )
      +
      P`\big( B_R\, \tilde\rho^\alpha_{ER} \, B_R^\dagger ,
      \zeta_\alpha\otimes\sigma_R )
      \nonumber\\
    &\leqslant d_L\sqrt{\nu}  + \epsilon' \ .
  \end{align}
  Returning to~\eqref{eq:prop-new-criterion-cert-aqecc-2-calc-1}, we now have
  $F`\big( B_R \rho^\alpha_{BR} B_R^\dagger, \zeta_\alpha\otimes\sigma_R)
  \geqslant \sqrt{1 - (d_L\sqrt{\nu} + \epsilon')^2}$ and hence
  \begin{align}
    F`\big(\widehat{\mathcal{N}\circ\mathcal{E}}(\sigma_{LR}),
    \zeta_{C'E}\otimes\sigma_R)
    \geqslant \sqrt{1 - (d_L\sqrt{\nu} + \epsilon')^2}\ .
  \end{align}
  As this holds for any $\ket{\sigma}_{LR}$, we deduce that
  \begin{align}
    f(\mathcal{N}\circ\mathcal{E}) \geqslant \sqrt{1 - (d_L\sqrt{\nu} + \epsilon')^2}\ ,
  \end{align}
  which implies
  \begin{align}
    \epsilon(\mathcal{N}\circ\mathcal{E})
    \leqslant d_L\sqrt{\nu} + \epsilon'\ .
    \tag*\qedhere
  \end{align}
\end{proof}


%% file: AppendixCalcTruncHaydenCode.tex
We complete the exposition in the main text in
\cref{sec:maintext-truncated-Hayden-code} by calculating the approximation
parameter $\epsilon_{\mathrm{worst}}(\mathcal{N}\circ\mathcal{E}^{(m)})$ of the
constructed code.

The strategy is to apply \cref{prop:new-criterion-cert-aqecc-2}.  First write the operators~\eqref{eq:criterion-cert-aqecc-rho-xxp} in our situation,
\begin{align}
  \rho_i^{x,x'} = \tr_{A\setminus A_i}( V \ketbra{x}{x'} V^\dagger )\ .
\end{align}
We need to show that $\rho_i^{x,x}$ is approximately constant of $x$ and that
$\rho_i^{x,x'}$ is very small for $x\neq x'$.  The latter condition turns out to
be simple: for any $i$ and for any $x\neq x'$, we will see that
$\rho_i^{x,x'}=0$; hence we may take $\nu=0$ in
\cref{prop:new-criterion-cert-aqecc-2}.

For each $i$, we would like to show that there exists a state $\zeta_i$ such
that $\rho_i^{x,x}$ is close to $\zeta_i$ in fidelity distance for each $x$.  We
choose to work with the trace distance instead, and deduce that the states are
close in fidelity using the relation
$F(\cdot,\cdot) \geqslant \sqrt{1 - 2\delta(\cdot,\cdot)}$ between the fidelity
and the trace distance.  We bound the trace distance as follows.  For each $i$,
we find a positive semidefinite operator $\tau_i$ with the property tht
$\rho_i^{x,x}\geqslant \tau_i$ for all $x$.  This implies that
$\rho_i^{x,x} = \tau_i + \Delta_i^x$ for some positive semidefinite operators
$\Delta_i^x$ with $\tr(\Delta_i^x) = 1 - \tr(\tau_i)$.  Define
$\zeta_i = \tau_i + \xi_i$, for any freely chosen $\xi_i\geqslant 0$ with
$\tr(\xi_i) = 1 - \tr(\tau_i)$. Then, we have
$\rho_i^{x,x} - \zeta_i = \Delta_i^{x} - \xi_i$, and
$\delta(\rho_i^{x,x}, \zeta_i) = (1/2) \norm{\rho_i^{x,x} - \zeta_i}_1 \leqslant
(1/2) (\tr(\Delta_i^x) + \tr(\xi_i)) = 1 - \tr(\tau_i)$.  To summarize: If we
find, for each $i$, an operator $\tau_i\geqslant 0$ with
$\rho_i^{x,x}\geqslant \tau_i$ for all $x$, then we can deduce that there are
states $\zeta_i$ such that
\begin{align}
  F(\rho_i^{x,x}, \zeta_i) \geqslant \sqrt{1 - {\epsilon'}^2}\ ,
\end{align}
where $\epsilon' = \min_i \sqrt{2(1 - \tr(\tau_i))}$.

We may calculate the corresponding operators $\rho_i^{x,x'}$, starting with
$i=1$:
\begin{align}
  \rho_{1}^{x,x'}
  &= \tr_{A \setminus A_1}(V\ketbra{x}{x'}V^\dagger)
    \nonumber\\
  &= \frac1{2m+1} \sum_{y,y'=-m}^{m}
    \ketbra{-3y}{-3y'}\,\delta_{y-x,y'-x'}\,\delta_{2(x+y),2(x'+y')}
    \nonumber\\
  &= \frac{\delta_{x,x'}}{2m+1} \sum_{y=-m}^{m} \proj{-3y}_{A_1}\ ,
\end{align}
since the two Kronecker deltas force $x'=x$ and $y'=y$.
Similarly, we have
\begin{align}
  \rho_2^{x,x'}
  &= \frac{\delta_{x,x'}}{2m+1} \sum_{y=-m}^{m} \proj{y-x}_{A_2}\\
  \rho_3^{x,x'}
  &=  \frac{\delta_{x,x'}}{2m+1} \sum_{y=-m}^{m} \proj{2(x+y)}_{A_3}
  \ .
\end{align}

First of all, for each of $i=1,2,3$ we have that $\rho_i^{x,x'} = 0$ if $x\neq x'$.
Then, we have that $\rho_1^{x,x}$ is already independent of $x$, so we may
choose $\tau_1 = \rho_1^{1,1} = \rho_1^{x,x}\ \forall\;x$.
Next, $\rho_2^{x,x}$ is diagonal, with constant diagonal elements $1/(2m+1)$ at
states ${-m-x},{-m-x+1}, \ldots, {m-x}$.  We may thus choose
\begin{align}
  \tau_2 = \frac1{2m+1} \sum_{u=-m+h}^{m-h} \proj{u}\ ,
\end{align}
such that $\rho_2^{x,x} \geqslant \tau_2$ for all $x$ (\cref{fig:trunc-H-code-tau-2}).
\begin{figure}
  \centering
  \includegraphics{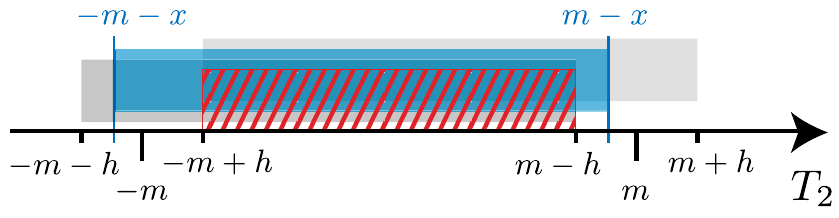}
  \caption{Finding the ``common minimal operator'' for the different
    $\rho_2^{x,x}$'s.  The solid rectangles illustrate the spectra of the
    different $\rho_2^{x,x}$.  The eigenvalues are all equal and the rectangles
    are displaced vertically for readability. The hatched region corresponds to
    a good choice for $\tau_2$.}
  \label{fig:trunc-H-code-tau-2}
\end{figure}
Finally, $\rho_3^{x,x}$ is also diagonal with elements $1/(2m+1)$ at states
$-2m+2x, -2m+2x+2, \ldots, 2m+2x$.  Similarly we may choose
\begin{align}
  \tau_3 = \frac1{2m+1} \sum_{u=-m+h}^{m-h} \proj{2u}\ ,
\end{align}
which guarantees that $\rho_3^{x,x} \geqslant \tau_3$ for each $x$.  We have
\begin{equation}
  \begin{aligned}
    \tr(\tau_1) &= 1\ ;
    \\
    \tr(\tau_2)
    &= \frac{2(m-h)+1}{2m+1} = 1 - \frac{2h}{2m+1}\ ;
    \\
    \tr(\tau_3) &= 1 - \frac{2h}{2m+1}\ ,
  \end{aligned}
\end{equation}
so we may set according to the above $\epsilon' = \sqrt{4h/(2m+1)}$.  According
to \cref{prop:new-criterion-cert-aqecc-2}, the code $V^{(m)}_{L\to A}$ is an
approximate quantum error-correcting code with
\begin{align}
  \epsilon_{\mathrm{worst}}(\mathcal{N}\circ\mathcal{E}^{(m)})
  \leqslant \epsilon'\ .
\end{align}
We have $1/(2m+1)\approx 1/2m$, and to first order in $h/m$, we have
\begin{align}
  \epsilon(\mathcal{N}\circ\mathcal{E}^{(m)})
  \lesssim
  \sqrt2 \sqrt{\frac{h}{m}}\ .
\end{align}
So our codes become good in the limit $h/m\to 0$.  

To compare with our bound~\eqref{eq:simple-bound}, we choose
$q_1=q_2=q_3=1/3$ and note that $\delta=0$, $\eta=0$, and $\Delta T_L = 2h$.
Also, we have
\begin{align}
  \Delta T_{1} &= 2\cdot 3m\ ;
  & \Delta T_{2} &= 2(m+h) \ ;
  & \Delta T_{3} &= 4(m+h) \ ;
\end{align}
so, for $m\gg h$, we have
$\max_i q_i^{-1} \Delta T_i \approx 18 m$.  Our bound then reads
\begin{align}
  \epsilon(\mathcal{N}\circ\mathcal{E})
  \geqslant \frac12 \frac{\Delta T_L}{\max_i q_i^{-1} \Delta T_i}
  \approx \frac12 \frac{2h}{18m}  = \frac1{18}\,\frac{h}{m}\ .
\end{align}


%% file: AppendixCalcGaussianHaydenCode.tex
Again, we make use of
\cref{prop:new-criterion-cert-aqecc-2}.
First, we compute the normalization factor as
\begin{align}
  c_w = \sum_{y=-\infty}^{\infty} \ee^{-\frac{y^2}{2 w^2}}
  = \sum_{y=-\infty}^{\infty} `*(\ee^{-\frac{1}{2 w^2}})^{(y^2)}
  = \vartheta_3`\Big(0,\ee^{-\frac{1}{2w^2}})\ ,
\end{align}
where $\vartheta_3(z,q)$ is Jacobi's theta function.\footnote{See DLMF:
  \url{http://dlmf.nist.gov/20}.  Our notation follows DLMF's notation.}  A
straightforward observation is that $c_w\geqslant 1$ (the term $y=0$ in the sum
is already equal to one).

We need to determine the operators $\rho_{1,2,3}^{x,x'}$.  We have
\begin{align}
  \rho_{1}^{x,x'}
  &= \tr_{A \setminus A_1}(V\ketbra{x}{x'}V^\dagger)
     \nonumber\\
  &=  c_w^{-1} \sum_{y,y'=-\infty}^{\infty}
    \ee^{-\frac{y^2}{4 w^2} - \frac{y^{\prime 2}}{4 w^2}}
    \ketbra{-3y}{-3y'}\,\delta_{y-x,y'-x'}\,\delta_{2(x+y),2(x'+y')}
    \nonumber\\
  &= \frac{\delta_{x,x'}}{c_w} \sum_{y=-\infty}^{\infty}
    \ee^{-\frac{y^2}{2 w^2}} \proj{-3y}_{A_1}\ .
\end{align}
Similarly, for the second and third systems,
\begin{align}
  \rho_2^{x,x'}
  &= \frac{\delta_{x,x'}}{c_w} \sum_{y=-\infty}^{\infty}
    \ee^{-\frac{(y+x)^2}{2 w^2}} \proj{y}_{A_2} \\
  \rho_3^{x,x'}
  &= \frac{\delta_{x,x'}}{c_w} \sum_{y=-\infty}^{\infty}
    \ee^{-\frac{(y-x)^2}{2 w^2}} \proj{2y}_{A_3} \ .
\end{align}

Hence, we have
$\norm{\rho_1^{x,x'}}_1 = \norm{\rho_2^{x,x'}}_1 = \norm{\rho_3^{x,x'}}_1 = 0$
for all $x\neq x'$, so the
conditions~\eqref{eq:new-criterion-cert-2-assumption-rho-x-xp-small} are
satisfied with $\nu=0$.

Now we need to verify the
conditions~\eqref{eq:new-criterion-cert-2-assumption-rho-x-x-purif-close}.  For
the first system, $\rho_1^{x,x}$ doesn't depend on $x$, so choosing
$\zeta_1 = \rho_1^{0,0}$ we have $P(\rho_1^{x,x}, \zeta_{1}) = 0$ for all $x$.
For the second system, we choose $\zeta_2 = \rho_2^{0,0}$ and calculate
\begin{align}
  F`\big(\rho_2^{x,x}, \zeta_2)
  &= \sum_{y=-\infty}^{\infty} \sqrt{\frac1{c_w}
    \ee^{-\frac{(y+x)^2}{2w^2}}} \sqrt{\frac1{c_w}{\ee^{-\frac{y^2}{2w^2}}}}
  = \frac1{c_w}\sum_{y=-\infty}^{\infty}
    \ee^{-\frac{(y+x)^2 + y^2}{4w^2}}
    \geqslant \ee^{-\frac{h^2}{8w^2}}\ ,
    \label{eq:Hayden-code-Gaussian-envelope-calc-1}
\end{align}
where the calculation of the last inequality is carried out below in
\cref{lemma:Hayden-code-Gaussian-envelope-calc-fid-expr}.  Hence
\begin{align}
  P`\big(\rho_2^{x,x}, \zeta_2) \leqslant
  \sqrt{1 - \ee^{-\frac{h^2}{4w^2}}}
  =   \frac{h}{2w}\sqrt{1 + O`\Big(`\Big(\frac{h}{w})^2)}
  = \frac{h}{2w} + O`\Big(`\Big(\frac{h}{w})^3)\ .
\end{align}
Now, we look at the third system.  Defining $\zeta_3 = \rho_3^{0,0}$, we have
\begin{align}
  F`\big(\rho_3^{x,x}, \zeta_3) = \sum_{y=-\infty}^{\infty} \sqrt{\frac1{c_w}
    \ee^{-\frac{(y-x)^2}{2w^2}}} \sqrt{\frac1{c_w}{\ee^{-\frac{y^2}{2w^2}}}}
  = \frac1{c_w}\sum_{y=-\infty}^{\infty}
    \ee^{-\frac{(y-x)^2 + y^2}{4w^2}}
    \geqslant \ee^{-\frac{h^2}{8w^2}}\ ,
\end{align}
invoking again the calculation in
\cref{lemma:Hayden-code-Gaussian-envelope-calc-fid-expr}.  Hence
\begin{align}
  P`\big(\rho_3^{x,x}, \zeta_3) \leqslant
  \sqrt{1 - \ee^{-\frac{h^2}{4w^2}}}
  = \frac{h}{2w} + O`\Big(`\Big(\frac{h}{w})^3)\ .
\end{align}

We are now in position to apply our criterion.
\cref{prop:new-criterion-cert-aqecc-2} tells us that
\begin{align}
  \epsilon(\mathcal{N}\circ\mathcal{E}^{(w)})
  \leqslant \sqrt{1 - \ee^{-\frac{h^2}{4w^2}}}
  = \frac{h}{2w} + O`\Big(`\Big(\frac{h}{w})^3)\ .
\end{align}

Hence, our code's performance scales as $(1/2)(h/w)$.  For instance, it performs
well in the limit $h/w \to 0$, for instance in the limit $w\to \infty$ with a
constant $h$.

Let's now see how our bound applies to our code (we need the more general bound,
because we are dealing with infinite-dimensional systems with an unbounded
charge observable).  We need to cut off tails of the codeword states on the
physical systems to make the range of charge values finite.
Choose cut-offs $W_1,W_2,W_3\geqslant 0$ for each physical system.  We would
like to compute an upper bound to
$\sum \psi_x \psi_{x'}^* \tr(\Pi_i^\perp\, \rho_i^{x,x'}) = \sum \abs{\psi_x}^2
\tr(\Pi_i^\perp\, \rho_i^{x,x})$, where $\Pi_i^\perp$ projects outside of the
cut-off region.  We have
\begin{align}
  \tr(\Pi_1^\perp\, \rho_1^{x,x}) =
  c_w^{-1} \sum_{\abs{3y} > W_1} \ee^{-\frac{y^2}{2w^2}}
  \leqslant
  c_w^{-1} \sum_{\abs{y} > \lfloor W_1/3\rfloor } \ee^{-\frac{y^2}{2w^2}}
  \leqslant \frac{2}{c_w}\,\frac{w^2}{\lfloor W_1/3\rfloor}\,
  \ee^{-\frac{(\lfloor W_1/3\rfloor)^2}{2w^2}}\ ,
\end{align}
where the bound is calculated in
\cref{lemma:Hayden-code-Gaussian-envelope-calc-tail-weight} below.  Then,
\begin{equation}
  \tr(\Pi_2^\perp\, \rho_2^{x,x}) =
  c_w^{-1} \sum_{\abs{y} > W_2} \ee^{-\frac{(y+x)^2}{2w^2}}
  \leqslant \frac{2}{c_w}\,\frac{w^2}{W_2-\abs{x}}\,
  \ee^{-\frac{(W_2-\abs{x})^2}{2w^2}}  
  \leqslant \frac{2}{c_w}\,\frac{w^2}{W_2-h}\,
  \ee^{-\frac{(W_2-h)^2}{2w^2}}\ .
\end{equation}
Similarly,
\begin{align}
  \tr(\Pi_3^\perp\, \rho_3^{x,x})
  &=
  c_w^{-1} \sum_{\abs{2y} > W_3} \ee^{-\frac{(y-x)^2}{2w^2}} \leqslant \frac{2}{c_w}\,\frac{w^2}{\lfloor W_3/2\rfloor - h}\,
    \ee^{-\frac{(\lfloor W_3/2\rfloor - h)^2}{2w^2}}\ .
\end{align}

Hence, choosing $W_1 = W_2 = W_3 =: W$ with $W\geqslant 2h$ and choosing for
simplicity $W$ as a multiple of $6$, we have
$\lfloor W_1/3\rfloor = W/3 \geqslant (1/3)(W - 2h)$, as well as
$W_2 - h\geqslant W - 2h$ and also $\lfloor W_3/2\rfloor - h = (1/2)(W - 2h)$;
furthermore $W - 2h \geqslant (1/2)(W-2h) \geqslant (1/3)(W - 2h)$.  Then,
\begin{align}
  \abs*{ \sum\nolimits_i \tr(\Pi_i^\perp \rho_i^{x,x}) }
  &\leqslant
    \frac{2}{c_w}\,\frac{w^2}{(1/3)(W - 2h)}\,
    \ee^{-\frac{(\frac13(W - 2h))^2}{2w^2}}
    +
    \frac{2}{c_w}\,\frac{w^2}{W - 2h}\,
    \ee^{-\frac{(W - 2h)^2}{2w^2}}
    \nonumber\\
  &\quad\ 
    +
    \frac{2}{c_w}\,\frac{w^2}{(1/2)(W - 2h)}\,
    \ee^{-\frac{(\frac12(W - 2h))^2}{2w^2}}
    \nonumber\\
  &\leqslant \frac{12}{c_w}\,\frac{w^2}{W-2h}\,\ee^{-\frac{(W-2h)^2}{18w^2}}
    \leqslant \frac{12\,w^2}{W-2h}\,\ee^{-\frac{(W-2h)^2}{18w^2}}
    =: \eta\ ,
\end{align}
recalling that $c_w\geqslant 1$.  Also, $\Delta T_i = 2W_i = 2W$ by
construction.  Furthermore $\Delta T_L = 2h$ and $\delta=0$.  So, our bound
reads (assuming that the noise erasure probabilities are $q_1=q_2=q_3=1/3$)
\begin{align}
  \epsilon(\mathcal{N}\circ\mathcal{E})
  &\geqslant \frac12\,\frac{1}{(\max_i q_i^{-1}) \cdot 2W}\, `*[2h - 2\eta]
    \nonumber 
   = \frac12\,\frac{1}{6W}\, `*[2h -
    \frac{24\,w^2}{W-2h}\,\ee^{-\frac{(W-2h)^2}{18w^2}}]
    \nonumber\\
  &\approx \frac12\,\frac{1}{6W}\, `*[2h -
    \frac{24\,w^2}{W}\,\ee^{-\frac{W^2}{18w^2}}]
   = \frac{h}{6W} -
    \frac{4\,w^2}{W^2}\,\ee^{-\frac{W^2}{18w^2}}\ .
\end{align}
considering the regime $W \gg h$, i.e., $W-2h\approx W$.  Now, if we choose the
cutoff $W = \beta w$ to be proportional to $w$, then we can write our bound as a
function of $h/w$:
\begin{align}
  \epsilon(\mathcal{N}\circ\mathcal{E})
  &\gtrsim
    \frac{1}{6\beta} \frac{h}{w} -
    \frac{4\,\ee^{-\beta^2/8}}{\beta^2}\ .
    \label{eq:bound-for-smooth-3rotor-code-with-beta}
\end{align}
The second term is exponentially suppressed in $\beta$; so choosing $\beta$ only
very moderately large, we get a bound which is effectively proportional to $h/w$
with a proportionality constant $1/(6\beta)$.

Now we find a suitable $\beta$ to plug
into~\eqref{eq:bound-for-smooth-3rotor-code-with-beta} to get a bound in terms
of $h/w$ only.  If we attempt to minimize the
bound~\eqref{eq:bound-for-smooth-3rotor-code-with-beta}, we get as minimization
condition
\begin{align}
  0 = \frac\partial{\partial\beta}`\big(\text{bound})
  = -\frac1{6\beta^2}\,\frac{h}{w} + \frac{8\,\ee^{-\beta^2/8}}{\beta^3}
  + \frac{\ee^{-\beta^2/8}}{\beta}
  =
  \frac1{\beta^2}`*[ -\frac{h}{6w} + \frac{8 + \beta^2}{\beta}\,\ee^{-\beta^2/8} ]\ .
\end{align}
Writing $z=\beta^2/4$ (i.e., $\beta = 2\sqrt{z}$) we obtain
$h/(6w) = \ee^{-z/2}\,(4+2z)/(\sqrt{z})$; the square of this equation gives
\begin{align}
  \frac{h^2}{36 w^2} = `*[ 4z + 16 + \frac{16}{z}] \ee^{-z}\ .
\end{align}
To render this equation tractable, and since we only have to come up with an
approximate educated guess for $\beta$, we may simplify this equation by keeping
the leading term, expecting that $z$ should be moderately large, yielding
\begin{align}
  `*(\frac{h}{12 w})^2 \approx z \ee^{-z}\ .
\end{align}
The solution to the equation $x^2 = z \ee^{-z}$ is given by the Lambert W
function\footnote{\url{https://dlmf.nist.gov/4.13}} with $z = -W`(-x^2)$.  Using
the expansion of the negative branch $W_{\mathrm{m}}$ of the function near
$z\to-\infty$, we have\footnote{\url{https://dlmf.nist.gov/4.13.E11}}
$-W_{\mathrm{m}}(-x^2) \approx \ln(1/x^2)$, and hence we may select
$z \approx \ln`\big((12w/h)^2) = 2\ln(12w/h)$.  This in turn yields the educated
guess $\beta = 2\sqrt{2\ln(12w/h)}$ to plug
into~\eqref{eq:bound-for-smooth-3rotor-code-with-beta}, and the bound becomes
\begin{align}
  \epsilon(\mathcal{N}\circ\mathcal{E})
  &\gtrsim
    \frac{h}{12w}`*[ \frac1{\sqrt{2\ln(12w/h)}} - \frac1{2\ln(12w/h)} ]
    \approx \frac{h/w}{12\sqrt{2\ln(w/h)}}\ ,
\end{align}
using $\sqrt{\ln(12w/h)} = \sqrt{\ln(w/h) + \ln(12)} \approx \sqrt{\ln(w/h)}$.

\begin{lemma}
  \label{lemma:Hayden-code-Gaussian-envelope-calc-fid-expr}
  We have for integer $x,h$, with $\abs{x}\leqslant h$ and with $h$ even,
  \begin{align}
    \frac1{c_w}\sum_{y=-\infty}^{\infty}
    \ee^{-\frac{(y \pm x)^2 + y^2}{4w^2}} \geqslant \ee^{-\frac{h^2}{8w^2}}\ .
  \end{align}
\end{lemma}
\begin{proof}[*lemma:Hayden-code-Gaussian-envelope-calc-fid-expr]
  First, we may assume without loss of generality that we have the ``$+$'' case
  in the exponent (or else simply send $x\to -x$).
  Completing the square, we have
  $(y+x)^2 + y^2 = 2y^2 + 2xy + x^2 = 2(y+x/2)^2 + x^2/2$, and hence
  \begin{align}
    \frac1{c_w}\sum_{y=-\infty}^{\infty}
    \ee^{-\frac{(y+x)^2 + y^2}{4w^2}}
    &= \frac1{c_w}\sum_{y=-\infty}^{\infty}
      \ee^{-\frac{(y+x/2)^2}{2w^2} - \frac{x^2}{8w^2}} = \frac{\ee^{-\frac{x^2}{8w^2}}}{c_w}\sum_{y=-\infty}^{\infty}
      \ee^{-\frac{(y+x/2)^2}{2w^2}}\ .
      \label{eq:Hayden-code-Gaussian-envelope-calc-2}
  \end{align}
  At this point we need to distinguish the case where $x$ is even from the case
  where $x$ is odd.  Assuming first that $x$ is even, we may redefine $y\to y+x/2$
  in the summation and we have
  \begin{align}
    \text{\eqref{eq:Hayden-code-Gaussian-envelope-calc-2} [$x$ even]}
    &= \frac{\ee^{-\frac{x^2}{8w^2}}}{c_w}\sum_{y=-\infty}^{\infty}
      \ee^{-\frac{y^2}{2w^2}} = \frac{\ee^{-\frac{x^2}{8w^2}}}{c_w} \cdot c_w
      =  \ee^{-\frac{x^2}{8w^2}} \geqslant \ee^{-\frac{h^2}{8w^2}} \ ,
  \end{align}
  recalling that $\abs{x}\leqslant h$.  In the case that $x$ is odd, we need to
  work a little bit more; we may redefine $y\to y+(x-1)/2$, and we have
  \begin{align}
    \text{\eqref{eq:Hayden-code-Gaussian-envelope-calc-2} [$x$ odd]}
    &= \frac{\ee^{-\frac{x^2}{8w^2}}}{c_w}\sum_{y=-\infty}^{\infty}
      \ee^{-\frac{(y+\frac12)^2}{2w^2}}
      = \frac{\ee^{-\frac{x^2}{8w^2}}}{c_w} \,
      \vartheta_2`\Big(0,\ee^{-\frac1{2w^2}})\ ,
      \label{eq:Hayden-code-Gaussian-envelope-calc-3}
  \end{align}
  using another theta function corresponding to this type of summation.
  \cref{lemma:theta-3-incr-pure-imag} shows that
  $\vartheta_2`\big(0,\ee^{-1/(2w^2)}) \geqslant \ee^{-1/(8w^2)}\,
  \vartheta_3`\big(0,\ee^{-1/(2w^2)})$, and so we have
  \begin{align}
    \text{\eqref{eq:Hayden-code-Gaussian-envelope-calc-3}}
    &\geqslant \ee^{-\frac{x^2+1}{8w^2}}
      \geqslant \ee^{-\frac{(\abs{x}+1)^2}{8w^2}}
      \geqslant \ee^{-\frac{h^2}{8w^2}}\ ,
  \end{align}
  where we have assumed that $h$ is even, and so $\abs{x}+1\leqslant h$.
\end{proof}

\begin{lemma}
  \label{lemma:Hayden-code-Gaussian-envelope-calc-tail-weight}
  We have, for $W\geqslant 0$,
  \begin{align}
    \sum_{\abs{y} > W} \ee^{-\frac{(y\pm x)^2}{2w^2}}
    \leqslant
    2\,\frac{w^2}{W-\abs{x}}\,\ee^{-\frac{(W-\abs{x})^2}{2w^2}}\ .
  \end{align}
\end{lemma}
\begin{proof}[*lemma:Hayden-code-Gaussian-envelope-calc-tail-weight]
  Assume $x\geqslant 0$, or else redefine $x\to -x$.  We have
  \begin{align}
    \sum_{\abs{y} > W} \ee^{-\frac{(y \pm x)^2}{2w^2}}
    \leqslant 2\cdot \sum_{y > W - x} \ee^{-\frac{y^2}{2w^2}}
    \leqslant 2\cdot \sum_{y \geqslant W - x + 1} \ee^{-\frac{y^2}{2w^2}}
    \leqslant 2\, \int_{W-x}^{\infty} dy \, \ee^{-\frac{y^2}{2w^2}}\ ,
    \label{eq:lemma-Hayden-code-Gaussian-envelope-calc-tail-weight-calc-1}
  \end{align}
  where the integral is necessarily an overestimation of the sum, as the sum can
  be seen as an integral of a step function, where each step is specified at the
  right edge by the value of the integrand function; this step function lies
  beneath the actual decreasing function $\ee^{-y^2/(2w^2)}$.  Setting
  $t=y/(w\sqrt2)$,
  \begin{align}
    \text{\eqref{eq:lemma-Hayden-code-Gaussian-envelope-calc-tail-weight-calc-1}}
    &= 2 \int_{\frac{W-x}{w\sqrt2}}^{\infty} dt\,w\sqrt{2}\,\ee^{-t^2}
                    = w\,\sqrt{2\pi} \,\frac{2}{\sqrt\pi}
      \int_{\frac{W-x}{w\sqrt2}}^{\infty} dt\,\ee^{-t^2}
                      = w\,\sqrt{2\pi} \, \erfc`*(\frac{W-x}{w\sqrt{2}})\ .
      \label{eq:lemma-Hayden-code-Gaussian-envelope-calc-tail-weight-calc-2}
  \end{align}
  We use the known bound\footnote{See for instance
    \url{http://dlmf.nist.gov/7.8.E4} or
    \url{http://mathworld.wolfram.com/Erfc.html}}
  \begin{align}
    \erfc`*(z) \leqslant \frac{\ee^{-z^2}}{z\sqrt{\pi}}\ ,
  \end{align}
  leading to
  \begin{align}
    \text{\eqref{eq:lemma-Hayden-code-Gaussian-envelope-calc-tail-weight-calc-2}}
    &\leqslant 2\,\frac{w^2}{W-x}\,\ee^{-\frac{(W-x)^2}{2w^2}}\ .
      \tag*\qedhere
  \end{align}
\end{proof}

Finally, we prove a property of the theta functions that we used above.

\begin{lemma}
  \label{lemma:theta-3-incr-pure-imag}
  Let $0< q \leqslant 1$, and let $z\in\mathbb{C}$ with $\Re(z)=0$ and
  $\Im(z)\geqslant 0$.  Then
  \begin{align}
    \vartheta_3(z,q) \geqslant \vartheta_3(0,q)\ .
    \label{eq:lemma-theta-3-incr-pure-imag-th33}
  \end{align}

  Furthermore, we have
  \begin{align}
    \vartheta_2(0,q) \geqslant q^{1/4}\,\vartheta_3(0,q)\ .
    \label{eq:lemma-theta-3-incr-pure-imag-th23}
  \end{align}
\end{lemma}
\begin{proof}[*lemma:theta-3-incr-pure-imag]
  We start by proving~\eqref{eq:lemma-theta-3-incr-pure-imag-th33}.  Writing
  $q = \ee^{i\pi\tau}$ with $\Re(\tau)=0$ and $\Im(\tau)\geqslant 0$, we
  have\footnote{See \url{http://dlmf.nist.gov/20.5.E7}, Eq.~(20.5.7)}
  \begin{align}
    \vartheta_3(z,q)
    = \vartheta_3(0,q) \cdot
    \prod_{n=1}^{\infty}
    \frac{\cos`*((n-\frac{1}{2})\pi\tau+z) \cos`*((n-\tfrac{1}{2})\pi\tau-z)}%
    {\cos^2`*((n-\frac12)\pi\tau)}
    =: \vartheta_3(0,q) \cdot \prod_{n=1}^{\infty} a_n\ .
  \end{align}
  We will show that the product is greater than $1$, by showing that
  $a_n\geqslant 1$ for each $n$.  We have
  \begin{align}
    a_n = \frac{\cos`*(i`*(a+b))\,\cos`*(i`*(a-b))}{\cos^2`*(i\,a)}\ ,
  \end{align}
  defining $a,b\geqslant 0$ as $a = (n-1/2)\,\pi\,\Im(\tau)$ and $b = \Im(z)$.
  Since $\cos`*(i\,\varphi) = \cosh`*(\varphi)$, we have
  \begin{align}
    a_n = \frac{\cosh`*(a+b)\,\cosh`*(a-b)}{\cosh^2`*(a)}\ .
  \end{align}
  (By the way, this is another way of seeing that $\vartheta_3(z,q)$ must be
  real and positive, since all the $a_n$ are real positive and
  $\vartheta_3(0,q)$ is real positive as given by its series representation.
  Recall that $z$ is pure imaginary with $\Im(z)\geqslant 0$, and that
  $0<q\leqslant 1$.)  With the usual properties of the hyperbolic functions, we
  have
  \begin{align}
    \cosh`*(a+b) \cosh`*(a-b)
    &= `*[\cosh(a)\cosh(b) + \sinh(a)\sinh(b)] `*[\cosh(a)\cosh(b) - \sinh(a)\sinh(b)]
      \nonumber\\
    &= \cosh^2(a)\cosh^2(b) - \sinh^2(a)\sinh^2(b)
      \nonumber\\
    &= \cosh^2(a)\cosh^2(b) `*[ 1 - \tanh^2(a)\tanh^2(b)]
      \nonumber\\
    &\geqslant \cosh^2(a)\cosh^2(b) `*[ 1 - \tanh^2(b)]
     = \cosh^2(a)\ ,
  \end{align}
  using $\tanh(a)\leqslant 1$ and $1/\cosh^2(b)=1-\tanh^2(b)$.  Hence finally,
  $a_n\geqslant 1$.  This proves~\eqref{eq:lemma-theta-3-incr-pure-imag-th33}.

  To prove~\eqref{eq:lemma-theta-3-incr-pure-imag-th23}, we invoke the following
  property of the theta functions,\footnote{See
    \url{http://dlmf.nist.gov/20.2.E12}, Eq.~(20.2.12)} valid for any
  $q=\ee^{i\pi\tau}$,
  \begin{align}
    \vartheta_2`\big(0,q) = q^{1/4} \, \vartheta_3`*(\frac12 \pi \tau, q)\ .
  \end{align}
  For $0<q \leqslant 1$, necessarily $\tau$ is pure imaginary with
  $\Im(\tau)\geqslant 0$; we may thus
  invoke~\eqref{eq:lemma-theta-3-incr-pure-imag-th33}, which
  proves~\eqref{eq:lemma-theta-3-incr-pure-imag-th23}.
\end{proof}


%% file: AppendixCalPerfectCode.tex
The normalized encoding for this code is
\begin{equation}
\ket{x} \rightarrow\frac{1}{\sqrt{c_{w,x}}}\sum_{j,k,l,m,n \in \mathbb{Z}} e^{-\frac{1}{4w^{2}}\left(j^{2}+k^{2}+l^{2}+m^{2}+n^{2}\right)}T_{jklmnx}^{(\infty)}\ket{j,k,l,m,n}~,
\end{equation}
where $c_{w,x}$ is the normalization and $T^{(\infty)}$ is defined in \cref{eq:tensor2}.

\paragraph{Single erasure.}

We first calculate $\rho_{\ell}^{x,x}$ for $\ell\in\{1,2,3,4,5\}$
and then outline why $||\rho_{\ell}^{x,x^{\prime}\neq x}||_{1}=O(e^{-cw^{2}})$.
By the cyclic permutation symmetry of the code, we only have to calculate
$\rho_{1}^{x,x}$. Performing the partial trace and simplifying all
Kronecker delta functions leaves us with the diagonal reduced density
matrix
\begin{equation}
\rho_{1}^{x,x}=\frac{1}{c_{w,x}}\sum_{j\in\mathbb{Z}}\left(\sum_{k,l,m\in\mathbb{Z}}e^{-\frac{1}{2w^{2}}\left(j^{2}+k^{2}+l^{2}+m^{2}+\left[j+k+l+m-x\right]^{2}\right)}\right)\proj j \label{eq:interimsumsinglemode}
\end{equation}
Now we apply the Poisson summation formula,
\begin{equation}
\sum_{n\in\mathbb{Z}}f\left(n\right)=\sum_{n\in\mathbb{Z}}\int_{-\infty}^{\infty}dxe^{2\pi inx}f\left(x\right)\,,\label{eq:poisson}
\end{equation}
to each of the three sums above. Typically, the $n=0$ term on the right-hand-side is dominant (i.e., the leading order contribution in the large-$w$ limit), and taking only this term is equivalent to approximating the sum with a Gaussian integral. Each of the remaining
terms suppressed as $O(e^{-cw^{2}})$, where $c$ is a positive constant
increasing with $n$. Because $c$ increases with $n$, the $n+1$-th
term is subleading with respect to the $n$th term. Thus, the entire
sum of exponentially suppressed terms can itself be bounded by an
exponential (e.g., $e^{-2x}+e^{-3x}<e^{-x}$ for $x>1$). We omit
these corrections and focus on the dominant term $k=l=m=0$ after having applied Poisson summation to \cref{eq:interimsumsinglemode}:
\begin{equation}
\rho_{1}^{x,x}\sim\frac{\sqrt{2}\pi^{3/2}w^{3}}{c_{w,x}}\sum_{j\in\mathbb{Z}}e^{-\frac{5j^{2}-2jx+x^{2}}{8w^{2}}}\proj j\,.
\end{equation}
Forcing $\text{Tr}\{\rho_{1}^{x,x}\}=1$ and once again approximating
the resulting sum with an integral solves for the normalization $c_{w,x}$
in the large $w$ limit. Plugging that back into the above equation
and simplifying produces
\begin{equation}
\rho_{1}^{x,x}\sim\sqrt{\frac{5}{2\pi}}\frac{1}{2w}\sum_{j\in\mathbb{Z}}e^{-\frac{(x-5j)^{2}}{40w^{2}}}\proj{j}\,.\label{eq:rhoxx1}
\end{equation}
Now we calculate the fidelity of the above state to $\rho_{1}^{0,0}$.
Using the fact that the states commute with each other, taking the
square root of each entry in the resulting diagonal matrix, and applying
Poisson summation yields
\begin{align}
F^{2}\left(\rho_{1}^{x,x},\rho_{1}^{0,0}\right) =\sqrt{\frac{5}{2\pi}}\frac{1}{2w}\sum_{j\in\mathbb{Z}}e^{-\frac{50 j^2-10 j x+x^2}{80 w^2}}\sim e^{-\frac{x^{2}}{160w^{2}}}\geq e^{-\frac{h^{2}}{160w^{2}}}\,.
\end{align}
Plugging this into the infidelity yields the result~\eqref{eq:fiverotorepsone}.

Returning to the $x\neq x^{\prime}$ case, we show why those cases
do not significantly contribute. The reduced density matrix is of
the form
\[
\rho_{1}^{x,x^{\prime}}=\frac{1}{\sqrt{c_{w,x}c_{w,x^{\prime}}}}\sum_{j\in\mathbb{Z}}\left(\sum_{k,l,m\in\mathbb{Z}}\gamma_{j,k,l,m}^{x,x^{\prime}}\right)\ket j\bra{j+x-x^{\prime}}
\]
where $\gamma_{j,k,l,m}^{x,x^{\prime}}$ is a product of a Gaussian in the
variables $k,l,m$ (just like the $x=x^{\prime}$ case above) and a phase
$\propto2\pi\Phi$ (which goes away when $x=x^{\prime}$).  We first apply Poisson
summation to the internal three sums and evaluate the normalizations in the
large-$w$ limit.  In this case, the centers of the Gaussians in the $k,l,m$-sum
depend on $\Phi$ and the dominant term on the right-hand-side of
\cref{eq:poisson} may not longer be the center-of-mass term $n=0$. We will
however set $\Phi$ to be an irrational number close to zero from now on, i.e.,
taking $\Phi \ll 1$ while making sure that $\Phi w\rightarrow \infty$. This
makes sure that the center-of-mass mode is dominant.  Writing the norm and
applying Poisson summation to the remaining sum reveals
\begin{align}
\norm{\rho_{1}^{x,x^{\prime}}}_1 &\sim \sqrt{\frac{5}{2\pi}}\frac{1}{2w}e^{-2\pi^{2}\Phi^{2}w^{2}\left(x-\text{\ensuremath{x^{\prime}}}\right)^{2}}\sum_{j\in\mathbb{Z}}e^{-\frac{25j^{2}-30jx+20j x^{\prime}+13x^{2}-20x x^{\prime}+8x^{\prime 2}}{40w^{2}}} \nonumber\\
&=O\left(e^{-2\pi^{2}\left(x-\text{\ensuremath{x^{\prime}}}\right)^{2}\Phi^{2}w^{2}}\right)\,.
\end{align}
We see that the one-norm is exponentially suppressed in $w^2$ for the off-diagonal (i.e., $x\neq x^\prime$) reduced matrices.

\paragraph{Two erasures.}

We first calculate $\rho_{\ell,\ell^{\prime}}^{x,x}$ for $\ell,\ell^{\prime}\in\{1,2,3,4,5\}$
and then argue that $||\rho_{\ell,\ell^{\prime}}^{x,x^{\prime}\neq x}||_{1}=O(e^{-cw^{2}})$.
Due to the cyclic permutation symmetry, we only need to calculate
$\rho_{1,2}^{x,x}$ and $\rho_{1,3}^{x,x}$. Performing the partial
trace, simplifying the Kronecker delta functions, plugging in the
normalization, and applying Poisson summation yields
\begin{equation}
\rho_{1,2}^{x,x}\sim\sqrt{\frac{5}{3}}\frac{1}{2\pi w^{2}}\sum_{j,k\in\mathbb{Z}}e^{-\frac{10j^{2}+5jk+10k^{2}-5\left(j+k\right)x+x^{2}}{15w^{2}}}\proj{j,k}\sim\rho_{1,3}^{x,x}\,.
\end{equation}
In other words, both $\rho_{1,2}^{x,x}$ and $\rho_{1,3}^{x,x}$ are identical in
the large $w$ limit. Note that, unlike $\rho_1^{x,x}$, these matrices have
off-diagonal elements that are exponentially suppressed in $w^2$. These elements
have been ignored above, but we mention them in the $x\neq x^\prime$ case
below. Taking the fidelity between $\rho_{\ell,\ell^{\prime}}^{x,x}$ and
$\rho_{\ell,\ell^{\prime}}^{0,0}$ as before yields the result
$F^{2}(\rho_{1,2}^{x,x},\rho_{1,2}^{0,0})\geq
e^{-\frac{h^{2}}{60w^{2}}}$ claimed in~\eqref{eq:fiverotorepstwo}.

The $x\neq x^{\prime}$ case is more difficult this time because the
unapproximated reduced density matrix no longer has just one nonzero diagonal. Without any approximations, it is
\begin{equation}
\rho_{1,2}^{x,x^{\prime}}=\sqrt{\frac{5}{3}}\frac{1}{2\pi w^{2}}\sum_{j,j^{\prime},k\in\mathbb{Z}}\left(\sum_{l,m\in\mathbb{Z}}\gamma_{j,j^{\prime},k,l,m}^{x,x^{\prime}}\right)\ket j\bra{j^{\prime}}\otimes\ket k\bra{k+x^{\prime}-x+j-j^{\prime}}\ .
\end{equation} Applying Poisson summation to the internal two sums for $x\neq x^{\prime}$ reveals that all matrix elements are exponentially suppressed with $w^2$,
\begin{equation}
\sum_{l,m\in\mathbb{Z}}\gamma_{j,j^{\prime},k,l,m}^{x,x^{\prime}}=\sqrt{\frac{5}{3}}\frac{1}{2\pi w^{2}}O\left(e^{-\frac{4}{3}\pi^{2}\left[3\left(j^{\prime}-j\right)^{2}-3\left(j^{\prime}-j\right)\left(x^{\prime}-x\right)+\left(x^{\prime}-x\right)^{2}\right]\Phi^{2}w^{2}}\right)\,.
\end{equation} However, there are particular values of $(j,j^\prime)$ for which the function in the exponent above is minimized; we select those and show that the trace norm is exponentially suppressed in $w^2$. For even $x-x^\prime$, the band at $j^{\prime}=j-\frac{x-x^{\prime}}{2}$ decays the slowest. Ignoring all other bands and calculating the trace norm yields
\begin{equation}
\norm{\rho_{1,2}^{x,x^{\prime}}}_1 = O\left(e^{-\frac{1}{3}\pi^{2}\left(x-x^{\prime}\right)^{2}\Phi^{2}w^{2}}\right)~.
\end{equation}
For odd $x-x^{\prime}$, there are two bands $j^{\prime}=j-\frac{x-x^{\prime}\pm 1}{2}$ whose entries decay the slowest. Calculating the square root of $\rho_{1,2}^{x,x^{\prime}}\rho_{1,2}^{x,x^{\prime}\dagger}$ is more difficult since the resulting matrix is tri-diagonal. However, ignoring the off-diagonal entries, taking the square root, and bounding the resulting integral still yields exponential scaling with $w^2$.


%% file: AppendixCalcThermoCode.tex
Here, we carry out the calculations that are relevant for
\cref{sec:thermodynamic-codes} of the main text.

The operators $\rho_d^{m,m}$ (reduced states on $d$ consecutive sites) are
provided as:
\begin{align}
  \rho_d^{m,m} =
  \sum_{r=-d}^{d}
  K_{r,d,m}^{N} \, \proj{h_r^d}_d\ ,
\end{align}
with
\begin{align}
  K_{r,d,m}^{N} = \frac{
  \begin{pmatrix} d \\ d/2+r/2 \end{pmatrix}
  \begin{pmatrix} N-d \\ (N-d)/2+(m-r)/2 \end{pmatrix}
  }{
  \begin{pmatrix} N \\ N/2+m/2 \end{pmatrix}
  }
  \ .
\end{align}

The fidelity between two states which commute reduces to the Bhattacharyya
coefficient (the classical version of the fidelity):
\begin{align}
  F(\rho_d^{m,m}, \rho_d^{0,0})
  = \sum_{r=-d}^d \sqrt{ K_{r,d,m}^{N} } \sqrt{ K_{r,d,0}^{N} } \ .
\end{align}
The complicated calculation is deferred to
\cref{lemma:thermo-code-Heisenberg-fidelity-calc} below, which gives us:
\begin{align}
  F(\rho_d^{m,m}, \rho_d^{0,0})
  &\geqslant
    1 - O(N^{-2})\ .
\end{align}
This, in turn, tells us that
\begin{align}
  P(\rho_d^{m,m}, \rho_d^{0,0}) \leqslant O(N^{-1})\ .
\end{align}

The ``logical off-diagonal'' terms $\rho_d^{m,m'}$ for $m\neq m'$ are exactly
zero, because we made sure to space out the codewords in magnetization by
$2d+1$, following the construction of ref.~\cite{Brandao2017arXiv_chainAQECC}.

Hence, applying \cref{prop:new-criterion-cert-aqecc-2}, we see that our code
is an AQECC against the erasure of $d$ consecutive sites, with
\begin{align}
  \epsilon(\mathcal{N}\circ\mathcal{E}) \leqslant O(N^{-1})\ .
\end{align}
This matches exactly the scaling of our bound \eqref{eq:simple-bound}.

\begin{lemma}
  \label{lemma:thermo-code-Heisenberg-fidelity-calc}
  There exists a constant $D_{d,m}$ of $N$ such that (for constant $d,m$):
  \begin{align}
    F(\rho_d^{m,m}, \rho_d^{0,0})
    &\geqslant
      1 - \frac{D_{d,m}}{N^2} + O(N^{-3})\ .
  \end{align}
\end{lemma}

\begin{proof}[*lemma:thermo-code-Heisenberg-fidelity-calc]
  We use Stirling's formula up to order $1/N^2$:
  \begin{align}
    \ln(N!) = N\ln(N) - N + \frac12\ln(2\pi N) + \frac1{12N} + O(N^{-3})\ ,
    \label{eq:Stirling}
  \end{align}
  (noting that there is in fact no term of order $1/N^2$).  Now, for any $x$,
  ignoring terms of order $O(N^{-3})$, we have:
  \begin{dmath*}
    \ln\begin{pmatrix} N \\ N/2 + x/2 \end{pmatrix}
    = { N\ln(N) + \frac{\ln(2\pi)}{2} + \frac{\ln(N)}{2} + \frac1{12 N} + O(N^{-3}) }
      - \left[
      `*(\frac{N}{2}+\frac{x}{2})\ln`*(\frac{N}{2}+\frac{x}{2})
      + \frac{\ln`(2\pi)}{2} + \frac12\ln`*(\frac{N}{2}+\frac{x}{2})
      + \frac1{6(N+x)}
    \right]
    - \left[
      `*(\frac{N}{2}-\frac{x}{2})\ln`*(\frac{N}{2}-\frac{x}{2})
      + \frac{\ln`(2\pi)}{2} + \frac12\ln`*(\frac{N}{2}-\frac{x}{2})
      + \frac1{6(N-x)}
    \right]\ .
  \end{dmath*}
  Using the expansions
  \begin{align}
    \ln`*(\frac{N}{2}\pm\frac{x}{2})
    &= \ln(N) - \ln(2) + \ln`*(1 \pm \frac{x}{N})  \nonumber\\
    &= \ln(N)-\ln(2) \pm \frac{x}{N} - \frac{x^2}{2N^2} \pm \frac{x^3}{3N^3}
      + O`*(N^{-4})\ ; \\
    \frac1{N\pm x}
    &= \frac1N `*(\frac1{1\pm x/N}) = \frac1N \mp \frac{x}{N^2} + O(N^{-3})\ ,
  \end{align}
  one continues, still keeping all the terms up to order $1/N^2$:
  \begin{dmath*}
      \ln\begin{pmatrix} N \\ N/2 + x/2 \end{pmatrix}
      = N\ln(2) - \frac{\ln(N)}{2} + \ln`*(\frac{2}{\sqrt{2\pi}}) - \frac1{4N}
      -\frac{x^2}{2N} + \frac{x^2}{2N^2}  + O(N^{-3})\ .
  \end{dmath*}
  Now we may apply this to calculate $\ln`*(\sqrt{K_{r,d,m}^N K_{r,d,0}^N})$,
  using the fact that
  $\ln(N-d) = \ln(N) + \ln(1-d/N) = \ln(N) - d/N - d^2/(2N^2) + O(N^{-3})$:
  \begin{dmath*}
    \ln\sqrt{K_{r,d,m}^N K_{r,d,0}^N}
    = \ln\begin{pmatrix}d\\d/2+r/2\end{pmatrix}
    + \frac12\ln\begin{pmatrix} N-d \\ (N-d)/2 + (m-r)/2 \end{pmatrix}
    + \frac12\ln\begin{pmatrix} N-d \\ (N-d)/2 - r/2 \end{pmatrix}
    - \frac12\ln\begin{pmatrix} N \\ N/2 + m/2 \end{pmatrix}
    - \frac12\ln\begin{pmatrix} N \\ N/2 \end{pmatrix}
    = \ln\begin{pmatrix}d\\d/2+r/2\end{pmatrix}
    - d\ln(2)  %
    + \frac12`*{ - \frac12`*[ \ln(N) - \frac{d}{N} - \frac{d^2}{2 N^2} ]
      - \frac1{4}`*(\frac1{N} + \frac{d}{N^2})
      - \frac{(m-r)^2}{2}`*(\frac1{N} + \frac{d}{N^2}) + \frac{(m-r)^2}{2 N^2}}
    + \frac12`*{ - \frac12`*[ \ln(N) - \frac{d}{N} - \frac{d^2}{2 N^2} ]
      - \frac1{4}`*(\frac1{N} + \frac{d}{N^2})
      - \frac{r^2}{2}`*(\frac1{N}+\frac{d}{N^2}) + \frac{r^2}{2 N^2} }
    - \frac12`*{ - \frac12\ln(N) - \frac1{4N} - \frac{m^2}{2N} + \frac{m^2}{2 N^2} }
    - \frac12`*{ - \frac12\ln(N) - \frac1{4N} }
     + O(N^{-3})
    = \ln`*(2^{-d}\begin{pmatrix}d\\d/2+r/2\end{pmatrix})
    + \frac{d}{2N} + \frac{A_{m,r}}{N} +
    \frac{B_{d,m,r}}{N^2}
    + O(N^{-3})\ ,
  \end{dmath*}
  with
  \begin{subequations}
    \begin{align}
      A_{m,r} &= \frac12\, r(m-r) \ ;
      \\  
      B_{d,m,r} &= \frac14\,`*[d^2 - d`*(1 + m^2 + 2r^2 - 2mr) + 2r^2 -2mr]\ .
    \end{align}
  \end{subequations}
  Using $0\leqslant \abs{r}\leqslant d$, write
  \begin{dmath*}
    B_{d,m,r} \geq \frac14`*[d^2 - d`*(1+m^2 + 2 d^2 + 2\abs{m}d) -2\abs{m}d]
    \geq \frac14`*[-2d^3 - d^2`*(2\abs{m} - 1) - d`*(1+m^2 + 2\abs{m})]
    ~ \mathrel{=:}~  - \frac14\, C_{d,m}\ .
  \end{dmath*}
  Then,
  \begin{align}
    F(\rho_d^{m,m}, \rho_d^{0,0})
    &\geq
    \ee^{d/(2N)} \, 2^{-d}\,
    \sum_{r=-d}^d \begin{pmatrix} d \\ d/2+r/2 \end{pmatrix}
    \exp`*{\frac{A_{m,r}}{N} + \frac{B_{d,m,r}}{N^2}  + O(N^{-3}) }
    \nonumber\\
    &\geq
    \exp`*{-\frac{C_{d,m}}{4N^2} + O(N^{-3})}\, \ee^{d/(2N)} \, 2^{-d} \, 
    \sum_{r=-d}^d \begin{pmatrix} d \\ d/2+r/2 \end{pmatrix}
    \exp`*{\frac{A_{m,r}}{N}}\ .
    \label{eq:fuiebfdljkncxmfabhsjkna}
  \end{align}
  Recall the identities
  \begin{subequations}
    \begin{align}
      \sum_{r=-d}^{d} \begin{pmatrix} d \\ d/2+r/2 \end{pmatrix}
     &= \sum_{k=0}^d \begin{pmatrix} d \\ k \end{pmatrix} = 2^d\ ;
      \\
      \sum_{k=0}^d \begin{pmatrix} d \\ k \end{pmatrix} k
     &= d\,2^{d-1}\ ;
      \\
      \sum_{k=0}^d \begin{pmatrix} d \\ k \end{pmatrix} k^2
     &= (d + d^2)\,2^{d-2}\ .
    \end{align}
  \end{subequations}
  We have
  $\exp`*{A_{m,r}/N} = 1 + r(m-r)/(2N) + r^2(m-r)^2/(8 N^{2}) + O(N^{-3}) \geq
  {1 + r(m-r)/(2N) + O(N^{-3})}$.  Replacing the summation index $r$ by
  $k=(d+r)/2=0,1,\ldots,d$, we calculate
  \begin{dmath*}
    \frac{2^{-d}}{2N} \sum_{r=-d}^d \begin{pmatrix} d \\ d/2+r/2 \end{pmatrix}
    r(m-r)
    = \frac{2^{-d}}{2N} \sum_{k=0}^d \begin{pmatrix} d \\ k \end{pmatrix}
    (2k-d)(m-2k+d)
    = \frac{2^{-d}}{2N} `*[ (-d\,m-d^2)\, \sum_{k=0}^d \begin{pmatrix} d \\ k \end{pmatrix}
    + (2m+2d+2d)\,  \sum_{k=0}^d \begin{pmatrix} d \\ k \end{pmatrix} k
    - 4 \sum_{k=0}^d \begin{pmatrix} d \\ k \end{pmatrix} k^2 ]
    = \frac{2^{-d}}{2N} `*[ `*(-d\,m-d^2)\, 2^d
    + `*(2m+2d+2d)\,  \frac{d}{2}\,2^{d}
    - 4 \frac{d+d^2}{4}\,2^{d} ]
    = -\frac{d}{2N}\ ,
  \end{dmath*}
  and then
  \begin{dmath*}
    2^{-d} \sum_{r=-d}^d \begin{pmatrix} d \\ d/2+r/2 \end{pmatrix}
    \exp`*{\frac{A_{m,r}}{N}}
    \geq
    2^{-d} \sum_{r=-d}^d \begin{pmatrix} d \\ d/2+r/2 \end{pmatrix}
    `*(1 + \frac{r(m-r)}{2N} + O(N^{-3}))
    = 1 - \frac{d}{2N} + O(N^{-3})\ .
  \end{dmath*}
  Finally, plugging into~\eqref{eq:fuiebfdljkncxmfabhsjkna} gives us
  \begin{align}
    F(\rho_d^{m,m}, \rho_d^{0,0})
    &\geq `*{ 1 - \frac{C_{d,m}}{4N^2} + O(N^{-3})}
    `*{ 1 + \frac{d}{2N} + \frac{d^2}{8N^2} + O(N^{-3})}
    `*{ 1 - \frac{d}{2N} + O(N^{-3})}
    \nonumber\\
    &\geq 1  - \frac{C_{d,m}}{4 N^2} - \frac{d^2}{8 N^2} + O(N^{-3})\ ,
  \end{align}
  so we may define $D_{d,m}=C_{d,m}/4 + d^2/8$, proving the claim.
\end{proof}


%% file: udproof.tex
\subsection{Equivalence of the existence of a universal transversal gate set
  and the $\UU(d_L)$-covariance property of the code}

First, we show that the setting of the Eastin-Knill theorem is equivalent to
studying the $\UU(d_L)$-covariance property of the corresponding code.
More precisely, we show that given a code $V_{L\to A}$, if there exists a
mapping $u$ of logical unitaries $U_L$ to transversal physical unitaries
$u(U_L) = U_1\otimes\cdots\otimes U_n$ satisfying $V^\dagger u(U_L) V = U_L$ for
all $U_L$, where $u$ does not even have to be continuous, then the code
$V_{L\to A}$ is necessarily covariant with respect to the full unitary group on
the logical space.

The statement is pretty intuitive, because given any rule that maps logical
unitaries to physical transversal unitaries, we can compose the physical
unitaries corresponding to different logical unitaries, and presumably generate
a \emph{bona fide} representation by starting from a minimal generating set of
unitaries.  This intuition proves correct, though it is not immediately clear if
the mapping generated in this way is continuous.  Here we provide a derivation
that smooths out these technical details.
\begin{proposition}
  \label{prop:eastin-knill-wlog-full-representation}
  Let $V_{L\to A}$ be any code, with $A=A_1\otimes\cdots\otimes A_n$.  Suppose
  that for each unitary $U_L$ on $L$ there exists a transversal unitary
  $U_A = u(U_L) = u_1(U_L) \otimes \cdots\otimes u_n(U_L)$ such that
  $V^\dagger u(U_L) V = U_L$ for all $U_L$.  Then there exists a mapping $u'$
  that maps any $U_L$ to a transversal physical unitary
  $u'(U_L) = u_1'(U_L) \otimes \cdots \otimes u_n'(U_L)$ such that
  \begin{itemize}
  \item $u'$ is continuous;
  \item for all $U_L$, $V^\dagger u'(U_L) V = V^\dagger u(U_L) V = U_L$; and
  \item for any $U_L, U_L'$, we have
    $u'(U_L U_L') = u'(U_L)\, u'(U_L')$.
\end{itemize}
\end{proposition}
\begin{proof}[*prop:eastin-knill-wlog-full-representation]
  Observe first that for all $U_L$, because $u(U_L)$ implements a logical
  unitary, it must fix the code space $\Pi=VV^\dagger$.  Hence, we must
  necessarily have $[u(U_L), VV^\dagger] = 0$ for all $U_L$.

  Let $T_L^{(j)}$ be a basis of the Lie algebra $\uu(d_L)$ of $\UU(d_L)$.
  Let $(\vartheta_k)$ be a sequence of positive reals converging to zero, and
  for each $j$, consider the sequence of transversal physical unitaries
  $`\big(u`\big(\ee^{-i\vartheta_k T_L^{(j)}}))_k$.  Let $\delta>0$.  Since the
  sequence of unitaries is supported on a compact set, it admits a convergent
  subsequence and hence, there exist $\vartheta',\vartheta''$ such that
  $\abs{\vartheta'-\vartheta''}\leqslant\delta$ and
  $\norm[\big]{u`\big(\ee^{-i\vartheta'T_L^{(j)}}) -
    u`\big(\ee^{-i\vartheta''T_L^{(j)}})}_\infty\leqslant \delta$.  We define
  $\tilde{U}^{(j)} =
  u^\dagger`\big(\ee^{-i\vartheta''T_L^{(j)}})\,u`\big(\ee^{-i\vartheta'T_L^{(j)}})$,
  which then satisfies
  \begin{align}
    \norm[\big]{ \tilde{U}^{(j)} - \Ident }_\infty\leqslant \delta\ .
    \label{eq:pwbfdu2qnfidsjfewifhd}
  \end{align}
  We also have that
  $V^\dagger\, \tilde{U}^{(j)} V = V^\dagger\,
  u^\dagger`\big(\ee^{-i\vartheta''T_L^{(j)}})\,V V^\dagger\,
  u`\big(\ee^{-i\vartheta'T_L^{(j)}}) V =
  \ee^{i\vartheta''T_L^{(j)}}\ee^{-i\vartheta'T_L^{(j)}} =
  \ee^{-i\alpha{}T_L^{(j)}}$, with $\alpha=\vartheta'-\vartheta''$, where we
  recall that $[u(U_L),VV^\dagger ]=0$ for any $U_L$.  Define
  \begin{align}
    \tilde{U}_i^{(j)}
    &= \ee^{-i\chi_i^{(j)}}\,
      u_i^\dagger`\big(\ee^{-i\vartheta''T_L^{(j)}})\,u_i`\big(\ee^{-i\vartheta'(T_L^{(j)})})\ ;
    &
      T_i^{(j)} &= \alpha^{-1}\,i\log`\big(\tilde{U}_i^{(j)})\ ,
  \end{align}
  where $\chi^{(j)}_i$ is chosen such that there exists an eigenvector
  $\ket{\chi^{(j)}_i}$ of $\tilde{U}_i^{(j)}$ with eigenvalue exactly equal to
  one.  Choosing $\chi^{(j)}$ in $\intervalco{-\pi}{\pi}$ such that
  $\chi^{(j)} = \sum_i \chi_i^{(j)}~\mathrm{mod}~2\pi$, we then have
  $\tilde{U}_1^{(j)}\otimes\cdots\otimes\tilde{U}_n^{(j)} = \ee^{-i\chi^{(j)}}
  \, \tilde{U}^{(j)}$.
  Recall that for any operator $X$, we have
  $\norm{X}_\infty = \max_{\ket\phi,\ket\psi} \Re\tr`{\ketbra\psi\phi\,X}$ where
  the optimization ranges over vectors satisfying
  $\norm{\ket\phi},\norm{\ket\psi}_\infty \leqslant 1$. 
  We have using~\eqref{eq:pwbfdu2qnfidsjfewifhd} and for a suitably
  chosen phase $\ee^{-i\xi}$ that
  \begin{align}
    \delta
    &\geqslant \Re\tr`\Big{\ee^{-i\xi} \,`\Big(\bigotimes\proj{\chi_i^{(j)}})
    `\big[\tilde{U}^{(j)} - \Ident] }
      \nonumber\\
    &= \Re`\Big{\ee^{-i\xi} \, `\Big(\bigotimes\bra{\chi_i^{(j)}})
      `\Big[`\Big(\bigotimes\ee^{i\chi_i^{(j)}}\tilde{U}_i^{(j)}) - \Ident]
      `\Big(\bigotimes\ket{\chi_i^{(j)}}) }
    \nonumber \\
    &= \Re`\Big{\ee^{-i\xi} \, `\Big(\ee^{i\chi^{(j)}} - 1) }
    \nonumber \\
    &= \abs[\big]{ \ee^{i\chi^{(j)}} - 1 }\ ,
  \end{align}
  where the phase $\ee^{-i\xi}$ is chosen such that
  $\ee^{-i\xi}`\big(\ee^{i\chi^{(j)}}-1)$ is real positive.  This implies that
  $\delta\geqslant\Im`{\ee^{i\chi^{(j)}}}$ and
  $\Re`{\ee^{i\chi^{(j)}}}\geqslant 1-\delta$, which implies in turn
  $\abs{\chi^{(j)}}\leqslant\arcsin(\delta)\leqslant 2\delta$.  This also
  implies that
  $\norm[\big]{\ee^{-i\chi^{(j)}}\tilde{U}^{(j)} - \tilde{U}^{(j)}}_\infty
  =\abs[\big]{\ee^{-i\chi^{(j)}}-1} \norm[\big]{\tilde{U}^{(j)}}_\infty
  \leqslant \delta$, and hence by triangle inequality
  \begin{align}
    \norm[\big]{\ee^{-i\chi^{(j)}}\tilde{U}^{(j)} - \Ident}_\infty
    \leqslant
    \norm[\big]{\ee^{-i\chi^{(j)}}\tilde{U}^{(j)} - \tilde{U}^{(j)}}_\infty +
    \norm[\big]{\tilde{U}^{(j)} - \Ident}_\infty
    \leqslant 2\delta\ .
  \end{align}
  Similarly, for each $i$, we have
  $\norm[\big]{\tilde{U}_i^{(j)} - \Ident}_\infty =
  \Re\tr`\big[\ketbra{\psi_i}{\phi_i}`\big(\tilde{U}_i^{(j)}-\Ident)]$ for some
  $\ket{\psi_i}$,$\ket{\phi_i}$, and hence
  \begin{align}
    2\delta
    &\geqslant \norm[\big]{\ee^{-i\chi^{(j)}}\tilde{U}^{(j)} - \Ident}_\infty
    \nonumber\\
    &\geqslant \Re\tr`\bigg{ `\bigg[ %
      \ketbra{\psi_i}{\phi_i} \otimes \bigotimes_{i'\neq i} \proj{\chi_{i'}^{(j)}} 
    ] \, `*(%
      \bigotimes\tilde{U}_i^{(j)} - \Ident) }
    \nonumber\\
    &= \Re\tr`\bigg{ `\bigg[\ketbra{\psi_i}{\phi_i}\otimes\Ident]`\bigg[
      \tilde{U}_i^{(j)}\otimes \bigotimes_{i'\neq i} \proj{\chi_{i'}^{(j)}} 
      - \Ident \otimes \bigotimes_{i'\neq i} \proj{\chi_{i'}^{(j)}} ] }
      \nonumber\\
    &= \Re\tr`\big{ \ketbra{\psi_i}{\phi_i}`\big(\tilde{U}_i^{(j)} - \Ident) }
      = \norm[\big]{\tilde{U}_i^{(j)} - \Ident}_\infty \ .
  \end{align}
  This implies that all eigenvalues of $\tilde{U}_i^{(j)}$ are $\delta$-close to
  one, and hence the corresponding phases are all close to zero; more precisely,
  every eigenvalue $\ee^{i\gamma}$ of $\tilde{U}_i^{(j)}$ satisfies
  $2\delta\geqslant\abs{\ee^{i\gamma} - 1}$; by the same reasoning as above, the
  statements $2\delta\geqslant\Im`{\ee^{i\gamma}}$ and
  $\Re`{\ee^{i\gamma}}\geqslant 1-2\delta$ imply that
  $\abs{\gamma}\leqslant\arcsin(2\delta)\leqslant 4\delta$, and hence
  \begin{align}
    \norm[\big]{\alpha T_i^{(j)}}_\infty \leqslant 4\delta\ .
  \end{align}
  Now we set $T_A^{(j)} = \sum_i T_i^{(j)} - `(\chi^{(j)}/\alpha)\,\Ident$,
  which is a sum of local terms.  We then have
  $\ee^{-i\alpha T_A^{(j)}} = \ee^{i\chi^{(j)}}\,
  \ee^{-i\alpha{}T_1^{(j)}}\otimes\cdots\otimes\ee^{-i\alpha{}T_n^{(j)}} =
  \ee^{i\chi^{(j)}}\, \tilde{U}_1^{(j)}\otimes\cdots\otimes \tilde{U}_n^{(j)} =
  \tilde{U}^{(j)}$,
  with also $\norm[\big]{\alpha T_A^{(j)}}_\infty \leqslant 4n\delta+2\delta$.
  For any $U_L$, the unitary $u(U_L)$ commutes with the code space $\Pi$, and
  therefore $[\tilde{U}^{(j)}, \Pi]=0$.  Furthermore, since
  $\ee^{-i\alpha T_A^{(j)}}$ and $T_A^{(j)}$ share the same eigenspaces, we also
  have that $[T_A^{(j)}, \Pi] = [T_A^{(j)}, VV^\dagger] = 0$.
  We thus have
  $\ee^{-i\alpha{}T_L^{(j)}} = V^\dagger \tilde{U}^{(j)} V = V^\dagger
  \ee^{-i\alpha T_A^{(j)}} V = \ee^{-i\alpha V^\dagger T_A^{(j)} V}$, and thus
  by taking the logarithm, we obtain $V^\dagger T_A^{(j)} V = T_L^{(j)}$, where
  no ambiguities arise from taking the logarithm since the operators
  $\alpha T_A^{(j)}$ and $\alpha T_L^{(j)}$ have small norm, controlled by a
  suitably small choice of $\delta$, and hence do not straddle the branch cut.

  We can then define the mapping $u'(U_L)$ as the Lie group representation of
  $\UU(d_L)$ generated by the operators $T_A^{(j)}$ that span the corresponding
  Lie algebra.  More explicitly, for any $U_L$, we may write
  $i\log(U_L) = \sum c_j T_L^{(j)}$ for some unique set of real coefficients
  $c_j = c_j(U_L)$ (up the zero-measure set of unitaries that have an eigenvalue
  that coincides with the logarithm branch cut), and we set
  \begin{align}
    u'(U_L) = \exp`*(-i \sum c_j(U_L)\,T_A^{(j)})
    \label{eq:snkajfiuhdsopf}
  \end{align}
  where $c_j$ are the unique coefficients of the expansion of $i\log(U_L)$ in
  terms of the $T_L^{(j)}$ as defined above.  We have that $u'(U_L)$ is
  transversal because each $T_A^{(j)}$ is a sum of local terms.  Then we also
  have for any $U_L$ that
  $V^\dagger u'(U_L) V = V^\dagger \exp`\big(-i \sum c_j(U_L)\,T_A^{(j)}) V =
  \exp`\big(-i \sum c_j(U_L)\,V^\dagger T_A^{(j)} V) = \exp`\big(-i \sum
  c_j(U_L)\,T_L^{(j)}) = U_L$.  Because $u'$ is a Lie group representation, it
  is continuous and compatible with the group structure.
\end{proof}

\subsection{Proof of the approximate Eastin-Knill bound}

Recall that each irrep of $U(d_L)$ is represented by a \emph{Young diagram}
$\lambda$, where $\lambda=(\lambda_1,\lambda_2,\cdots, \lambda_{d_L})$, and
$\lambda_1\geq \lambda_2\geq\cdots\geq \lambda_{d_L}=0$, and
$\lambda_i \in
\ZZ$. %
The dimension of each irrep is given by the Weyl dimension formula, which for
the $\UU(d_L)$ group is equal to the Schur polynomial $S_\lambda$ evaluated at
the vector $(1,1,\cdots, 1)$. More explicitly, it can be evaluated to
\begin{equation}
  \label{eq:rep-dim}
  D_\lambda= \prod_{1\leq i  <  j \leq d_L} { \frac{\lambda_i- \lambda_j +j-i}{j-i}}\ .
\end{equation}
To derive \cref{thm:approximate-EK-main} we need to first prove few intermediate
results. First, we will prove a bound on $D_\lambda$, based on $\lambda_1$.
\begin{lemma}
  \label{lem:hard_bound}
  The symmetric representation has the minimal dimension among the
  representations with fixed $\lambda_1$. More precisely, the following
  inequality holds,
  \begin{equation}
    \label{lem:young-ineq}
    D_{\text{Sym}^{\lambda_1}}=\binom{d_L-1+\lambda_1}{d_L-1}\leq D_\lambda\ ,
  \end{equation}
  where $\lambda$ is a representation of $\UU(d_L)$ and
  $D_{\text{Sym}^{\lambda_1}}$ is the dimension of the symmetric representation
  with the Young diagram $\lambda=(\lambda_1,0,0,\cdots,0)$.
\end{lemma}
\begin{proof}[*lem:hard_bound]
  Suppose that $\lambda_1=l$. We use the dimension formula
  \cref{eq:rep-dim}. Consider the logarithm of the dimension, which is (up to a
  fixed constant) equal to:
  \begin{equation}
    f(\lambda_2, \cdots ,\lambda_{d_L-1})
    = \sum_{1\leq i< j \leq d_L} \log (\lambda_i- \lambda_j +j-i)\ .
  \end{equation}
  Note that we fix $\lambda_1=l$ and $\lambda_{d_L}=0$, so they do not appear as
  parameters of $f$. Also, the vector
  $\hat \lambda=(\lambda_2,\cdots , \lambda_{d_L-1})$ is an integer vector in
  the simplex $\Delta$ with $d_L$ extremal points $\hat v_i\in \RR^{d_L-1}$,
  where $\hat v_0=(0,0,\cdots ,0)$, $\hat v_1=(l,0,\cdots ,0)$,$\cdots$, and
  $\hat v_{d_L-1}=(l,l,\cdots ,l)$.

  We first extend the function $f$ to all of the real points in $\Delta$, and
  show that $f$ is a concave function inside $\Delta$. This would show that the
  minimum of $f$ is attained at one of its extremal points.

  A direct computation of the Hessian of $f$, reveals that for
  $ 2\leq r,s \leq d_L-1$,
  \begin{align}
    H_{r,s}
    = \delta_{rs}\left[ -\sum_{1\leq i  \leq d_L, i\neq s}K_{is} \right]+(1-\delta_{rs}) K_{rs}\ ,
  \end{align}
  where $K_{rs}=1/(\lambda_s-\lambda_r+r-s)^2$. One can see that if
  $w=\sum_{2\leq i \leq d_L-1} \alpha_i e_i$ is an arbitrary vector, then
  \begin{equation}
    w^\dagger H w = - \left( \sum_{2\leq i \leq d_L-1}|\alpha_i|^2 (K_{1i}+K_{id_L})
      + \sum_{2\leq i<j \leq d_L-1} K_{ij} |\alpha_i -\alpha_j|^2\right)\ .
  \end{equation}
  This is a negative number, and shows that $f$ is strictly concave. Therefore,
  the minimum of $f$ is attained on one of the extremal point $\hat v_i$,
  $0\leq i \leq d_L-1$.  Using the Weyl-dimension formula we have,
  \begin{align}
    f(\hat v_i)= \prod_{j=0}^i {
    \frac{
    \binom{l+d_L-1-j}{l}
    }{
    \binom{l+j}{l}
    }
    }\ .
  \end{align}

  One can easily see that $f(\hat v_i)$ is increasing for $i\leq (d_L-1)/2$ and
  decreasing for $i\geq (d_L-1)/2$.  Therefore, its minimum is attained at
  $f(\hat v_0)= f(\hat v_{d_L-1})$.
\end{proof}

Consider a fixed element in the Cartan subalgebra of $\su(d_L)$, a
$d_L \times d_L$ matrix $T = \text{diag}(1,0,0,0,\cdots, -1)$, and $T_\lambda$,
the corresponding generator in the representation given by the Young diagram
$\lambda$.  We have the following lemma:

\begin{lemma}
  \label{lem:opnorm_length}
  It holds that
  $\norm{ T_\lambda }_\infty \leq \lambda_1$.
\end{lemma}
\begin{proof}[*lem:opnorm_length]
  A basis for the representation $\lambda$ is given by different semi-standard fillings of the Young diagram $\lambda$ with numbers $1 \cdots d_L$. 
  If we indicate fillings of the $\lambda$ by $m_\lambda$, then $\{\ket {m_\lambda}\}$ forms a basis for the representation $\lambda$. 
  Although this is not an orthogonal basis, if the number content of $m_\lambda$ and $m'_\lambda$ are different then $\ket {m_\lambda}$ and $\ket {m'_\lambda}$ are orthogonal. This basis diagonalizes $T_\lambda$.

  In particular, if $\#_i m_\lambda$ indicates the number of times that $i$ appears in the filling $m_\lambda$, then $\bra{m_\lambda} T_\lambda \ket {m_\lambda}= \#_1 m_\lambda - \#_{d} m_\lambda$.
  This immediately leads to the conclusion that the eigenvalues of $T_\lambda$ are $\#_1 m_\lambda - \#_{d} m_\lambda$, for different fillings $m_\lambda$. 

  For any semi-standard filling of the Young diagrams, the numbers are strictly increasing in the columns. Therefore, $\#_i m_\lambda \leq \lambda_1$, as there are no repeats in the columns. So we showed that eigenvalues of $T_\lambda$ are between $-\lambda_1$ and $\lambda_1$, which completes the proof.
\end{proof}

We are now in a position to prove \cref{thm:approximate-EK-main}.  First,
however, we prove a version of \cref{thm:approximate-EK-main} that provides a
stronger bound expressed as a binomial coefficient, which we will use to prove
the bounds stated in \cref{thm:approximate-EK-main}.
\begin{theorem}
  \label{thm:approximate-EK-stronger-version}
  Let $V_{L\to A}$ be an isometry that is covariant with respect to the full
  $\SU(d_L)$ group on the logical space, and write
  $\mathcal{E}(\cdot) =V(\cdot)V^\dagger$.  Consider the single erasure noise
  model represented by $\mathcal{N}$ in~\eqref{eq:noise-map-one-erasure} with
  equal erasure probabilities, $q_i=1/n$ for all $i$.  Then
  \begin{align}
    \max_i d_i &\geq
                 \binom{d_L-1+ \bigl\lceil{
                 (2n\epsilon_{\mathrm{worst}}(\mathcal{N}\circ\mathcal{E}))^{-1}
                 }\bigr\rceil
                 }{d_L-1}\ .
                 \label{eq:approx-EK-strongver-epsilonworst}
  \end{align}
  In terms of the average entanglement fidelity measure, the bound reads instead
  \begin{align}
    \max_i d_i &\geq
                 \binom{d_L-1+ \bigl\lceil{
                 (nd_L\epsilon_{\mathrm{e}}(\mathcal{N}\circ\mathcal{E}))^{-1}
                 }\bigr\rceil
                 }{d_L-1}\ .
                 \label{eq:approx-EK-strongver-epsilone}
  \end{align}
\end{theorem}

The bound in \cref{thm:approximate-EK-stronger-version} allows to derive
slightly stronger bounds than those obtained from the simplified expressions in
\cref{thm:approximate-EK-main}.  For instance, suppose that $d_L = d_i$ as in
the examples given in the main text.  The binomial coefficient $\binom{a+b}{b}$
is increasing in $b$, which can be seen using the recurrence relation
$\binom{a+b+1}{b+1} = \frac{a+b+1}{b+1}\binom{a+b}{b} \geqslant \binom{a+b}{b}$.
Also, the binomial coefficient $\binom{a+b}{b}$ for $b\geqslant2$ satisfies
$\binom{a+b}{b} \geqslant \binom{a+2}{2} = (a+2)(a+1)/2\geqslant a+2$ (assuming
$a\geqslant 1$).  Hence, if $d_L=d_i$, then
condition~\eqref{eq:approx-EK-strongver-epsilonworst} implies that
$\lceil{ (2n\epsilon_{\mathrm{worst}}(\mathcal{N}\circ\mathcal{E}))^{-1} }\rceil
\leqslant 1$, because otherwise we would have
$\binom{d_L-1 + \lceil{
    (2n\epsilon_{\mathrm{worst}}(\mathcal{N}\circ\mathcal{E}))^{-1}
  }\rceil}{d_L-1} \geqslant \binom{d_L-1 + 2}{2} \geqslant d_L+1$.  This implies
that, for $d_L=d_i$, we must have
$\epsilon_{\mathrm{worst}}(\mathcal{N}\circ\mathcal{E}) \geqslant 1/(2n)$.

\begin{proof}[*thm:approximate-EK-stronger-version]
  Combining \cref{lem:hard_bound} and \cref{lem:opnorm_length}, we get
  \begin{equation}
    D_\lambda\geq \binom{d_L-1+\lceil\norm{ T_\lambda }_\infty\rceil }{d_L-1}
  \end{equation}
  Now, we return to the original problem of approximate Eastin-Knill theorem,
  where the group $\SU(d_L)$ acts on physical subsystems.  We fix the generator
  $ T= \diag(1,0,0,0,\cdots, -1)$ of $\su(d_L)$, and let $T_i$ be the
  corresponding generator acting on the subsystem $i$.  Let
  $T_i=\bigoplus_\lambda T_\lambda$ be the decomposition of $T_i$ with respect
  to the decomposition of the representation on subsystem $i$, and assume that
  $\hat\lambda{(i)}$ is the Young diagram in this direct sum with the
  largest $\norm{T_\lambda}_\infty$. Therefore,
  $\norm{T_i}_\infty = \norm[\big]{T_{\hat{\lambda}(i)}}_\infty$, and we have:
  \begin{align}
    d_i \geq D_{\hat\lambda(i)}
    \geq \binom{d_L-1 + \lceil\norm{ T_{\hat\lambda(i)} }_\infty\rceil }{d_L-1}
    = \binom{d_L-1 + \lceil\norm{ T_{i} }_\infty\rceil }{d_L-1}\ .
    \label{eq:hjuytihljkvguyiuhjlknb}
  \end{align}
  This implies
  \begin{align}
    \max_i d_i
    \geqslant
    \max_i \binom{d_L-1 + \lceil\norm{ T_{i} }_\infty\rceil }{d_L-1}
    = \binom{d_L-1 + \lceil\max_i \norm{ T_{i} }_\infty\rceil }{d_L-1}\ .
  \end{align}
  Let $i'$ denote the index of the subsystem that maximizes
  $\norm{T_i}_\infty$, such that our bound~\eqref{eq:simple-bound} with
  $\Delta T_L=2$ and $\Delta T_i\leqslant 2\norm{T_i}_\infty$ reads
  \begin{align}
    \epsilon_{\mathrm{worst}} \geqslant \frac1{2n\norm{T_{i'}}_\infty}\ ,
    \label{eq:rteyutkgluhkljopihgh}
  \end{align}
  noting that
  $\max_i \Delta T_i \leqslant 2\max_i \norm{T_i}_\infty =
  2\norm{T_{i'}}_\infty$, and writing
  $\epsilon_{\mathrm{worst}} =
  \epsilon_{\mathrm{worst}}(\mathcal{N}\circ\mathcal{E})$ as a shorthand.
  Therefore,
  $\norm{T_{i'}}_\infty \geqslant (2n\epsilon_{\mathrm{worst}})^{-1}$, and we
  obtain
  \begin{align}
    \max_i d_i \geq
    \binom{d_L-1+ \bigl\lceil{
    (2n\epsilon_{\mathrm{worst} })^{-1}
    }\bigr\rceil
    }{d_L-1}\ .
  \end{align}
  If we had used the bound~\eqref{eq:main-thm-bound-avg-entgl-fid} instead
  of~\eqref{eq:simple-bound}, we would have instead
  of~\eqref{eq:rteyutkgluhkljopihgh} that
  \begin{align}
    \epsilon_{\mathrm{e}} \geqslant \frac{1}{n d_L\norm{T_{i'}}_\infty}\ ,
  \end{align}
  and we can perform the replacement
  $\epsilon_{\mathrm{worst}} \to (d_L/2)\,\epsilon_{\mathrm{e}}$ in the
  argument above.
\end{proof}

\begin{proof}[*thm:approximate-EK-main]
  We use the following standard inequality of binomial coefficients.  For
  integers $a,b>0$, we have the two lower bounds
  \begin{align}
    \binom{a + b}{a} \geqslant
    \begin{cases}
      `*(1 + \frac{a}{b})^b \\[1ex]
      `*(1 + \frac{b}{a})^a\ ,
    \end{cases}
    \label{eq:binom-lower-bound}
  \end{align}
  noting that $\binom{a + b}{b} = \binom{a + b}{a}$.
  Consider~\eqref{eq:approx-EK-strongver-epsilonworst}, with $a=d_L-1$ and
  $b=\lceil(2n\epsilon_{\mathrm{worst}})^{-1}\rceil$.  The first lower bound
  gives us
  \begin{align}
    \max_i \ln(d_i)
    &\geqslant
      b\, \ln\,`*(1 + \frac{d_L-1}{b})
      \geqslant b\,\ln\,\frac{d_L-1}{b}
      \geqslant \frac1{2n\epsilon_{\mathrm{worst}}}\,
      \ln\,`*[\frac{d_L-1}{(2n\epsilon_{\mathrm{worst}})^{-1}+1}]\ ,
  \end{align}
  and hence
  \begin{align}
    \max_i \ln(d_i)
    &\geqslant \frac1{2n\epsilon_{\mathrm{worst}}}\, \ln\,`*(d_L-1)
    - \frac{\ln\,`*(1 + (2n\epsilon_{\mathrm{worst}})^{-1})}{2n\epsilon_{\mathrm{worst}}}\ .
  \end{align}
  This proves~\eqref{eq:approximate-EK-di-geq-poly-dL}.

  Applying the second bound in~\eqref{eq:binom-lower-bound}
  to~\eqref{eq:approx-EK-strongver-epsilonworst}, we obtain
  \begin{align}
    \max_i d_i
    &\geqslant 
    `*[ \frac{d_L-1 + \lceil(2n\epsilon_{\mathrm{worst}})^{-1}\rceil}{d_L-1} ]^{d_L-1}
    \nonumber\\
    &= \exp`*{ (d_L-1)
      \ln`*(1 + \frac{\lceil(2n\epsilon_{\mathrm{worst}})^{-1}\rceil}{d_L-1}) }\ .
      \label{eq:urjhkgiuojklij}
  \end{align}
  We can rearrange~\eqref{eq:urjhkgiuojklij} to
  \begin{align}
    \exp`*{\frac{\max_i\ln(d_i)}{d_L-1}} - 1
    \geqslant \frac{\lceil(2n\epsilon_{\mathrm{worst}})^{-1}\rceil}{d_L-1}
    \geqslant \frac{(2n\epsilon_{\mathrm{worst}})^{-1}}{d_L-1}\ ,
  \end{align}
  which in turn implies
  \begin{align}
    \epsilon_{\mathrm{worst}}
    \geqslant \frac1{2n`(d_L-1)}\,`*[ \max_i `*(\exp`*{ \frac{\ln(d_i)}{d_L-1} } - 1) ]^{-1}\ .
    \label{eq:tfygufidjkblhfjhs}
  \end{align}
  Henceforth we let $i$ denote the index of the physical subsystem with largest
  dimension, i.e., $d_i = \max_{i'} d_{i'}$.
  For large $d_L$, we have
  \begin{align}
    (d_L-1) `*(\exp`*{ \frac{\ln(d_i)}{d_L-1} } - 1)
    = \ln(d_i) + O`\bigg(\frac{\ln^2(d_i)}{d_L})
    = \ln(d_i)\,`*[ 1 + O`\bigg(\frac{\ln(d_i)}{d_L}) ]\ ,
  \end{align}
  and thus
  \begin{align}
    \epsilon_{\mathrm{worst}}
    \geqslant \frac1{2n \, \max_i \ln(d_i)}`*[ 1 + O`*(\frac{\ln(d_i)}{d_L}) ]
    =  \frac1{2n \, \max_i \ln(d_i)} + O`*(\frac{1}{n d_L})\ ,
  \end{align}
  which is the desired bound~\eqref{eq:appek1}.
  The bound~\eqref{eq:approximate-EK-di-geq-exp-dL} follows
  from~\eqref{eq:urjhkgiuojklij} by noting that
  $\lceil (2n\epsilon_{\mathrm{worst}})^{-1}\rceil \geqslant
  (2n\epsilon_{\mathrm{worst}})^{-1}$ and that $\log(1+x) \geqslant \log(x)$.

  The alternative expressions for $\epsilon_{\mathrm{e}}$ follow from the use of
  the bound~\eqref{eq:approx-EK-strongver-epsilone}, following the same steps as
  above while effecting the replacement
  $\epsilon_{\mathrm{worst}} \to d_L \epsilon_{\mathrm{e}}/2$.
\end{proof}


%% file: Random_Consts.tex
The proof of \cref{thm:randomcode} %
is technical, and relies on the recent developments in the representation theory
of $\UU(d)$, and new counting formulas for the \emph{Littlewood-Richardson}
coefficients. Here, we sketch the proof strategy, and refer the reader to
\cref{appx:Random_Theorems} for the technical details.

Although our randomized constructions do not properly work for producing good
$\UU(2)$-covariant codes,\footnote{More precisely, our techniques do not lead to
  proper lower bounds for the fidelity of recovery of random $\UU(2)$-covariant
  codes, but this might only be caused by not lower bounding the fidelity of
  recovery with strong enough inequalities.\\} for the $\UU(3)$ case we can find
explicit (non-asymptotic) bounds with a slightly different scaling. There, one can benefit from the fact that the fusion rules of $\UU(3)$ representation theory are known~\cite[Section~5]{rassart2004polynomiality}. We will not discuss
$\UU(3)$ case further, and will focus on $d_L\geq 4$ for the rest of this
section.

Consider codes that map logical information on the Hilbert space $\Hil_L$ to three physical
subsystems $\Hil_A = \Hil_{A_1}\otimes \Hil_{A_2}\otimes \Hil_{A_3}$, and denote by $d_i$ the dimension
of $\Hil_{A_i}$.  In order to precisely define what we mean by the random isometry $V$, consider the state corresponding to $V$ (similar to what we did in the analysis of correlation
functions in \cref{appx:corrbound}). The corresponding state, $\ket \Psi$, lives
on $\Hil_R\otimes \Hil_{A_1}\otimes \Hil_{A_2}\otimes \Hil_{A_3}$, where as before $\Hil_R \simeq \Hil_L$ is a
reference system. The covariance of $V$ translates to the invariance of
$\ket \Psi$:
\begin{align}
  `\big[ \overline{U} \ot r_1(U) \ot r_2(U) \ot r_3 (U) ] \, \ket \Psi_{RA_1A_2A_3}
  = \ket \Psi_{RA_1A_2A_3} \text{ for all } U\in \UU(d_L).
\end{align}

Therefore, $\Psi_{RA_1A_2A_3}$ lives on an invariant subspace of the unitary
group.  The projector to this invariant subspace is given by
\begin{align}
  \Pi_{RA_1A_2A_3} = \int dU\,\overline{U}\ot r_1(U) \ot r_2(U)\ot r_3(U)\ .
\end{align}
We denote by $d_P = \tr`\big(\Pi_{RA_1A_2A_3})$ the dimension of the invariant
subspace.  Further, define
$\Pi_{RA_i}:=\tr_{A\setminus A_i}`\big(\Pi_{RA_1A_2A_3})$ and
$\Pi_{\widehat{R A_i}}:=\tr_{RA_i}`\big(\Pi_{RA_1A_2A_3})$.

Now, we can chose the state $\ket \Psi_{RA_1A_2A_3}$ randomly from $\Pi_{RA_1A_2A_3}$,
and \emph{define} $V$ to be the corresponding isometry, i.e.,
$V_{L\rightarrow A_1A_2A_3}:= \bra{\Phi}_{LR}\ket \Psi_{RA_1A_2A_3}$, where
$\ket\Phi = \sum \ket{k}_L\ket{k}_R$ for some standard choice of bases on $\Hil_L$
and $\Hil_R$.

As in \cref{appx:corrbound}, we consider single erasures at known locations,
i.e., the noise channel is given by
$\mathcal{N}(\cdot) = \sum q_i \proj{i}_C\otimes\mathcal{N}^i(\cdot)$, where
$\mathcal{N}^i$ erases the $i$-th system as
per~\eqref{eq:noise-map-general-alpha}.  If the isometry $V$ is chosen at random
in the space of covariant isometries, then on average, the fidelity of recovery of the code defined by the isometry is lower bounded as follows.
\begin{lemma}
  \label{lem:avg_fidelity}
  Suppose that the covariant isometry $V$ is chosen randomly as above. Then, the
  infidelity of the code after erasure of subsystem $i\in \{1,2,3\}$,
  averaged over all covariant isometries, satisfies the following inequality:
  \begin{align}
    \frac12
    \EE[\epsilon^2_{\mathrm{e}}(\mathcal{N}^i\circ\mathcal{E})] 
    \leq \frac12 \norm*{ \frac{\Pi_{RA_i}}{d_P}-\frac{\Ident_{RA_i}}{d_L d_i} }_1
    +\frac12 \sqrt{d_L  d_i}\sqrt{\frac{\tr`\big(\Pi_{\widehat{RA_i}}^2)}{d_P^2}}\ .
  \end{align}
\end{lemma}
Intuitively, \cref{lem:avg_fidelity} states that in order to get good quantum
codes we need to do the followings:
\begin{enumerate}
\item Control the constant offset,
  $\norm*{ d_P^{-1} \Pi_{RA_i} - \Ident_{RA_i}/(d_R d_i)}_1$.  This can be
  achieved by making sure that $\Pi_{RA_i}$ is close to a multiple of identity.
\item Control the \emph{fluctuations} by minimizing
  $d_p^{-2}\tr`\big(\Pi_{\widehat{RA_i}}^2)$. Note that this is the purity of density matrix $\Pi_{\widehat{RA_i}}/d_p$, so it would be small if $\Pi_{\widehat{RA_i}}$ is close to a multiple of a projector.
\end{enumerate}
\Cref{lem:avg_fidelity} is how far we can go without discussing the detailed
representation theory of $\UU(d_L)$. From now on, we focus on analyzing
$\Pi_{RA_i}$ and $\Pi_{\widehat{RA_i}}$.

Without loss of generality assume that $i=1$. Also, suppose
$\lambda, \mu,\nu$ are the Young diagrams defining the irreducible
representations $r_1$, $r_2$ and $r_3$. Similarly, $r_{e_1}(U)=U$, where
$e_{1}=(1,0,0,\cdots,0)$ is the Young diagram of the standard
representation. Now, we use representation theory techniques to explicitly compute
$\Pi_{RA_1}$ and $\Pi_{\widehat{RA_1}}=\Pi_{A_2A_3}$.

The degeneracies of fusion of different irreps of $\UU(d_L)$ are known, and
specified by the so called Littlewood-Richardson coefficients
$c^{\theta}_{\mu\nu}$:
\begin{align}
  r_{\mu} \ot r_{\nu} = \bigoplus_{\theta} r_{\theta} \ot I_{c^{\theta}_{\mu\nu}}\ .
\end{align}
A specific case of this formula which is also applicable to our analysis is a
version \emph{Pieri formula} (See Appendix~A.1
of~\cite{fulton2013representation}): If $e_i$ is the $i$-th computation basis
vector, then
\begin{align}
  \overline{r}_{e_1} \ot r_{\lambda} = \bigoplus_{i \in \mathcal I} r_{\lambda-e_i}\,
\end{align}
where $\mathcal I\in \{1,2,\cdots,d_L\}$ is the index set that $\lambda-e_i$ is
a valid Young diagram, i.e., a non-increasing sequence. In particular, if $\lambda$
is strictly decreasing then $\mathcal I= \{1,2,\cdots,d_L\}$. This relation can be derived by either directly applying the Littlewood-Richardson rule~\cite[Appendix~A.1]{fulton2013representation}, or starting from the standard Pieri formula and dualizing representations. With this, we have
\begin{equation}\label{eq:Exp_Pi}
\Pi_{RA_1A_2A_3}=\int dU(\bigoplus_{i\in \mathcal I} r_{\lambda-e_i}(U) ) \ot
(\bigoplus_{\theta} r_\theta(U) \ot I_{c^\theta_{\mu\nu}}).
\end{equation} 
From the Schur orthogonality relations for compact groups (Peter-Weyl theorem), we have that $\int dU[ \tr{\overline{r_{\beta}(U)}} ] r_{\alpha}(U) =\delta_{\alpha\beta} I_\alpha/d_\alpha $. Applying this to \cref{eq:Exp_Pi} we get the following explicit relations:
\begin{align}\label{eq:Pi1}
  &\Pi_{RA_1}=\bigoplus_{i\in \mathcal I}
    \frac{c_{\mu\nu(\lambda-e_i)}}{d_{\lambda-e_i}} I_{d_{\lambda-e_i}},\\
  &\Pi_{A_2A_3}=\bigoplus_{i\in \mathcal I}
    \frac{1}{d_{\lambda-e_i}} I_{d_{\lambda-e_i}}\ot I_{c_{\mu\nu(\lambda-e_i)}}\ \label{eq:Pi2},
\end{align}
where $c_{\mu\nu\lambda} := c^{\overline\lambda}_{\mu\nu}$, and
$\overline\lambda$ is the dual of $\lambda$. Recall that in order for the random
codes to perform well, we need that $\Pi_{RA_1}$ and $\Pi_{A_2A_3}$ to be close
to multiples of projectors. \Cref{eq:Pi1,eq:Pi2} show that to achieve this we
only need $\frac{c_{\mu\nu(\lambda-e_i)}}{d_{\lambda-e_i}}$ and
$ \frac{1}{d_{\lambda-e_i}}$ to be almost constants as $i$ varies. The following
lemma makes this observation quantitative:
\begin{lemma}
  \label{lem:smooth}
  Suppose that $0\leq \delta\leq 1/2$ is a real number such that for all $i\in \mathcal I$,
  \begin{eqnarray}
    1-\delta \leq &\frac{c_{\mu \nu (\lambda-e_i)}}{c_{\mu \nu \lambda}}&\leq 1+\delta \\
    1-\delta \leq &\frac{d_{ \lambda-e_i}}{d_\lambda}&\leq 1+\delta,
  \end{eqnarray}
  then,
  \begin{align}
    \frac12
    \EE[\epsilon^2_{\mathrm{e}}(\mathcal{N}^{1}\circ\mathcal{E})]
    \leq 4\delta + \frac{5}{2\sqrt{c_{\mu \nu \lambda}}}\ .
  \end{align}
\end{lemma}

\Cref{lem:smooth} demonstrates that in order to get useful lower bounds on the
fidelity of the codes, one has to show that $d_\lambda$ and
$c_{\mu\nu\lambda}$ are stable under perturbations by $e_i$.  We
construct our irreps such that they achieve this stability.

Define $\abs{\lambda}: = \sum_i {\lambda_i}$ for arbitrary Young
diagram $\lambda$. It is known that if
$\abs{\mu}+\abs{\nu}+\abs{\lambda} \neq 0$, then $c_{\mu\nu\lambda}=0$.  Now,
the construction is as follows: Fix a triplet of Young diagrams
$(\hat \mu, \hat \nu, \hat \lambda)$ such that
$\abs{\hat \mu_i}+\abs{\hat \nu_i}+\abs{\hat \lambda_i} = 0$ and
set
\begin{align}
  (\mu,\nu,\lambda) = (N\hat \mu+e_1,N \hat \nu, N\hat \lambda)\ ,
\end{align}
for large values of $N$.  We used $N\hat \mu+e_1$ instead of
$N\hat \mu$ is to ensure that $\abs{\mu}+\abs{\nu}+\abs{\lambda-e_i} = 0$ as we
need $c_{\mu\nu(\lambda-e_i)}$ to be non-zero.

Showing smoothness of $d_\lambda$ is much simpler, because by the \emph{Weyl
  dimension formula} (see~\cite[Section 15.3]{fulton2013representation}) it is polynomial in
$\lambda=(\lambda_1,\lambda_2,\cdots,\lambda_{d_L})$. So by basic Taylor
expansion we have
$d_{\lambda+e_i}=d_\lambda + \partial d_\lambda/\partial \lambda_i +
1/2\,\partial^2 d_\lambda/\partial \lambda_i^2 +\cdots$. Note that the total
degree of the terms in the sum decreases by differentiation. Hence, $d_\lambda$ is the dominant term in the
expansion of $d_{\lambda+e_i}$ and other terms
are lower order in $N$. Therefore, there exist $N_0'$ and $C_0'$
such that for $N\geq N_0'$,
\begin{equation}
  \label{eq:as_1}
  1-\frac {C_0'}{N} \leq \frac{d_{\lambda-e_i}}{d_\lambda} \leq 1+\frac{C_0'}{N}.
\end{equation}
The Littlewood-Richardson coefficients are much more complicated. They can be
computed using efficient algorithms, such as the Littlewood-Richardson rule, but
there is no explicit formula. In fact, they are specific cases of the called
\emph{Kronecker coefficients} whose computation is known to be
NP-hard~\cite{Ikenmeyer2017cc_vanishing}. However, a series of new developments
in the representation theory of the unitary group has revealed interesting
polynomiality properties for the LR coefficients.

It is known that $c_{\mu\nu\lambda}$ as a function of $\mu$, $\nu$ and $\lambda$ (in
the $3d_L - 1$ dimensional subspace constrained by the condition
$\abs{\mu}+\abs{\nu}+\abs{\lambda} = 0$) is non-zero if and only if $(\mu,\nu, \lambda)$ is in a particular convex
cone. This cone, or~\emph{chamber complex}, is then divided to several sub-cones
or~\emph{chambers}. In Ref.~\cite{rassart2004polynomiality} it is shown that
$c_{\mu\nu\lambda}$ is a polynomial within each chamber. See
\cref{fig:chamberc}.

We chose $(\hat \mu, \hat \nu, \hat \lambda)$ such that it is in the
interior of one of the chambers, and $c_{N\hat \mu, N\hat \nu, N\hat \lambda}$
is not constant. If $N$ is large enough, $c_{\mu \nu (\lambda-e_i)}$ will remain
in the interior of the same chamber for all $i$, and are described by
the same polynomial. Hence, similar to $d_\lambda$, we have:
\begin{equation}
  \label{eq:as_2}
  1-\frac {C_0''}{N}
  \leq \frac{c_{\mu\nu(\lambda-e_i)}}{c_{\mu\nu\lambda}}
  \leq 1+\frac{C_0''}{N}\ ,
\end{equation} 
where $N\geq N_0''$ and for some $N_0''$ and $C_0''$. Clearly, these bounds the smoothness required for \cref{lem:smooth} to work, and therefore we
get our main theorem. Their detailed proof will be in
\cref{appx:Random_Theorems}.
\begin{theorem}
  \label{thm:main_random}
  Suppose that $d_L \geq 4$. There exist Young diagrams $\hat \lambda$,
  $\hat \mu$ and $\hat \nu$, an integer $N_0$, and a constant $C_0$, such that
  if $(\mu,\nu,\lambda)=(N \hat \mu+e_1,N \hat \nu, N\hat \lambda)$ and $V$ is a
  random covariant isometry in the sense of~\eqref{eqn:cov_def}, we have,
  \begin{equation}
    \label{eq:main_thm_1}
    \epsilon_{\mathrm{e}}(\mathcal{N}\circ\mathcal{E})\leq \frac{C_0}{\sqrt N}
    \quad \text{ for } \quad N\geq N_0\ .
  \end{equation}
  For these constructions, we have,
  \begin{equation}
    \label{eq:main_thm_2}
    \epsilon_{\mathrm{e}}(\mathcal{N}\circ\mathcal{E})
    \leq C_1 \, (\max_i d_i)^{-\frac{1}{d_L(d_L-1)}}\ .
  \end{equation}
\end{theorem}

Finally, \cref{thm:randomcode} follows immediately from \cref{thm:main_random}.


%% file: Random_Theorems.tex
First, we prove \cref{lem:avg_fidelity}.
\begin{proof}[*lem:avg_fidelity]
  First, we express the error-correcting accuracy of the code $V$ according to
  the average entanglement fidelity in terms of the distance of the codewords to
  a maximally mixed state, including the reference system.  By
  B\'eny/Oreshkov~\eqref{eq:Beny-Oreshkov-f-fixedinput} (choosing
  $\zeta_E=\Ident_{A_i}/d_i$), we have
  \begin{align}
    f_{\mathrm{e}}(\mathcal{N}^i\circ\mathcal{E})
    \geqslant F`\big( \widehat{\mathcal{N}^i\circ\mathcal{E}}(\hat\phi_{LR}), 
    \zeta_E\otimes\hat\phi_R )
    =  F`*( \Psi_{RA_i}, \frac{\Ident_{A_i}}{d_i}\otimes\frac{\Ident_{R}}{d_L} )\ ,
  \end{align}
  and hence
  \begin{align}
    \epsilon^2_{\mathrm{e}}(\mathcal{N}^i\circ\mathcal{E})
    \leqslant 
    1 - F^2`\big( \Psi_{RA_i}, \frac{\Ident_{RA_i}}{d_L d_i} )
    \leqslant \norm*{\Psi_{RA_i} - \frac{\Ident_{RA_i}}{d_L d_i}}_1\ ,
    \label{eq:iaidoauojfnkjafd}
  \end{align}
  where we recall the usual relations between trace distance and the fidelity.

  We denote by $\EE$ the averaging over all possible invariant states
  $\Psi_{RA_1A_2A_3}$.  Taking an average over~\eqref{eq:iaidoauojfnkjafd} gives
  us
  \begin{align}
    \frac12\EE(\epsilon^2_{\mathrm{e}}(\mathcal{N}^i\circ\mathcal{E}))
    \leq \frac12\EE \, \norm*{\Psi_{RA_i}-\tau_{RA_i}}_1\ ,
  \end{align}
  where we write $\tau_{RA_i} = \Ident_{RA_i}/(d_L d_i)$.  Applying triangle
  inequality, Cauchy-Schwarz inequality, and the concavity of square root gives
  us (see Ref.~\cite{hayden2008decoupling} for similar calculations),
  \begin{align}
    \label{eq:lb1}
    \frac12 \EE(\epsilon^2_{\mathrm{e}}(\mathcal{N}^i\circ\mathcal{E}))
    &\leq \frac12 \EE \,\norm{ \Psi_{RA_i}-\tau_{RA_i} }_1
      \nonumber\\
    &\leq \frac 12\norm{ \EE \Psi_{RA_i} - \tau_{RA_i}}_1
      + \frac 12 \EE \norm{\Psi_{RA_i}-\EE \Psi_{RA_i}}_1
      \nonumber\\
    &\leq \frac 12 \norm{ \EE \Psi_{RA_i} - \tau_{RA_i} }_1 +
      \frac12 \sqrt{d_R d_i \left(\tr[\EE \Psi_{RA_i}^2]-\tr [(\EE \Psi_{RA_i})^2] \right)}.
  \end{align}
  Now, consider the rank-$d_P$ projector to the invariant space $\Pi_{RA_1A_2A_3}$
  that we constructed in \cref{sec:random-codes}. Define
  $L: \CC^{d_P} \rightarrow \Hil_A\otimes\Hil_R$ be the isometry mapping to
  the invariant space, which satisfies
  \begin{align}
    L^\dagger L &= \Ident_{d_P}\ ;\text{ and}
    & L L^\dagger &=\Pi_{RA_1A_2A_3}\ .
  \end{align}

  We can define $\ket \Psi_{RA}=L \ket \chi$, where $\ket \chi $ is a random state
  in $\CC^{d_p}$. Then,
  \begin{align}
    \EE \Psi_{RA_i} = \EE \tr_{A A_i}`*( L\chi L^\dagger)
    =\tr_{\widehat{R A_i}}`*( L (\EE \chi) L^\dagger)
    = \frac{1}{d_P}\tr_{\widehat{RA_i}} `*( L L^\dagger)
    = \frac{\Pi_{RA_i}}{d_P}\ ,
  \end{align}
  where we used $\EE \chi=\Ident/d_P$.  For simplicity, we henceforth set $i=1$
  without loss of generality.  If $\mathcal F_{A_2A_3}$ is flip operator swapping
  two copies of the Hilbert space $\Hil_{A_2A_3}=\Hil_{A_2} \ot \Hil_{A_3}$, we
  have
  \begin{align}
    \EE\tr[\Psi_{RA_1}^2]
    &=\EE \tr[ \Psi^{\ot 2} \mathcal F_{A_2A_3}]
      =\tr[ L^{\ot 2} \EE \chi^{\ot 2} L^{\dagger \ot 2} \mathcal F_{A_2A_3}]
      \nonumber\\
    &= \tr[ L^{\ot 2} \frac{I + \mathcal F}{d_P(d_P+1)} L^{\dagger \ot 2} \mathcal F_{A_2A_3}]
      = \frac{1}{d_P(d_P+1)}\tr[ \Pi_{RA}^{\ot 2}(\mathcal F_{RA_1}
      + \mathcal F_{A_2A_3})]
      \nonumber\\
    &=\frac{\tr(\Pi_{RA_1}^2)+\tr(\Pi_{A_2A_3}^2)}{d_P(d_P+1)}\ .
  \end{align}
  Substituting into~\eqref{eq:lb1} and applying basic inequalities lead to,
  \begin{align}
    \label{eq:lb2}
    \frac12 \EE(\epsilon^2_{\mathrm{e}}(\mathcal{N}^1\circ\mathcal{E}))
    &\leq
      \frac 12 \norm*{ \frac{ \Pi_{RA_i}}{d_P}-\tau_{RA_1} }_1
      + \frac12 \frac{\sqrt{d_Rd_1}}{d_P} \sqrt{\frac{1}{1+1/d_P}\tr (\Pi_{A_2A_3}^2)
      -\frac{1}{d_P+1}\tr (\Pi_{RA_1}^2)}
      \nonumber\\
    &\leq \frac 12 \norm*{ \frac{ \Pi_{RA_1}}{d_P}-\tau_{RA_1} }_1
      + \frac12 \frac{\sqrt{d_Rd_1}}{d_P} \sqrt{\tr (\Pi_{A_2A_3}^2)}\ ,
  \end{align}
  which is the desired formula.
\end{proof}

Next, we prove \cref{lem:smooth}.
\begin{proof}[*lem:smooth]
  For simplicity of exposition, define two probability distributions
  $p,q :\mathcal I \rightarrow \RR_{\geq 0}$,
  \begin{align}
    p_i &= \frac{c_{\mu\nu(\lambda-e_i)}}{d_P}\ ;
    &q_i &=\frac{d_{\lambda-e_i}}{d_R d_1}\ .
  \end{align}
  From \cref{lem:avg_fidelity}, we have,
  \begin{align}
    \frac12 \EE(\epsilon^2_{\mathrm{e}}(\mathcal{N}^i\circ\mathcal{E}))
    \leq \frac12 \norm*{ \frac{\Pi_{RA_i}}{d_P}-\tau_{RA_i} }_1
    + \frac12 \sqrt{d_L  d_1}\sqrt{\frac{\tr{\Pi_{A_2A_3}^2}}{d_P^2}}\ .
    \label{eq:ogydfaiudhojdklbhasiu}
  \end{align}
  We would like to bound both terms on the right hand side
  of~\eqref{eq:ogydfaiudhojdklbhasiu}.  We have
  \begin{align}
    \norm*{ \frac{ \Pi_{RA_i}}{d_P}-\tau_{RA_i} }_1
    = \sum_{i\in \mathcal I} d_{\lambda-e_i}
    \abs*{ \frac{c_{\mu\nu(\lambda-e_i)}}{d_P d_{\lambda-e_i}}-\frac{1}{d_R d_1}}
    = \sum_{i\in \mathcal I} \abs{p(i)-q(i)}\ .
  \end{align}
  Also,
  \begin{align}
    \tr`\big( \Pi_{A_2A_3}^2 )
    = \sum_{i\in \mathcal I} ( d_{\lambda-e_i} c_{\mu\nu(\lambda-e_i)})
    \frac{1}{d_{\lambda-e_i}^2}
    = \frac{d_P}{d_Rd_1} \sum_{i\in \mathcal I} \frac{p(i)}{q(i)}\ .
  \end{align}

  Now, the condition of the lemma can be written as
  \begin{align}
    1-\delta \leq \frac{p_i}{c_{\mu\nu\lambda}/d_P}\leq 1+\delta\ .
  \end{align}
  By summing over $i$, we get,
  \begin{align}
    \frac1{\abs{\mathcal I}(1+\delta)}
    \leq \frac{c_{\mu\nu\lambda}}{d_P}
    \leq \frac{1}{\abs{\mathcal I}( 1-\delta)}\ .
  \end{align}
  With some algebra, we obtain
  \begin{align}
    \abs*{ p_i - \frac{1}{\abs{\mathcal I}} }
    \leq \abs*{ p_i - \frac{c_{\mu\nu\lambda}}{d_P} }
    + \abs*{ \frac{c_{\mu\nu\lambda}}{d_P} - \frac{1}{\abs{\mathcal I}} }
    \leq \frac{\delta}{\abs{\mathcal I}(1-\delta) }
    + \frac{\delta}{\abs{\mathcal I}(1-\delta) } = \frac{4\delta}{\abs{\mathcal I} }\ .
  \end{align}
  Similarly,
  $\abs[\big]{ q_i - 1/\abs{\mathcal I} } \leq {4\delta}/{\abs{\mathcal{I}}}$.
  Therefore,
  \begin{align}
  \norm*{ \frac{ \Pi_{RA_i}}{d_P}-\tau_{RA_i} }_1
    = \sum_{i\in \mathcal I} \abs{p_i-q_i}
    \leq \sum_{i\in \mathcal I} \abs*{ p_i-\frac1{\abs{\mathcal I}}}
    +\abs*{q_i-\frac1{\abs{\mathcal I}}}
    \leq 8\delta\ .
  \end{align}
  On the other hand, $p_i\leq (1+\delta)\frac{c_{\mu\nu\lambda}}{d_P}$, and
  $1/q_i \leq \frac{d_Rd_1}{d_\lambda(1-\delta)}$. Now, we get that
  $p_i/q_i \leq (1+\delta)^2/(1-\delta)^2$. So,
  \begin{align}
    \tr`\big(\Pi_{A_2A_3}^2)
    = \frac{d_P}{d_Rd_1}
    \sum_{i\in \mathcal I}\frac{p_i}{q_i}
    \leq \frac{d_P\abs{\mathcal I}}{d_Rd_1} \left(\frac{1+\delta}{1-\delta}\right)^2
    \nonumber\\
    \leq \frac{d_P^2}{d_Rd_1 c_{\mu\nu\lambda}}\frac{(1+\delta)^2}{(1-\delta)^3}
    \leq 5^2 \frac{d_P^2}{d_Rd_1 c_{\mu\nu\lambda}}\ .
  \end{align}
  Substituting in the formula for the fidelity completes the proof.
\end{proof}
Next, we would like to prove our main theorem on random constructions,
\cref{thm:main_random}. Before that, we need to show that the
Littlewood-Richardson coefficients can grow significantly with the size the
Young diagrams. This is the content of next lemma:
\begin{lemma}
  \label{lem:LR-growth}
  In the chamber complex of Littlewood-Richardson coefficients discussed in
  \cref{sec:random-codes}, there are chambers in which $c_{\mu\nu\lambda}$ is
 a polynomial of degree $\binom{d_L-1}{2}$
\end{lemma}
\begin{proof}[*lem:LR-growth]
  Consider the following relation for the Littlewood-Richardson coefficients,
  derived by comparing dimensions:
  \begin{equation}
    \label{eq:most-trivial}
    d_\mu d_\nu=\sum_{\lambda}{c_{\mu\nu\lambda}d_\lambda}.
  \end{equation}
  Define the average of Littlewood-Richardson coefficients weighted by the
  dimension $d_\lambda$, i.e.,
  \begin{align}
    \overline c
    = \frac{\sum_{\lambda}{c_{\mu\nu\lambda}d_\lambda}}{\sum_{\lambda}{d_\lambda}}\ .
  \end{align}
  Also, assume that the number of $\lambda$'s where $c_{\mu\nu\lambda}\neq 0$ is
  $N_{\mu\nu}$ and the average dimension of $d_\lambda$, averaged over such
  $\lambda$'s is,
  \begin{align}
  \overline d
    = \frac{\sum_{\lambda \text{ where }c_{\mu\nu\lambda}\neq 0}{d_\lambda}}{N_{\mu\nu}}\ .
  \end{align}

  Now~\eqref{eq:most-trivial} becomes
  \begin{align}
    \frac{d_\mu d_\nu}{N_{\mu\nu}\overline d}=\overline c\ .
  \end{align}
  Consider the case where $\mu=N \mu_0$ and $\nu=N \nu_0$, for some fixed
  $\mu_0$ and $\nu_0$ and large $N$.  It is known that the dimension of the
  chamber complex is $3d_L-1$, see, e.g., Proposition~1 in~\cite{knutson2004honeycomb}.  Therefore, as two
  $d_L$ dimensional axis are fixed by $\mu$ and $\nu$, the section of the cone
  corresponding to $c_{\mu\nu\lambda}\neq 0$ is $d_L-1$ dimensional, and
  therefore $N_{\mu\nu}=O(N^{d_L-1})$. From the Weyl dimension formula, it is
  known that $d_\mu$, $d_\nu$ , and $\overline d$ are all
  $O\left(N^{d_L(d_L-1)/2}\right)$. So,
  \begin{align}
    \overline c= O\left((N^{(d_L-1)(d_L-2)/2}\right)\ .
  \end{align}
  This shows that there exists at least one chamber whose polynomial is at least
  degree $\binom{d_L-1}{2}$.  On the other hand, it is known that degree of the
  polynomials are bounded above by $\binom{d_L-1}{2}$ (see Corollary~4.2
  in~\cite{rassart2004polynomiality}). This completes the proof.
\end{proof}

\begin{proof}[*thm:main_random]
  We start from \cref{eq:as_1,eq:as_2}.  If we set $C_0=\max(C_0',C_0'')$ and
  $N_0=\max(N_0', N_0'')$, we have
  \begin{subequations}
    \begin{eqnarray}
      1-\frac {C_0}{N} \leq &\frac{d_{\lambda-e_i}}{d_\lambda}
      &\leq 1+\frac{C_0}{N}\ ;\\
      1-\frac {C_0}{N} \leq &~\frac{c_{\mu\nu(\lambda-e_i)}}{c_{\mu\nu\lambda}}
                              ~&\leq 1+\frac{C_0}{N}\ .
    \end{eqnarray}
  \end{subequations}
  Further, suppose that $\hat \mu$,$\hat \nu$, and $\hat \lambda$ where chosen
  such that $c_{\mu\nu\lambda}$ grows superlinearly as a function of $N$. This
  is possible for $d_L\geq 4$ as a result of \cref{lem:LR-growth}.  Using this
  fact and \cref{lem:smooth}, we get
  \begin{align}
    \EE(\epsilon^2_{\mathrm{e}}(\mathcal{N}^1\circ\mathcal{E}))
    = O( 1/N)\ .
  \end{align}
  In fact, the same relation holds for
  $\epsilon_{\mathrm{e}}(\mathcal{N}^2\circ\mathcal{E})$ and
  $\epsilon_{\mathrm{e}}(\mathcal{N}^3\circ\mathcal{E})$, and using the Markov
  inequality and the union bound we can show that there exists $\hat
  \mu$,$\hat \nu$, and $\hat \lambda$ for which
  \begin{align}
    \max`\big(\epsilon^2_{\mathrm{e}}(\mathcal{N}^1\circ\mathcal{E}),
    \epsilon^2_{\mathrm{e}}(\mathcal{N}^2\circ\mathcal{E}),
    \epsilon^2_{\mathrm{e}}(\mathcal{N}^3\circ\mathcal{E}))
    = O(1/N)\ .
  \end{align}
  As a consequence, and using \cref{lem:inv-fid-to-global-fid}, we
  get~\eqref{eq:main_thm_1}.  The second equation,~\eqref{eq:main_thm_2},
  follows from~\eqref{eq:main_thm_1} using the Weyl dimension which indicates that
 $d_i=O\left(N^{d_L(d_L-1)/2}\right)$.
\end{proof}


%% file: AppendixGeneralLemmas.tex
A first lemma relates the correctability of the code to the environment's
ability to distinguish two states in terms of the trace distance.
\begin{lemma}
  \label{lemma:aqecc-environtrdist}
  For any encoding channel $\mathcal{E}$ and noise channel $\mathcal{N}$, and
  for any two logical states $\sigma_L, \sigma'_L$, and if
  $\widehat{\mathcal{N}\circ\mathcal{E}}$ is a complementary channel of
  $\mathcal{N}\circ\mathcal{E}$, we have that
  \begin{align}
    \epsilon_{\mathrm{worst}}(\mathcal{N}\circ\mathcal{E}) \geqslant \frac12\,
    \delta`\big( \widehat{\mathcal{N}\circ\mathcal{E}}(\sigma_{L}),
    \widehat{\mathcal{N}\circ\mathcal{E}}(\sigma'_{L}) )\ .
  \end{align}
\end{lemma}

\begin{proof}[*lemma:aqecc-environtrdist]
  Let $\zeta$ be the state achieving the optimum
  in~\eqref{eq:Beny-Oreshkov-f-worst}.  We have
  \begin{align}
    \epsilon_{\mathrm{worst}}^2(\mathcal{N}\circ\mathcal{E})
    &= 1 - f_{\mathrm{worst}}^2(\mathcal{N}\circ\mathcal{E})
      \nonumber\\
    &= 1 - \min_{\phi_{L R}}
      F^2( \widehat{\mathcal{N}\circ\mathcal{E}}(\phi_{L R}),
      \mathcal{T}_\zeta(\phi_{L R}) )
      \nonumber\\
    &= \max_{\phi_{LR}}
      \bigl[ 1 - F^2( \widehat{\mathcal{N}\circ\mathcal{E}}(\phi_{L R}),
      \zeta\otimes\phi_R) \bigr]
      \nonumber\\
    &\geqslant \max_{\phi_{LR}}
      \delta`\big(\widehat{\mathcal{N}\circ\mathcal{E}}(\phi_{L R}),
      \zeta\otimes\phi_R)^2\ ,
  \end{align}
  recalling that the trace distance obeys
  $\delta(\rho,\sigma) \leqslant \sqrt{1-F^2(\rho,\sigma)}$ (see,
  e.g.,~\cite{Tomamichel2010IEEE_Duality}).  Choosing the optimization
  candidates $\sigma_L\otimes \proj0_R$ and $\sigma'_L\otimes \proj0_R$ in the
  last inequality, we obtain both
  \begin{align}
    \epsilon_{\mathrm{worst}}(\mathcal{N}\circ\mathcal{E})
    &\geqslant
      \delta`\big( \widehat{\mathcal{N}\circ\mathcal{E}}(\sigma_L), 
      \zeta )\ ;\\
    \epsilon_{\mathrm{worst}}(\mathcal{N}\circ\mathcal{E})
    &\geqslant
      \delta`\big( \widehat{\mathcal{N}\circ\mathcal{E}}(\sigma'_L), 
      \zeta )\ .
  \end{align}
  Hence, by triangle inequality,
  \begin{align}
    \delta`\big(\widehat{\mathcal{N}\circ\mathcal{E}}(\sigma_L) ,
    \widehat{\mathcal{N}\circ\mathcal{E}}(\sigma'_L) )
    \leqslant 
    \delta`\big( \widehat{\mathcal{N}\circ\mathcal{E}}(\sigma_L), \zeta) +
    \delta`\big( \zeta, \widehat{\mathcal{N}\circ\mathcal{E}}(\sigma'_L) )
    \leqslant 2 \, \epsilon_{\mathrm{worst}}(\mathcal{N}\circ\mathcal{E}) \ .
    \tag*\qedhere
  \end{align}
\end{proof}

The following lemma relates the global fidelity of the code to the fidelities
corresponding to the correction of individual errors.  Note that we do not
necessarily expect a similar result to hold for the worst-case entanglement
fidelity, because the worst-case input state might be different for each erasure
event.

\begin{lemma}
  \label{lem:inv-fid-to-global-fid}
  Let $\mathcal{N}_{A\to A}^\alpha$ and
  $\mathcal{N}_{A\to AC}(\cdot) = \sum q_\alpha
  \mathcal{N}^\alpha(\cdot)\otimes\proj{\alpha}_C$ correspond to a noise model of erasures
  at known locations, as given in~\eqref{eq:noise-map-general-alpha}.  Then, for
  any $\ket\phi_{LR}$, the average entanglement fidelity of the code with
  respect to $\ket\phi_{LR}$ is directly related to the individual fidelities of
  recovery for each possible erasure:
  \begin{align}
    f^2_{\ket{\phi}}(\mathcal{N}\circ\mathcal{E})
    = \sum q_\alpha \, f^2_{\ket{\phi}}(\mathcal{N}^\alpha\circ\mathcal{E})\ ,
  \end{align}
  and consequently,
  \begin{align}
    \epsilon^2_{\ket{\phi}}(\mathcal{N}\circ\mathcal{E})
    = \sum q_\alpha \, \epsilon^2_{\ket{\phi}}(\mathcal{N}^\alpha\circ\mathcal{E})\ .
  \end{align}
\end{lemma}
\begin{proof}[*lem:inv-fid-to-global-fid]
  The average entanglement fidelity associated with the different noise
  channels can be written as:
  \begin{subequations}
    \begin{align}
      f^2_{\ket{\phi}}(\mathcal{N}\circ\mathcal{E})
      &= \max_{\mathcal{R}} \,
        \bra\phi_{LR}`\big[
        \mathcal{R}\circ\mathcal{N}\circ\mathcal{E}(\phi_{LR}) ] \ket\phi_{LR}\ ;
      \\
      f^2_{\ket{\phi}}(\mathcal{N}^\alpha\circ\mathcal{E})
      &= \max_{\mathcal{R}^\alpha} \,
        \bra\phi_{LR}`\big[
        \mathcal{R}^\alpha\circ\mathcal{N}^\alpha\circ\mathcal{E}(\phi_{LR}) ] \ket\phi_{LR}\ ,
        \label{eq:yuporpiqfs}
    \end{align}
  \end{subequations}
  where the optimizations range over recovery channels $\mathcal{R}_{AC\to L}$
  and $\mathcal{R}_{A\to L}^\alpha$, respectively.  We have
  \begin{align}
    \hspace*{4em}
    &\hspace*{-4em}
    \max_{\mathcal{R}} \,
    \bra\phi_{LR} `\big[
    \mathcal{R}\circ\mathcal{N}\circ\mathcal{E}(\phi_{LR}) ] \ket\phi_{LR}
      \nonumber\\
    &= 
      \max_{\mathcal{R}} \,\sum q_\alpha
      \bra\phi_{LR} `\big[
      \mathcal{R}`*(\proj{\alpha}_C\otimes`\big(\mathcal{N}^\alpha\circ\mathcal{E})(\phi_{LR})) ]
      \ket\phi_{LR}
      \nonumber\\
    &\leqslant \sum q_\alpha \max_{\mathcal{R}} \,
      \bra\phi_{LR} `\big[
      \mathcal{R}`*(\proj{\alpha}_C\otimes`\big(\mathcal{N}^\alpha\circ\mathcal{E})(\phi_{LR})) ]
      \ket\phi_{LR}
      \nonumber\\
    &\leqslant \sum q_\alpha \max_{\mathcal{R}^\alpha_{A\to L}} \,
      \bra\phi_{LR} `\big[
      \mathcal{R}^\alpha`*(`\big(\mathcal{N}^\alpha\circ\mathcal{E})(\phi_{LR})) ]
      \ket\phi_{LR}\ ,
  \end{align}
  showing that
  \begin{align}
    f^2_{\ket{\phi}}(\mathcal{N}\circ\mathcal{E})
    \leqslant
    \sum q_\alpha f^2_{\ket{\phi}}(\mathcal{N}^\alpha\circ\mathcal{E})\ .
  \end{align}

  Physically, the reverse inequality follows because a global recovery strategy
  is to measure the register containing the record that indicates which error
  occurred, and to apply the optimal recovery strategy corresponding to that
  error.  Specifically, if $\mathcal{R}^\alpha_{A\to L}$ are optimal choices
  in~\eqref{eq:yuporpiqfs} for each $\alpha$, then
  we define 
  \begin{align}
    \mathcal{R}_{AC\to L}(\cdot)
    = \sum \mathcal{R}^\alpha_{A\to L}`\big(\dmatrixel{\alpha}{ (\cdot) }_C)\ .
  \end{align}
  Then,
  \begin{align}
    f^2_{\ket{\phi}}(\mathcal{N}\circ\mathcal{E})
    &\geqslant
      \bra\phi_{LR}`\big[
      \mathcal{R}\circ\mathcal{N}\circ\mathcal{E}(\phi_{LR}) ] \ket\phi_{LR}
      \nonumber\\
    &= \sum q_\alpha \bra\phi_{LR}`\big[
      \mathcal{R}`\big( \proj{\alpha}_C\otimes`\big(\mathcal{N}^{\alpha}\circ\mathcal{E})(\phi_{LR}) ) ]
      \ket\phi_{LR}
      \nonumber\\
    &= \sum q_\alpha \bra\phi_{LR}`\big[
      \mathcal{R}^\alpha`\big( `\big(\mathcal{N}^{\alpha}\circ\mathcal{E})(\phi_{LR}) ) ]
      \ket\phi_{LR}
      \nonumber\\
    &= \sum q_\alpha f^2_{\ket{\phi}}(\mathcal{N}^\alpha\circ\mathcal{E})\ ,
  \end{align}
  as claimed.
\end{proof}

The following lemma is a technical consequence of the concavity of the fidelity
function.

\begin{lemma}
  \label{lemma:fidelity-trick-op-less-than}
  Let $\rho,\sigma$ be two (normalized) quantum states. Let $\tau\geqslant 0$
  with $\rho\geqslant\tau$.  Then
  \begin{align}
    F(\rho,\sigma) \geqslant \tr(\tau)\,F`\Big(\frac{\tau}{\tr(\tau)}, \sigma)\ .
  \end{align}
\end{lemma}
\begin{proof}[*lemma:fidelity-trick-op-less-than]
  Since $\rho\geqslant \tau$, we have $\rho - \tau =: \Delta \geqslant 0$.  Then
  $\rho = \tau+ \Delta = \tr(\tau)\,\frac{\tau}{\tr(\tau)} +
  \tr(\Delta)\,\frac{\Delta}{\tr(\Delta)}$, and by concavity of the fidelity,
  \begin{align}
    F(\rho,\sigma) = F`\Big(\tr(\tau)\,\frac{\tau}{\tr(\tau)} +
    \tr(\Delta)\,\frac{\Delta}{\tr(\Delta)} , \sigma)
    \geqslant  \tr(\tau)\, F`\Big(\frac{\tau}{\tr(\tau)} , \sigma) +
    \tr(\Delta)\, F`\Big(\frac{\Delta}{\tr(\Delta)} , \sigma)\ .
  \end{align}
  The claim follows by noting that
  $\tr(\Delta)\, F`\big(\Delta/\tr(\Delta) , \sigma)\geqslant 0$.
\end{proof}
